\documentclass{article}
\usepackage{kr}

\usepackage{times}
\usepackage{soul}
\usepackage{url}
\usepackage[hidelinks]{hyperref}
\usepackage[utf8]{inputenc}
\usepackage[small]{caption}
\usepackage{graphicx}
\usepackage{amsmath}
\usepackage{amsthm}
\usepackage{booktabs}
\usepackage{algorithm}
\usepackage{multirow}

\usepackage{fontawesome}
\usepackage{microtype}
\usepackage{graphicx} % Required for inserting images
\usepackage[%textwidth=15mm,
]
{todonotes} % author comments

\usepackage{subcaption} 
\usepackage{tikz} %for graphics
\usetikzlibrary{positioning, calc, fit, shapes.geometric, arrows, decorations.pathreplacing, arrows.meta, bending}
\usepackage{url}

\usepackage{hyperref}

\usepackage{amsthm, thmtools} % theorem environments etc
\usepackage{amsfonts} % mathbb etc
\usepackage{amssymb} % symbols like \restriction
\usepackage{amsmath} % align etc
\usepackage{bm} %bold symbols
\usepackage{bbm} %mathbb on numbers
\usepackage{xspace} % spaces after commands without braces
\usepackage{mathtools} % coloneqq etc
\usepackage{algcompatible}
\usepackage{booktabs}
\usepackage{enumitem}
\usepackage{multicol}
\usepackage{float}

\theoremstyle{plain}
\newtheorem{corollary}{Corollary}
\newtheorem{lemma}{Lemma}

\theoremstyle{definition}
\newtheorem{definition}{Definition}

\theoremstyle{remark}
\newtheorem{remark}{Remark}
\newtheorem{example}{Example}

\newenvironment{proofidea}{\noindent \textit{Proof Sketch.}}{\qed}
\author{Author1}
\usepackage{pifont}
\newcommand{\xmark}{\ding{55}}
\newcommand{\cmark}{\ding{51}}

\newcommand{\N}{\ensuremath{\mathbb{N}}\xspace}
\newcommand{\R}{\ensuremath{\mathbb{R}}\xspace}
\newcommand{\bO}{\mathcal{O}}

\newcommand{\NN}{\ensuremath{\mathcal{N}}\xspace}
\newcommand{\C}{\ensuremath{\mathcal{C}}\xspace}

\newcommand{\Graph}{\ensuremath{\mathfrak{Graph}}\xspace}
\newcommand{\lgraph}{\ensuremath{\mathfrak{G}}\xspace} %labelled Graph

\newcommand{\recF}[2]{\ensuremath{\text{rec}}[#1]\text{-}#2} %recurrent circuit function class

\newcommand{\CGNN}{\text{C-GNN}\xspace}

\newcommand{\recNN}{\ensuremath{\text{rec-}\mathcal{N}}\xspace}

\newcommand{\halt}{\ensuremath{f_\text{halt}}} %halting function
\newcommand{\haltGNN}{\ensuremath{t\text{-}f_\text{halt}}} %halting function GNN (invariant)
\newcommand{\actfct}{\ensuremath{\sigma}\xspace}
\newcommand{\faci}{\mathrm{FAC}^i_\R}
\newcommand{\faco}{\ensuremath{\mathrm{FAC}^0_{\R}}\xspace}

\newcommand{\size}{\textit{size}}
\newcommand{\depth}{\textit{depth}}
\newcommand{\SD}[1]{\mathrm{FSIZE}\text{-}\mathrm{DEPTH}_{#1}}

\newcommand{\SDr}{\SD{\R}}

\newcommand{\sgre}[1]{\ensuremath{\lgraph({#1})}} %symbolic graph encoding encoding
\newcommand{\gre}[1]{\ensuremath{\lgraph({#1}, \ol{x})}} %graph encoding encoding

\newcommand{\ol}[1]{\overline{#1}}
\newcommand{\enc}[1]{\ensuremath{\langle #1 \rangle}}

\newcommand{\lms}{\ensuremath{\left\{\!\!\left\{}}
\newcommand{\rms}{\ensuremath{\right\}\!\!\right\}}}

\newcommand{\dfn}{\coloneqq}

\NewDocumentCommand{\ms}{om}{%
  \sbox0{\mathsurround=0pt$#1\{$}%
  \sbox2{\{}%
  \ifdim\ht0=\ht2
    \{\kern-.625\wd2 \{#2\}\kern-.625\wd2 \}%
  \else
    \mathopen{#1\{\kern-.7\wd0 #1\{}
    #2
    \mathclose{#1\}\kern-.7\wd0 #1\}}
  \fi
}

\DeclareUnicodeCharacter{211D}{$\mathbb{R}$}

\title{Recurrent Graph Neural Networks and Arithmetic Circuits }

\author{Anonymous Author(s)}

\author{%
Timon Barlag$^1$\and
Vivian Holzapfel$^1$\and
Laura Strieker$^{1}$\\
Jonni Virtema$^2$ \and
Heribert Vollmer$^1$\\
\affiliations
$^1$Institut für Theoretische Informatik, Leibniz Universität Hannover, Germany\\
$^2$School of Computing Science, University of Glasgow, UK\\
\emails
\{barlag, holzapfel, strieker, vollmer\}@thi.uni-hannover.de,\\
jonni.virtema@glasgow.ac.uk
}
\begin{document}
\maketitle

% \textbf{current ToDos: (all done)}
% \begin{itemize}
%     \item compare recurrent GNN defs (Grohe, Ahvonen, Pflueger)
%     \item settle on rec GNN def
%     \item come up with rec Circuit def
%     \item transfer C-GNNs to recurrent setting if needed
% \end{itemize}

% \textbf{Plan for Paper:}
% \begin{itemize}
%     \item def recurrent arithmetic circuit model
%     \item connect recurrent GNNs and recurrent circuits (probably via rec C-GNNs)
%     \item (give particular circuit family for simple rec GNNs)
% \end{itemize}

\begin{abstract}
    %\todo[inline]{finish abstract until 12.2.}
%     \vivian{old paper: We characterize the computational power of neural networks that follow the graph neural network (GNN) architecture, not restricted to aggregate-combine GNNs or other particular types.
% We establish an exact correspondence between the expressivity of GNNs using diverse activation functions and arithmetic circuits over real numbers.
% In our results the activation function of the network becomes a gate type in the circuit. 
% Our result holds for families of constant depth circuits and networks, both uniformly and non-uniformly, for all common activation functions.}

    We characterise the computational power of recurrent graph neural networks (GNNs) in terms of arithmetic circuits over the real numbers. Our networks are not restricted to aggregate-combine GNNs or other particular types.
    Generalising similar notions from the literature, we introduce the model of recurrent arithmetic circuits, which can be seen as arithmetic analogues of sequential or logical circuits.
    These circuits utilise so-called \emph{memory gates} which are used to store data between iterations of the recurrent circuit.
    While (recurrent) GNNs work on labelled graphs, we construct arithmetic circuits
that obtain encoded labelled graphs as real valued tuples and then compute the same function.
   For the other direction we construct recurrent GNNs which are
able to simulate the computations of recurrent circuits.
These GNNs are given the circuit-input as initial feature vectors and then, after the GNN-computation, have the circuit-output among the feature vectors of its nodes.
%
%of their output labelled graph when given an encoded recurrent circuit as labelled input graph as input.  
    In this way we establish an exact correspondence between the expressivity of recurrent GNNs and recurrent arithmetic circuits operating over real numbers.
    Our results both deepen our understanding of the capabilities of trained neural networks and open new approaches to study recurrent neural networks using the lens of circuit complexity theory.
\end{abstract}

\section{Introduction}

Graph neural networks (GNNs) have become a widely-used and popular machine learning model on graph inputs. Graph nodes are labelled with numerical feature values, which are updated in a sequence of steps (so called layers) depending on the feature values of the neighbours. Graph neural networks have been studied intensively from a theoretical perspective, in particular regarding their expressive (or: computational) power. In this context, the training process, i.\,e.~the question how a GNN can be obtained from its training data is ignored, but its intrinsic power as an optimally trained network is considered. Going back to a seminal paper by \citeauthor{Barcelo_2020}~\shortcite{Barcelo_2020} the expressive power of GNNs was closely related to unary formulas of a certain counting variant of first-order predicate logic that classify nodes in a given graph. Grohe \shortcite{DBLP:conf/lics/Grohe23} extended these results and related GNNs to the Boolean circuit class TC$^0$ of families of polynomial size and constant depth threshold circuits.\looseness=-1

A further step in understanding the capabilities of GNNs, that is very important for the present paper, was done by Barlag et al.{} (\citeyear{DBLP:conf/nips/BarlagHSVV24}). 
First, since GNNs intrinsically compute functions over graphs labelled with feature vectors of real numbers, they were not compared to Boolean logic (like Barceló) or Boolean circuits (like Grohe), but to circuits computing over real numbers.
%\jonni{The follwing sentence is very long.}
Second, this paper considered a general form of GNNs where the message passing from one layer to the next was not restricted to so-called AC-GNNs. These are GNNs where the computation in each node of a layer of the network consists first of an aggregation step (collecting information from the adjacent nodes of the graph—very often simply summation), followed by a combination step (combining the aggregated information with the current information of the node—mostly a linear function followed by a unary non-linear activation function such as sigmoid or ReLU). \citeauthor{DBLP:conf/nips/BarlagHSVV24} retained the message-passing architecture of GNNs, but allowed any computation that can be performed by constant-depth arithmetic circuits to be used; their model was called circuit-GNN (C-GNN). 
In this way, a close analogy between GNNs (with a constant number of layers) and arithmetic circuits of polynomial size and constant depth was obtained.

In 2024, \citeauthor{DBLP:conf/aaai/PfluegerCK24}~\shortcite{DBLP:conf/aaai/PfluegerCK24} introduced recurrent GNNs which are allowed to perform more than a fixed number of message-passing iterations. 
%\jonni{The following sentence should be revised.}
They considered several variants for termination, one by performing iterations until a form of stabilisation was achieved, and another in which the number of iterations can depend on the graph size. They then proved that both variants are more powerful than GNNs with constant number of layers, and related them to a fragment of monadic monotone fixpoint logic. Further halting conditions were studied \cite{DBLP:conf/nips/AhvonenHKL24,DBLP:journals/corr/abs-2505-11050} and related to modal and predicate logics.\looseness=-1

\noindent{\bfseries Our contributions.} 
Our goal is to exactly characterise the computational power of recurrent GNNs as a model of computation over real numbers or labelled graphs, by relating it to other well-studied computational models in theoretical computer science.
In the above mentioned works, recurrent GNNs have been compared with Boolean circuits, and their expressive power has been compared with certain logical languages. Since GNNs compute over real numbers, significant effort in these results is placed into issues of coding or approximating real numbers, and no exact correspondence between a Boolean model and GNNs has been proven (with the exception of some special cases for Boolean node classifiers).
We aim at relating recurrent GNNs to well-studied computation models operating over the reals.
More concretely, we follow the above mentioned approach by Barlag et al.{} (\citeyear{DBLP:conf/nips/BarlagHSVV24}) and study the power of recurrent GNNs by comparing them to arithmetic circuits. 

Interestingly a somewhat surprising extension of usual circuits turns out to be important:
In the area of digital systems, one distinguishes between combinational circuits, i.\,e., regular circuits with gates and connecting wires, but without any form of memory, and sequential circuits (sometimes called logic circuits) with memory gates. These consist of combinational circuits plus memory states and feedback edges that allow memory states to be used in the next iteration of the circuit computation \cite{Lewin1992}. Sequential circuits are used to implement finite automata or Moore and Mealy machines \cite{DBLP:books/daglib/0011240}.
Here, we use the well-studied model of arithmetic circuits \cite{Vollmer99} (as previously also done by Barlag et al.{} in the context of non-recurrent GNNs) but now extend it with memory gates and backward edges, and consider the computation of the circuit in iterations. We define halting conditions that for their part can be computed by arithmetic circuits depending on values of certain so-called \emph{halting gates}. We call these circuits \emph{recurrent arithmetic circuits}.
We define \emph{recurrent graph neural networks} to iterate a periodic sequence of GNN layers. Halting of the computation is determined by a halting function, depending on all feature vectors of the current layer.

Our main result is a comparison of the computational power of recurrent GNNs and recurrent arithmetic circuits. We do not restrict our attention to AC-GNNs, but as in Barlag et al.{} we allow the message passing to be computed by arithmetic circuits, which then can be recurrent themselves. Hence we distinguish inner recurrence (where the circuits computing the messages passed are recursive), outer recurrence (the recursion by iterating GNN layers), and their combination. 

We show that any recurrent C-GNN can be simulated by a recurrent arithmetic circuit. 
The circuit obtains an encoding of the input graph of the C-GNN as an input and computes an encoding of the output graph of the GNN.
Conversely, we show that any recurrent arithmetic circuit can be simulated by a recurrent C-GNN. The GNN obtains (an encoding of) the circuit's input as feature values 
%\vivianS{I think we mostly called them feature values} 
and computes the output of the circuit among the feature values
%\vivianS{do we better say values here instead of vectors? we restrict ourselves to $\R$} 
in its terminating layer. 
Both results hold for inner and outer recurrence, but as we will see, certain restrictions are necessary for syntactic reasons.\looseness=-1

By studying GNNs in comparison with another computation model over the reals, we obtain a close correspondence between recurrent GNNs and arithmetic circuits. 
In order to understand the capabilities of GNNs as Boolean classifiers, arithmetic circuits can be related with a Boolean computation model in a subsequent step.
Mixing these two steps obscures statements about the computational power of the networks by mixing up two different issues.
By separating the computational aspect of recurrent GNNs from the Boolean coding aspect, we do not only obtain an upper bound %as in the previous studies, 
but an equivalence between recurrent GNNs and recurrent arithmetic circuits over the reals. In this way we shed more light on the computational power of recurrent GNNs by showing very explicitly what elementary operations they can perform. In particular, any known limitations of recurrent arithmetic circuits known or to be proven in the future transfer immediately to recurrent GNNs.

\noindent{\bfseries Comparison to previous work.}
Recurrent GNNs were defined by \citeauthor{DBLP:conf/aaai/PfluegerCK24} \shortcite{DBLP:conf/aaai/PfluegerCK24} considering two halting conditions, as explained above. 
Subsequent papers \cite{DBLP:journals/corr/abs-2505-11050,DBLP:conf/nips/AhvonenHKL24,DBLP:journals/corr/abs-2505-00291} have considered further halting conditions and related the obtained recurrent GNN model to various logics. For some cases of Boolean node classifiers, an equivalence has been shown.\looseness=-1 

Contrary to the cited papers, we do not consider Boolean classifiers, but see both GNNs and circuits as devices that compute functions (from labelled graphs to labelled graphs, or from tuples of reals to tuples of reals), and show that both are equivalent (modulo an appropriate input encoding). 
%\vivianS{see above regarding vectors}
%
We do not restrict ourselves to the AC-GNN architecture but allow a more general message passing in the form of C-GNNs.
We do not restrict ourselves to a particular type of halting condition, but allow arbitrary halting functions computed by arithmetic circuits, subsuming all cases studied so far in the literature. 

\noindent{\bfseries Organisation.}
We start by defining our models of recurrent arithmetic circuits in Section~\ref{sec:rec_arithmetic_circs} and recurrent \CGNN{}s in Section~\ref{sec:model_of_comp}.
Subsequently, we give a short overview of our results in Section~\ref{sec:reoadmap}, after which we delve into the relation between the models in Section~\ref{sec:circs_and_gnns}, where we show how and under which restrictions recurrent arithmetic circuits can simulate recurrent \CGNN{}s and vice versa.

\section{Recurrent Arithmetic Circuits} \label{sec:rec_arithmetic_circs}
Throughout this paper, the graphs we consider have an ordered set of vertices and are undirected unless otherwise specified. 
We use overlined letters to denote tuples, and write $\lvert \overline{x} \rvert$ to denote the length of $\ol{x}$ (i.\,e. the number of elements of $\ol{x}$) and $x_j$ to denote the jth element $\overline{x}$.
% , and use $[n]$ to denote the first $n$ nonzero natural numbers $\{1, \dots, n\}$ \vivianS{check if we use that anywhere}.
For a $k \in \N$, we write $\R^k$ for the set of all $k$-tuples from $\R$ and $\R^*$ for the set of all $n$-tuples from $\R$, for $n\in \N$.
The notation $\lms \rms$ is used for multisets and tuple notation for ordered multisets.
Hence $\R^*$ is also interpreted as the set of all multisets over $\R$.
We write $\chi[f_B] \to \{0,1\}$ for the characteristic function of a Boolean expression $f_B$.

We start with defining the base model of arithmetic circuits and then add recurrence. 
For simplicity, we define our model of computation over the reals $\R$, but all our results would also hold for $\R^k$ with $k \in \N$.
Equivalence of  $\R$- and $\R^k$-circuits was shown in \cite[Remark 2.4]{DBLP:conf/nips/BarlagHSVV24}. 
\begin{definition}
    \label{def:circuit}
    Let $n,m \in \mathbb{N}$.
    An \emph{arithmetic circuit} with $n$ inputs and $m$ outputs is a simple directed acyclic graph of labelled nodes, also called \emph{gates}, such that
    \begin{itemize}[noitemsep,topsep=0pt]
        \item there are exactly $n$ input gates, which each have indegree $0$,\looseness=-1
        \item there are exactly $m$ output gates, which have indegree $1$ and outdegree $0$,
        \item there are constant gates, which have indegree $0$ and are labelled with a value $c \in \R$,
        % \item gates labelled \emph{sign} have indegree $1$ and outdegree $1$, \timonS{do we need sign gates for anything?}
        %\item there are projection gates labelled \emph{projection$_{i, j}$} for $1 \leq i,j \leq k$, which have in- and outdegree $1$,
        \item there are arithmetic gates, which are associated with addition or multiplication.
    \end{itemize}
    % \timon{input gates labelled input etc is a bit weird}
    % \jonni{Inputs could we those that have in-degree 0, and are not constants (nor auxilliary memory gates). Output gates could be the ones without outgoing edges. This could create a issue with non-output halting gates, when "designing" circuits though.}
    Both the input and the output gates are ordered.

    The \emph{depth} of $C$ (written $\depth(C)$) is the length of the longest path from an input gate to an output gate in $C$ and the \emph{size} of $C$ (written $\size(C)$) is the number of gates in $C$.
%    We will also write $\depth(C)$ and $\size(C)$ to denote these values.
%
    For a gate $g$ in $C$, we will write $\depth(g)$ to denote the length of the longest path from an input gate to $g$ in $C$.
\end{definition}

% \heribert{In Def.~\ref{def:labelled_graph}, we define labelled graph, but here we (additionally) use different labels.}
% \jonni{We could use the term "associate" and "associated" instead of label. }
\begin{definition}
\label{def:circuit_function}
%    Let $n, m \in \N$ and let $C$ be an arithmetic circuit with $n$ inputs and $m$ outputs.
    % Each gate in an arithmetic circuit $C$ with $n$ inputs and $m$ outputs computes the function $f^C_g \colon (\R)^n \to (\R)$ as follows:

    % \begin{itemize}
    %     \item If $g$ is the $i$th input gate, then $f_g(\ol{x}) = x_i$.
    %     \item If $g$ is a constant gate labelled $c$, then $f_g(\ol{x}) = c$.
    %     \item If $g$ is a projection gate labelled \textit{projection$_{i, j}$}, then $f_g(\ol{x})$
    % \end{itemize}
    An arithmetic circuit $C$ with $n$ inputs and $m$ outputs computes the function $f_C \colon \R^n \to \R^m$ as follows:
    First, the input to the circuit is placed in the input gates.
    Then, recursively, each arithmetic gate whose predecessor gates all have a value, computes its own value by applying the function it is labelled with to the values of its predecessors.
    %By addition and multiplication we refer to the respective operations.
    %and the projection computes the function $\proj[i,j]{} \colon \R \rightarrow \R,\: (x_1, \dots, x_i, \dots x_k) \mapsto (0, \dots, 0, \underset{\text{position }j}{x_i}, 0, \dots, 0)$.
    Finally, each output gate takes the value of its predecessor, once its predecessor has one. 
    The output of $f_C$ is then the tuple of values in the $m$ output gates of $C$ after the computation.
    % At times it is convenient to be able to talk about the function computed up to an individual gate of a circuit.
    % For a gate $g$ in $C$, we thus define $f_g^C \colon \R^n \to \R$, where $f_g^C(\ol{x})$ is the value that $g$ takes in the computation of $C$ on input $\ol{x}$.
    The function computed up to an individual gate $g$ of a circuit $C$ is defined as $f_g^C \colon \R^n \to \R$, where $f_g^C(\ol{x})$ is the value that $g$ takes in the computation of $C$ on input $\ol{x}$.
\end{definition}

% \begin{definition}
%     An \emph{arithmetic circuit family} $\C$ is a sequence $(C_{n})_{n \in \N}$ of circuits, where each circuit $C_{n}$ has exactly $n$ input gates. 
%     Its $\depth$ and $\size$ are functions mapping natural number $n$ to $\depth(C_n)$ and $\size(C_n)$, respectively.

%     An arithmetic circuit family $\C = (C_{n})_{n \in \N}$ computes the function $f_\C \colon \R^* \to \R^*$ defined as 
%     \(
%         f_\C(\ol{x}) \coloneqq f_{C_{\lvert \ol{x} \rvert}}(\ol{x}).
%     \)
%     \label{def:circuit_family}
% \end{definition}

% \begin{definition}
%     \label{def:fnc_class}
%     For any $s, d \colon \N \to \N$, the class $\SDk(s, d)$ is the class of all functions $\R^* \to \R^*$ that are computable by arithmetic circuit families of size  $\bO(s)$ and depth $\bO(d)$.
%     Let $\mathcal{A}$ be a set of functions of the form $f \colon \R \to \R$. 
%     Then $\SDk(s, d)[\mathcal{A}]$ is the class of functions computable by arithmetic circuit families with the same constraints, but with additional gate types $g_f$, with indegree and outdegree $1$ that compute $f$, for each $f\in\mathcal{A}$.  
% \end{definition}
%     We call all classes of the form $\SDk(s,d)[\mathcal{A}]$ (and their subclasses) \emph{circuit function classes}.
% As usual in circuit complexity, an \emph{$\mathfrak{F}$-circuit family}, where $\mathfrak{F}$ is a class of functions, will denote a circuit family that computes a function in $\mathfrak{F}$.
    
Adding recurrence extends the definition of an arithmetic circuit by a halting function and the notion of memory gates. 
It allows for a circuit to be iteratively executed multiple times until a halting condition is met. 
In order to save information between the iterations of computation, we first extend our model of arithmetic circuits with so called \emph{auxiliary memory gates} and call the resulting model an \emph{extended arithmetic circuit}.

% \jonni{The reader will probably think that what weird things these auxiliary memory gates are. Also, if the circuit really encodes the first-value of the memory gates. Then the circuit and circuit family computes a total function (without $\ol{a}$), and can be seen to compute the partial function.}
\begin{definition} \label{def:ext_circ}
    Let $\ell ,m,n \in \mathbb{N}$.
    An \emph{extended arithmetic circuit} with $n$ inputs and $m$ outputs is an arithmetic circuit with $\ell$ additional gates labelled \emph{auxiliary memory gates}
    %which are labelled with real values $\ol{a}\in \R^\ell$.
    which are ordered and behave in the same way as input gates.
    %apart from being labelled with their initial values.
    The term \emph{memory gates} is used to refer to the set of input and auxiliary memory gates.%\looseness=-1

    An extended arithmetic circuit $C$ 
    %with $n$ inputs, $\ell$ auxiliary memory gates and $m$ outputs 
    extends the computation of an arithmetic circuit by dependence on the auxiliary memory gates and computes the function $f_C \colon \R^n\times \R^\ell \to \R^m$.
    The functions computed up to individual gates are extended in the same way such that for each gate $g$ in $C$, we define $f_g^C \colon \R^n \times \R^\ell \to \R$, where $f_g^C(\ol{x}, \ol{a})$ is the value that $g$ takes in the computation of $C$ on input $\ol{x}$ and with auxiliary memory gate values $\ol{a}$.
\end{definition}
%The size and depth of extended arithmetic circuits are defined as in Definition~\ref{def:circuit}, where the auxiliary memory gates are counted like the other non-input gates.
%\jonni{Memory gates should be treated as input gates in the size depth calculations. I think depth has to be defined as a maximum distance from one gate to another, otherwise some weird things might happen.}
The size of extended arithmetic circuits is defined as in Definition~\ref{def:circuit}, whereas the depth is the maximum distance between any two gates.

\begin{definition} \label{def:extended_circuit_family}    
 An \emph{extended arithmetic circuit family} $\C$ is a sequence $(C_{n})_{n \in \N}$ of extended circuits, where for all $n\in \N$ the circuit $C_{n}$ has exactly $n$ input gates. 
    Its $\depth$ and $\size$ are functions mapping natural number $n$ to $\depth(C_n)$ and $\size(C_n)$, respectively.

    % An extended arithmetic circuit family $\C = (C_{n})_{n \in \N}$ computes the partial function $f_\C \colon \R^*\times \R^*\to \R^*$ defined as 
    % \[
    %     f_\C(\ol{x}, \ol{a}) \coloneqq \begin{cases}
    %         f_{C_{\lvert \ol{x} \rvert}}(\ol{x}, \ol{a}), &\textnormal{if $C_{\lvert \ol{x} \rvert}$ has exactly } \lvert \ol{a} \rvert\\ &\text{auxiliary memory gates} \\
    %         \textnormal{undefined,} & \textnormal{otherwise.}
    %     \end{cases}
    % \] 
    An extended arithmetic circuit family $\C = (C_{n})_{n \in \N}$ computes the function $f_\C \colon \R^*\times \R^*\to \R^*$ defined as 
        $f_\C(\ol{x}, \ol{a}) \coloneqq
            f_{C_{\lvert \ol{x} \rvert}}(\ol{x}, \ol{a})$.
    If $C_{\lvert \ol{x} \rvert}$ has more than $\lvert \ol{a} \rvert$ auxiliary memory gates, the input to $f_{C_{\lvert \ol{x} \rvert}}$ is padded with zeroes, and
    if it has fewer, the overflow %rest of the input into the auxiliary memory gates 
    is cut off.
    % \(
    %     f_\C(\ol{x}, \ol{a}) \coloneqq f_{C_{\lvert \ol{x} \rvert}}(\ol{x}, \ol{a})
    % \) 
    % if $C_{\lvert \ol{x} \rvert}$ has exactly $\lvert \ol{a} \rvert$ auxiliary memory gates and is otherwise undefined.
    The number of auxiliary memory gates need not correlate with the input size and can change for every circuit in the family.
\end{definition}

%\jonni{The functions below are partial functions, but do we want $\faci$ to denote partial function of type $f \colon \R^* \times \R^* \to \R^*$ or functions of type $f \colon \R^* \to \R^*$?}

\begin{definition}
    \label{def:fnc_class}
    For $s, d \colon \N \to \N$, $\SDr(s, d)$ is the class of all functions $\R^* \times \R^*\to \R^*$ that are computable by extended arithmetic circuit families of size  $\bO(s)$ and depth $\bO(d)$.
    Let $\mathcal{A}$ be a set of functions $f \colon \R \to \R$. 
    Then $\SDr(s, d)[\mathcal{A}]$ is the class of functions computable by extended arithmetic circuit families with the same constraints, but with additional gate types $g_f$, with indegree and outdegree $1$ that compute $f$, for each $f\in\mathcal{A}$.  
\end{definition}

We are interested in functions that are ``simple'' in the sense that they can be computed by circuit families with reasonable size and depth.

\begin{definition}
   % \timon{Redefined circuit function classes to just be $\faci[\mathcal{A}]$}
    \label{def:fac_class}
    For $i \in \N$, the class $\faci$ denotes the set of functions $f \colon \R^* \times \R^* \to \R^*$ computable by extended arithmetic circuit families of size $\bO(n^{\bO(1)})$ and depth $\bO((\log n)^i)$.

    % Let $\mathcal{A}$ be a set of functions of the form $f \colon \R \to \R$.
    % Then $\faci[\mathcal{A}]$ is the class of functions computable by extended arithmetic circuit families with the same constraints, but with additional gate types $g_f$ with indegree and outdegree $1$ that compute $f$ for each $f \in \mathcal{A}$.
\end{definition}
For a set $\mathcal{A}$, the class $\faci[\mathcal{A}]$ is defined as in Definition \ref{def:fnc_class}.
We call such classes, i.\,e., classes of the form $\faci[\mathcal{A}]$, \emph{circuit function classes}.
Since those are the only classes we are interested in, whenever we write $\mathfrak{F}$ in the remainder of the paper, we mean that $\mathfrak{F}$ is a circuit function class.

Furthermore, for any $\mathfrak{F} = \faci[\mathcal{A}]$ we will use the term \emph{$\mathfrak{F}$-circuit family} to denote those circuit families adhering to the size and depth requirements of $\mathfrak{F}$.
\begin{remark}
    % The same terms of circuit families and circuit function classes are defined for classical arithmetic circuits as described in Definition \ref{def:circuit}. 
    Circuit families and circuit function classes are defined analogously for (non-extended) arithmetic circuits as described in Definition~\ref{def:circuit}.
\end{remark}

The second component of a recurrent arithmetic circuit is the halting condition.
We define the halting condition as a function based on a natural number which will represent the number of iterations of the circuit, and a tuple of reals which will represent values of some predefined gates of the circuit.
Note that this functionality could also be simulated by a function without the additional parameter for the number of iterations by using a sub circuit to count this instead, but we keep it to be consistent with the model of recurrent GNNs we introduce later on. 
This equivalence is shown in Lemma~\ref{lem:hlt_fnc_equivalence_extra_parameter} in the appendix.

\begin{definition} \label{def:halting_function}
    A \emph{halting function} is a function of the form $\halt \colon \N_{>0} \times \R^* \to \{0,1\}$.
\end{definition}
% The function need not depend on the natural number.
% \vivian{maybe add something that a recurrent circuit with a function that is only dependent on the natural number could just be simulated by one where the circuit is composed multiple times}
% \begin{example}
%     Let $\halt \colon \R^* \to \{0,1\}$, \[\halt (x, M) \coloneqq (x \geq 5) \cdot \left(\sum_{\ol{m} \in M} \ol{m} < 20\right).\]
%     This halting function evaluates to one if and only if the circuit it is based on is iterated at least 5 times and the sum over all values isn't greater than 20. 
%     \vivian{it would be nice to have an example of a halting function here already, but not sure how smart it is without the definition of a recurrent circuit}
% \end{example}

Note that the halting function maps to discrete values. In what follows, we will also consider halting functions that can be computed by circuits, i.\,e.~they are of some circuit function class. 
Since plain arithmetic circuits can only compute continuous functions, we have to make use of the additional gate types to achieve the non-continuous nature of the halting function. 
For this we often use the sign function as an additional gate type:\looseness=-1
\[ 
\text{sign}(x) \dfn 
\begin{cases}
    1, \text{ if }x > 0\\
    0, \text{ otherwise.}
\end{cases}
\]
Given a circuit function class $\mathfrak{F}$ we write $\mathfrak{F}_s$ %\vivianS{we need to decide how we want to do it for actual classes, like $\faco$, ${\faco}_s$ does not look good} 
as a shorthand for $\mathfrak{F}[\text{sign}]$.
%Throughout the paper we assume that the halting function is never the constant zero function, since we want our models to be able to halt eventually.% \vivianS{consider functions that are constant 0 but not the zero function in the results}
%\vivianS{probably don't need this if we make the halting only dependent on a fixed set of gates of the circuit}
%\timonS{Do we need to exclude the case that the halting function is constant $0$}

% \jonni{Do we want to restrict the outdegree of output gates (also the target should be a memory gate)? One possibility is to say that memory gates are particular constant gates that  have in-degree 1. Probably the circuit below should be defined more directly, now the definition is hard to find.}
% \vivian{decide if we want to have output gates of outdegree 1 that connect to memory gates i.\,e. every memory gate has an output gate as its predecessor. Restricts the output dimension of the circuit. Same could probably be simulated by a circuit without this restriction.}
% \vivian{or just have the output gates as in a normal arithmetic circuit with outdegree 0 and only have edges to memory gates only from non-output gates}

% To distinguish recurrent arithmetic circuits from plain ones, we often use the prefix $\textit{rec-}$ when referring to recurrent circuits.

\begin{definition} \label{def:rec_circ}   
    Let $\ell, m,n \in \mathbb{N}$.
    A \emph{recurrent arithmetic circuit} with $n$ inputs, $\ell$ auxiliary memory gates and $m$ outputs is a tuple $(C, \ol{a}, E_\text{rec}, \halt, V_\text{halt})$, where $C$ is an extended arithmetic circuit with $n$ inputs, $m$ outputs and $\ell$ auxiliary memory gates with a tuple of initial values $\ol{a} \in \R^\ell$. 
    We call $C$ the \emph{underlying extended arithmetic circuit}.
    $E_\text{rec}$ is a set of edges consisting of exactly one incoming edge for each memory gate in $C$.
    As before, the memory gates consist of the input and auxiliary memory gates.
    Every memory gate has fan-in one, the output gates have fan-out zero.

    % \timon{Should the below maybe be put to the definition of the computed function definition?}
    % \vivian{rewrote it, check whether this is good}
    % \timon{looks good to me}
    
    % The rounds of execution of the circuit are called \emph{iterations}.
    % The halting of the recurrent circuit is determined by a halting function $\halt \colon \N_{>0} \times \R^{\vert V_\text{halt}\rvert} \to \{0,1\}$ defined on the current number of iteration and the values of the ordered \emph{halting gates} $V_\text{halt}$, a subset of the gates of $C$.
    The halting of the recurrent circuit is determined by a halting function $\halt \colon \N_{>0} \times \R^{\vert V_\text{halt}\rvert} \to \{0,1\}$, where $V_\text{halt}$ are the \emph{halting gates}, a subset of the gates of $C$.
\end{definition}
The structure of a recurrent circuit is depicted in Figure~\ref{fig:rec_circ_structure} and a complete example is given in Example~\ref{ex:fib_numbers}.

\begin{figure}[tb]
    \centering
    \scalebox{.85}{
            \begin{tikzpicture}
                [every node/.style={minimum width =.9cm, minimum height= .5cm, outer sep=0pt}, every path/.style={-stealth}
                    ]
                    \node     (i1)        {$\textit{in}_1$};
                    \node     (dots1)     [right =0.1 of i1]      {$\dots$};
                    \node     (in)        [right =0.1 of dots1]     {$\textit{in}_n$};
                    \node     (m1)        [right =0.5 of in]           {$\textit{aux}_1$};
                    \node     (dots2)     [right =0.1 of m1]      {$\dots$};
                    \node     (ml)        [right =0.1 of dots2]     {$\textit{aux}_\ell$};
                    \node     (o1)        [below = 3.1cm of in.south west]           {$\textit{out}_1$};
                   
                    \node     (om)        [below = 3.1cm of m1.south east]     {$\textit{out}_m$};
                     \node     (dots3)     at    ($(o1)!.5!(om)$)      {$\dots$};
                    \node[fill=gray!25, rectangle, minimum width =6cm, minimum height=2cm, draw] (rectangle) [below = 0.8 of $(in)!0.5!(m1)$] {
                    };
                    \node[ellipse, fill=white, draw, minimum height=1.2cm] (vhalt) [below = 1.2cm of $(in)!0.5!(m1)$] {
                    $V_\text{halt}$};
                    \node[inner sep =8pt, rectangle, draw, fit=(i1) (ml) (om) (rectangle), label={below: \textit{underlying extended arithmetic circuit} $C$}] {};
        
                    \draw[-stealth]     (i1)      -- ++(-0,-0.8cm)      (rectangle);
                    \draw[-stealth]     (in)      -- ++(-0,-.8cm)      (rectangle);
                    \draw[-stealth]     (m1)      -- ++(-0,-.8cm)      (rectangle);
                    \draw[-stealth]     (ml)      -- ++(-0,-.8cm)      (rectangle);
                    \draw[stealth-]     (o1)-- ++(-0,.8cm) (rectangle) ;
                    \draw[stealth-]      (om)  -- ++(-0,.8cm)  (rectangle);
                    \draw[dashed]   ($(rectangle.east)+(0,-0.8)$) |- ++(1,0)  |- ($(m1)+(0,1)$) -- (m1.north);
                    \draw[dashed]   ($(rectangle.east)+(0,-0.5)$) |- ++(0.8,0)  |- ($(ml)+(0,0.8)$) -- (ml.north);
                    \draw[dashed]   ($(rectangle.west)+(0,-0.8)$) |- ++(-1,0)  |- ($(in)+(0,1)$) -- (in.north);
                    \draw[dashed]   ($(rectangle.west)+(0,-0.5)$) |- ++(-0.8,0)  |- ($(i1)+(0,0.8)$) -- (i1.north);
            \end{tikzpicture}}
    \caption{Structure of a recurrent arithmetic circuit: $\text{in}_i$ 
    are the input gates, $\text{aux}_i$ are the auxiliary memory gates, $\text{out}_i$ are the output gates and $V_\text{halt}$ the halting gates. Dashed edges are the recurrent edges of the set $E_\text{rec}$. 
    The halting function is not depicted.}
    \label{fig:rec_circ_structure}
\end{figure}
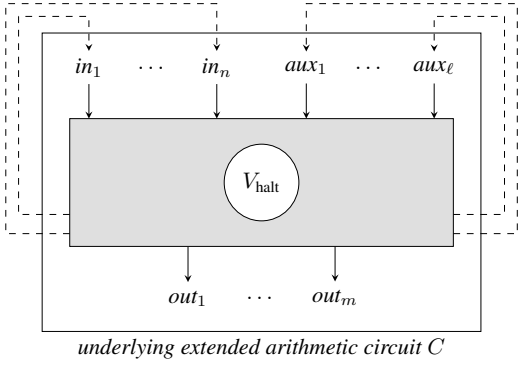

% A recurrent arithmetic circuit $\textit{rec-}C =(C, \ol{a}, E_\text{rec}, \halt, V_\text{halt})$ with $n$ inputs, $\ell$ auxiliary memory gates and $m$ outputs, where $\ell, m,n \in \mathbb{N}$, computes a function $f_{\textit{rec-}C} \colon \R^n \to \R^m$ as follows:
%     In addition to the input gates auxiliary memory gates get assigned their initial constant.
%     The underlying extended arithmetic circuit is then executed with the input values and the initial memory gate values as in Definition \ref{def:ext_circ} until all gates have a value. 
%     If $\halt(i, \text{val}(V_\text{halt}))=1$, where $i \in \N$ is the number of the current iteration and $\text{val}(V_\text{halt})$ is the tuple of values in the predefined ordered subset $V_\text{halt}$ of the gates,
%     the output of $f_{\text{rec-}C}$ is the tuple of values in the $m$ output gates.
%     Otherwise the memory gates take the value of their respective predecessors as defined by $E_\text{rec}$, all other non-constant gates are reset to have no value and the computation of the underlying circuit is executed again.

    A recurrent arithmetic circuit computes a function as follows:
    In addition to the input gates, the auxiliary memory gates get assigned their initial constant.
    The underlying extended arithmetic circuit is then executed with the input values and the initial memory gate values as in Definition~\ref{def:ext_circ} until all gates have a value. 
    If the halting function evaluates to $1$ the output of the function is the tuple of values in the output gates of the recurrent circuit.
    Otherwise the memory gates take the value of their respective predecessors, all other non-constant gates are reset to have no value and the computation of the underlying circuit is executed again.
    %This is defined formally as follows.

\begin{definition}
    Let $\ell, m,n,p \in \mathbb{N}$.
    Let $\textit{rec-}C =(C, \ol{a}, E_\text{rec}, \halt, V_\text{halt})$ be a recurrent arithmetic circuit with $n$ inputs, $\ell$ auxiliary memory gates and $m$ outputs, an ordered set of $p$ halting gates $V_\text{halt} = \{v_1, \dots,v_p\}$ and a halting function $\halt \colon \N_{>0} \times \R^p \to \{0,1\}$.
    The halting function operates on the current iteration number and the ordered values of the gates $V_\text{halt}$.
    For all gates $g$ in $C$ let $f_g \colon\R^n \times \R^\ell \to \R$ be the function that computes the value that $g$ takes in the computation of $C$.
    The recurrent arithmetic circuit $\textit{rec-}C$ computes the function $f_{\textit{rec-}C} \colon \R^n \to \R^m$ defined as $f_{\textit{rec-}C}(\ol{x}) \coloneqq f_\text{mem}^1(\ol{x}, \ol{a})$,
    where $\ol{a}$ are the initial constants in the auxiliary memory gates and $\ol{x}$ is the input. 
    The function $f_\text{mem}^i \colon\R^n \times \R^\ell \to \R^m$ is recursively defined as 
    % \[
    %     f_\text{mem}(\ol{x}, \ol{a}) \coloneqq \begin{cases}
    %     \begin{aligned}
    %         f_\text{mem}\big(&f^C_{\text{pi}_1}( \ol{x}, \ol{a}), \dots, f^C_{\text{pi}_n}( \ol{x}, \ol{a}), \\
    %         &f^C_{\text{pa}_1}( \ol{x}, \ol{a}), \dots, f^C_{\text{pa}_\ell} ( \ol{x}, \ol{a})\big),
    %     \end{aligned}
    %          & \begin{aligned}
    %              f_\text{halt}\Big(&i, \big(f^C_{v_1}(\ol{x}, \ol{a}), \dots,\\ 
    %              &f^C_{v_p}(\ol{x}, \ol{a})\big)\Big)=0
    %          \end{aligned}\\
    %        \left(f^C_{\text{out}_1}(\ol{x}, \ol{a}\right) \dots f^C_{\text{out}_m}(\ol{x}, \ol{a})),  & \text{else}
    %     \end{cases}
    % \] \vivianS{maybe use some shorthand vector notation instead?}
    % \jonni{Something like $f_\text{mem}^{i+1}\big(\vec{y}, \vec{b}\big)$, where $y_j=f^C_{\text{pr}(\text{in}_j)}( \ol{x}, \ol{a})$ and  $b_j=f^C_{\text{pr}(\text{aux}_j)}( \ol{x}, \ol{a})$ and $\text{pr}(\text{aux}_j)$ is the predecessor gate of $\text{aux}_j$.}
    % \vivian{alternative:
    % \[
    %     f_\text{mem}^i(\ol{x}, \ol{a}) \coloneqq \begin{cases}
    %         f_\text{mem}^{i+1}\big((f^C_{\text{pr}(\text{in}_j)}( \ol{x}, \ol{a}))_{j \in [n]}, (f^C_{\text{pr}(\text{aux}_j)}( \ol{x}, \ol{a}))_{j \in [\ell]})\big),
    %              &f_\text{halt}(i, (f^C_{v_j}(\ol{x}, \ol{a})_{j\in[p]})=0\\
    %        \left(f^C_{\text{out}_j}(\ol{x}, \ol{a})\right)_{j\in [m]},  & \text{else}
    %     \end{cases}
    % \]}
    \[
        f_\text{mem}^i(\ol{x}, \ol{a}) \coloneqq \begin{cases}
            f_\text{mem}^{i+1}\left(\ol{y}, \ol{b}\right),   &\halt(i, \ol{v})=0\\
            \ol{o},                      &\text{else}
        \end{cases}
    ,\]
    where $\ol{y} \in \R^n $ with $y_j=f^C_{\text{in}_j}( \ol{x}, \ol{a})$ for each $j\leq n$, $\ol{b}\in \R^\ell$ with $b_j=f^C_{\text{aux}_j}( \ol{x}, \ol{a})$ for each $j\leq \ell$, $\ol{v} \in \R^p$ with $v_j = f^C_{v_j}(\ol{x}, \ol{a})$ for each $j\leq p$, and $\ol{o} \in \R^m$ with $o_j = f^C_{\text{out}_j}(\ol{x}, \ol{a})$ for each $j\leq m$.
    % $\ol{y} \in \R^n $ with $y_j=f^C_{\text{pr}(\text{in}_j)}( \ol{x}, \ol{a})$, $\ol{b}\in \R^\ell$ with $b_j=f^C_{\text{pr}(\text{aux}_j)}( \ol{x}, \ol{a})$ and $\ol{o} \in \R^m$ with $o_j = f^C_{\text{out}_j}(\ol{x}, \ol{a})$ and $\ol{v} \in \R^p$ with $v_j = f^C_{v_j}(\ol{x}, \ol{a})$.
    Here, $\text{in}_j$ ($\text{aux}_j$, $\text{out}_j$, $\text{out}_j$, resp.) refers to the predecessor of the $j$th input gate (auxiliary memory gate, halting gate, output gate, resp.) and $i$ is the number of the current iteration. 
    The tuples $\ol{y}, \ol{b}, \ol{o}, \ol{v}$ are the new values of these gates after one iteration of the underlying circuit $C$.
    The function $f_\text{mem}^i$ maps memory gate values to memory gate values until the halting condition is reached.
\end{definition}

\begin{example} \label{ex:fib_numbers}
For $x_1 \in \N$, 
the recurrent circuit $\textit{rec-}C = (C, \ol{a}, E_\text{rec}, \halt, V_\text{halt})$ as illustrated in Figure~\ref{fig:examplereccirc} computes the $x_1$-th Fibonacci number when given $x_1 > 1$ as the input.

The recurrent circuit $\textit{rec-}C$ has one input gate $\textit{in}_1$ which receives the input $x_1$, two auxiliary memory gates $\textit{aux}_1$ and $\textit{aux}_2$ and one output gate $\textit{out}_1$. 
%Its recurrent edges $E_\text{rec}$ are given by the dashed lines and 
The set $V_\text{halt}$ consists only of $\textit{in}_1.$% which is marked in bold.
The halting function $f_\mathrm{halt}$ of $\textit{rec-}C$ is the function computed by the circuit $C_\mathrm{halt}$ on the right.

The circuit computes the $x_1$-th Fibonacci number as follows.
Initially, the input $x_1$ to $\textit{rec-}C$ is the index of the desired Fibonacci number and the auxiliary memory gates are set up with the initial values $\ol{a} = (a_1, a_2)$ where $a_1=1$ and $a_2=0$, respectively, which are the first (and 0th) Fibonacci numbers. 
The addition gate now computes the sum of the values of the gates $\textit{aux}_1$ and $\textit{aux}_2$.
Before the next iteration of the circuit begins, the halting condition is checked.
In this case, $C_\text{halt}$ examines whether the number of iteration rounds (the value of input gate $i$) and the input to $\textit{rec-}C$, $x_1$, are equal.
%\jonni{Something weird with the explanation here.}
But since the first Fibonacci number was stored in $\textit{aux}_1$ and not calculated by the recursive scheme, we need to subtract $1$ from $x_1$.
The function of the equality gate ($=$) can be computed using a (non-recurrent) circuit of depth 7 facilitating the following formula:
    $x_1 = x_2 \coloneqq ((\text{sign}(x_1-x_2) + \text{sign}(x_2-x_1)) \times -1) +1$.
    %where each subtraction is computed as $x_1 - x_2 \coloneqq x_1 + (x_2 \times (-1))$. 
If $C_\text{halt}$ outputs 1, $\textit{rec-}C$ outputs the value given to the output gate.
Otherwise, the values of the auxiliary memory gates are updated:
The value of $\textit{aux}_1$ gets updated with the sum of the previous values of $\textit{aux}_1$ and $\textit{aux}_2$, while the value of $\textit{aux}_2$ gets updated with the previous value of $\textit{aux}_1$.\looseness=-1
%\laura{todo: add description of how to compute equality (=) gate}
\setlength\columnsep{1cm}
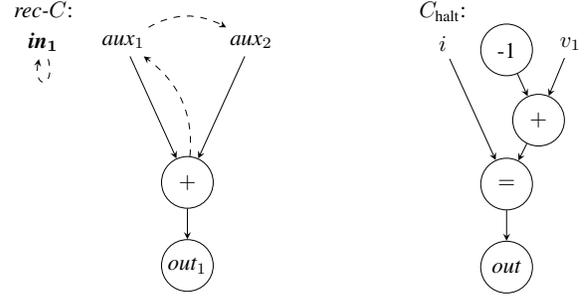
\begin{figure}[tb]
    \centering
    \vspace{-.3cm}
    \begin{multicols}{2}
    \scalebox{.85}{\begin{tikzpicture}
        every node/.style={minimum size =1cm, outer sep=0pt}, every path/.style={-stealth, ->}
        \tikzset{basic/.style={draw,fill=none,
                       text badly centered,minimum width=2em}}
        \tikzset{gate/.style={basic,circle,minimum width=2.3em, inner sep=0pt}}
        \tikzset{back/.style={basic, %color=blue,
        ->, dashed}}

        \node[] (name) at (0.75,.5) {$\textit{rec-}C$:};
        
        %input and aux memory gates
        \node[] (x1)    at (0.75,0)    {$\mathbf{\textbf{\textit{in}}_1}$};
        \node[] (m1)    at (2,0)    {$\textit{aux}_1$};
        \node[] (m2)    at (4,0)    {$\textit{aux}_2$};
        %\node[gate] (m3)    at (6,0)    {$m_3$};

        %arithmetic gates
        \node[gate] (add1)  at (3,-2.2) {$+$};
        %\node[gate] (mul1)  at (4.5,-2) {$\times$};
        %\node[gate] (add2)  at (3.75,-4){$+$};
        \node[gate] (out)   at (3,-3.5) {$\textit{out}_1$};
        %\node[gate, color=red] (add3)  at (6,-2)   {$+$};

        %constant gates
        %\node[gate] (c-1)   at (5,-1)   {-1};
        %\node[gate] (c1)    at (7, -1)  {1};

        %wires
        \draw[-stealth] (m1) to (add1) {};
        \draw[-stealth] (m2) to (add1) {};
        %\draw[->] (m2) to (mul1) {};
        %\draw[->] (c-1) to (mul1) {};
        \draw[-stealth] (add1) to (out) {};
        %\draw[->] (mul1) to (add2) {};
        %\draw[->] (m3) to (add3) {};
        %\draw[->] (c1) to (add3) {};
        %\draw[->] (add1) to (add2) {};

        %loops back
        \draw[bend right, back, -stealth] (add1) to (m1) {};
        %\draw[back] (add2) to (m2) {};
        \draw[back, bend left, -stealth] (m1) to (m2) {};
        \draw[back, loop below, -stealth] (x1) to (x1) {};

    \end{tikzpicture}}
    \scalebox{.85}{
    \begin{tikzpicture}
        every node/.style={minimum size =1cm, outer sep=0pt}, every path/.style={-stealth, ->}
        \tikzset{basic/.style={draw,fill=none,
                       text badly centered,minimum width=2.3em}}
        \tikzset{gate/.style={basic,circle,minimum width=2.3em}}
        \tikzset{back/.style={basic, color=blue, ->, dashed}}

        \node[] (name) at (0,.5) {$C_{\text{halt}}$:};

        %input and aux memory gates
        % \node[color=red] (x1)    at (0,0)    {$i$};
        % \node[color=red] (x2)    at (2,0)    {$v_1$};
        \node[] (x1)    at (0,0)    {$i$};
        \node[] (x2)    at (2,0)    {$v_1$};

        %arithmetic gates
        % \node[gate] (mul1)  at (1,-2) {$\times$};
        % \node[gate] (mul2)  at (3,-2) {$\times$};
        % \node[gate] (add1)  at (1,-3){$+$};
        % \node[gate] (add2)  at (3,-3) {$+$};
        % \node[gate] (add3)  at (2,-5)   {$+$};
        % \node[gate] (sig1)  at (1,-4)   {s};
        % \node[gate] (sig2)  at (3,-4)   {s};
        \node[gate] (out)  at (1,-3.5)   {\textit{out}};
        \node[gate] (add4) at (1.5,-1.2)    {$+$};
        \node[gate] (eq1) at (1,-2.2) {$=$};

        %constant gates
        % \node[gate] (c1)   at (1,-1)   {-1};
        % \node[gate] (c2)   at (3, -1)  {-1};
        % \node[gate] (c3)  at (2,-4)   {-1};
        \node[gate] (c4)  at (1,-0.1) {-1};

        %wires
        % \draw[->, bend right] (x1) to (add1) {};
        \draw[-stealth] (x2) to (add4) {};
        % \draw[->] (add4) to (mul2) {};
        % \draw[->] (add4) to (add2) {};
         \draw[-stealth] (add4) to (eq1) {};
         \draw[-stealth] (x1) to (eq1) {};
        % \draw[->] (c2) to (mul2) {};
        \draw[-stealth] (c4) to (add4) {};
        % \draw[->] (mul1) to (add2) {};
        % \draw[->] (mul2) to (add1) {};
        % \draw[->] (add1) to (sig1) {};
        % \draw[->] (add2) to (sig2) {};
        % \draw[->] (c3) to (add3) {};
        % \draw[->] (sig1) to (add3) {};
        % \draw[->] (sig2) to (add3) {};
        \draw[-stealth] (eq1) to (out) {};

        %loops back
        % \draw[bend right, back] (add1) to (m1) {};
        % \draw[back] (add2) to (m2) {};
        % \draw[back, bend right] (add3) to (m3) {};
        % \draw[back, loop below] (x1) to (x1) {};
    \end{tikzpicture}}
    \end{multicols}
    \vspace{-.5cm}
    \caption{Illustration of recurrent circuit $\textit{rec-}C$ with halting circuit $C_\text{halt}$ computing the $x_1$-th Fibonacci number for $1 < x_1 \in \N$. The dashed arrows mark the recurrent edges and the bold gate $\textit{in}_1$ is the single halting gate, i.\,e.~the gate whose value forms the input for $C_\text{halt}$ along with the iteration number $i$.}
    \label{fig:examplereccirc}
\end{figure}
\end{example}
As the Fibonacci sequence cannot be expressed as a polynomial there is no non-recurrent arithmetic circuit family that computes it. 

Next, we define the notion of recurrent circuit families similar to Definition~\ref{def:extended_circuit_family}.

% \begin{remark}
%     For technical reasons we assume the arity of the halting functions of a recurrent circuit family to be strictly monotonically increasing with the input size of the corresponding circuit.
%     We achieve this by increasing the number of halting gates by adding additional gates to the circuit.
%     This ensures that we only have one halting function per arity.
%     The additional inputs are ignored in the computation of $\halt$ and can be given arbitrary input values.
% \end{remark}
\begin{definition} \label{def:rec_circ_family}
    A \emph{recurrent arithmetic circuit family} $\textit{rec-}\C$ is a sequence $(\textit{rec-}C_{n})_{n \in \N}$ of recurrent circuits, where each circuit $\textit{rec-}C_{n}$ has exactly $n$ input gates.
    The underlying circuits of $\textit{rec-}\C$ define an \emph{extended arithmetic circuit family} $\C = (C_{n})_{n \in \N}$, also called the \emph{underlying circuit family}.
    
    The $\depth$ and $\size$ of $\textit{rec-}\C$ are defined by the $\depth$ and $\size$ of the underlying circuit family $\C$.

    A recurrent arithmetic circuit family $\textit{rec-}\C = (\textit{rec-}C_{n})_{n \in \N}$ computes the function $f_{\textit{rec-}\C} \colon \R^* \to \R^*$ defined as 
    \(
        f_{\textit{rec-}\C}(\ol{x}) \coloneqq f_{{\textit{rec-}C}_{\lvert \ol{x} \rvert}}(\ol{x}).
    \)
    % It holds for the sequence of halting functions $(\halt^n)_{n \in \N}$ that the arity of $\halt^n$ is strictly monotonically increasing with $n$.
    % The same holds for the size $|V^n_\text{halt}|$, the number of halting gates of the recurrent circuit $\textit{rec-}C_n$.
    The set of halting functions of a recurrent circuit family $\textit{rec-}\C=(\textit{rec-}C_{n})_{n \in \N}$ is defined as $\mathcal{F} = \{\halt^n \mid \halt^n \text{ in } \textit{rec-}C_n,\,  n \in \N \}$, where $(\halt^n)_{n \in \N}$ is the sequence of halting functions of $\textit{rec-}\C$.    
\end{definition}

\begin{remark}
%    Let $p \in \N$.
    Without loss of generality, we stipulate that if two recurrent circuits $\textit{rec-}C$ and $\textit{rec-}C'$ in a recurrent circuit family have the same number of halting gates,
    %i.\,e. $|V_\text{halt}| = |V_\text{halt}'|$, 
    then they use the same halting functions, i.e.,  $|V_\text{halt}| = |V_\text{halt}'|$ implies $f_\text{halt}=f'_\text{halt}$.
    %and $f_\text{halt} \colon \N_{>0} \times \R^p \to \{0,1\}$ and $f'_\text{halt} \colon \N_{>0} \times \R^p \to \{0,1\}$ where there exist $i \in \N, \ol{v}\in \R^p$ such that $\halt(i, \ol{v}) \neq f'_\text{halt}(i, \ol{v})$.
    This is not a restriction as we can always increase the arity of a function by adding additional inputs that do not affect the output.
\end{remark}

%In the following definition, we will consider halting functions that are computed by circuits. 
The halting functions for each circuit are determined by the number of halting gates of the circuit.
Note that a recurrent circuit family might only use a subset of the functions of a circuit function family as its halting functions, as not all arities are necessarily utilised, and it might use the same halting function for circuits of different input size.
When expressing that the halting complexity of a recurrent circuit family is $\mathfrak{F}_s$, we mean that the respective set $\mathcal{F}$ is a subset of an $\mathfrak{F}_s$ function family, rather than an element of $\mathfrak{F}_s$.

\begin{definition}
    We say that a function $\R^* \to \R^*$ is in the \emph{recurrent circuit function class} rec$[\mathfrak{F}_\text{halt}]$-$\mathfrak{F}_C$ if it can be computed by a recurrent arithmetic circuit family where the function which its underlying arithmetic circuit family computes is in the circuit function class $\mathfrak{F}_C$ and its set of halting functions is a subset of a circuit family in the circuit function class $\mathfrak{F}_\text{halt}$.
\end{definition}

If not specified otherwise, we consider recurrent circuit function classes of the form $\recF{\mathfrak{F}_s}{\mathfrak{F}}$, i.\,e., recurrent circuit families where the underlying circuit family is in a circuit function class $\mathfrak{F}$ and the set of halting functions is a subset of a $\mathfrak{F}_s$-family, the same class extended by $\mathrm{sign}$. 
%of functions in the same class extended by $\mathrm{sign}$: $\mathfrak{F}_s$.
As every circuit family is also a recurrent circuit family that uses a constant halting function we have that $\mathfrak{F} \subseteq \recF{\mathfrak{F}_s}{\mathfrak{F}}$. 

%\vivian{difficult, only thing we can say is that if the halting depends on the iteration number the circuit halts (if it halts) in polynomial steps, as the halting function can only be a (piecewise) polynomial}

Note that our recurrent circuits are able to simulate halting by reaching a fixed point by adding an additional auxiliary memory gate for each output gate that saves the value of the last iteration of the circuit. 
% Those gates as well as the predecessors of the output gates need to be defined as halting gates. 
% The circuit computing the halting function then only checks for pairwise equality of the values.
The procedure is similar to that used in our Example~\ref{ex:fib_numbers} to save the previously computed Fibonacci numbers. 
Similarly it is able to simulate halting by some (input dependent) initial counter value reaching zero via adding an auxiliary memory gate to the underlying circuit that is decreased by one in each iteration.
The halting function is used to check whether the value of that gate reached zero.%\looseness=-1

For our later results we need to define the composition of two function families from a circuit function class $\recF{\mathfrak{F}_s}{\mathfrak{F}}$.
  Note that the circuit family computing the composition of those function families needs to have access to the sign gate.
%\jonni{In the theorem, we need to say that $\mathfrak{F}$ is a circuit function class. Then then one can simply say that $f$ is in $\recF{\mathfrak{F}_s}{\mathfrak{F}}$ (using Def. 3.14) }
%\jonni{Maybe a corollary that $\recF{\mathfrak{F}_s}{\mathfrak{F_s}} \circ \recF{\mathfrak{F}_s}{\mathfrak{F_s}} = \recF{\mathfrak{F}_s}{\mathfrak{F_s}}$.}

% \jonni{The below should be restricted to the FAC classes. And indeed, as Timon comments, we should concatenate function families and not functions.}
% \vivian{the whole paper is restricted to FAC classes, stated in the circuit section, maybe reformulate it there better so it is clearer here, that when we talk about any $\mathfrak{F}$ it is an FAC class, otherwise would also have problems in all proofs that are based on the composition result}
%\jonni{Let's see if CS people remember the wrongway arounf notation of function composition.}
\begin{restatable}{theorem}{ConcatRecCirc}
\label{thm:concat_rec_circ}
    %  Let $f \colon \R^n \to \R^m$ and $f' \colon \R^m \to \R^{m'}$ be functions from function families in a recurrent circuit function class $\recF{\mathfrak{F}_s}{\mathfrak{F}}$.
    % Then $f' \circ f$ is from a function family in $\recF{\mathfrak{F}_s}{\mathfrak{F}_s}$ where the circuit function class of the underlying circuit family is now required to use sign gates. 
    Let $(f_n)_{n\in \N}$, $(f'_n)_{n\in \N}$ be two function families in $\recF{\mathfrak{F}_s}{\mathfrak{F}}$.
    Then $(f'_n)_{n\in \N} \circ (f_n)_{n\in \N} = \left(f'_{n'}\circ f_n \right)_{n\in\N}$, where $n'\in\N$ is the output dimension of $f_n$, is a function family in $\recF{\mathfrak{F}_s}{\mathfrak{F}_s}$. %where the circuit function class of the underlying circuit family is now required to use sign gates. 
\end{restatable}

    % \timon{Is the formulation of the theorem right? Shouldn't it be that we concatenate two function families?}
    % \laura{tried to reformulate, check if this is what we want}
    % \vivian{looks good, maybe just don't use $\mathcal{F}$ as we used that for a set of halting functions before?}
    % \laura{if you know of any better letter, feel free to edit}

\begin{proofidea}
The circuit computing $f'_{n'} \circ f_n$ is constructed by using the circuits computing $f_n$ and $f'_{n'}$, as well as the circuit computing the halting function of $f_n$.
They are connected in such a way, that the first circuit is iterated until the first halting condition is met. 
Only then the computations of the second circuit are started with meaningful values. 
% \jonni{The following sentenence are hard to read.}
% To directly have access to the output of the first halting function it needs to be incorporated into the the circuit.
% This results in the extension going from a circuit function class without sign for the underlying circuits of $f_n$ and $f'_{n'}$ to one with sign.
% The new halting function for $f'_{n'} \circ f_n$ is a conjunction of the halting functions of $f_n$ and $f'_{n'}$.
To be able to check whether halting function of the first circuit evaluated to 1 the computation of the halting function needs to be directly incorporated into the circuit computing the composition.
As the halting function is from a circuit function class that allows for sign gates the composed function is also in a such a circuit function class.
Our restriction to only $\faci$ classes to ensure the halting function is implementable in the same class.
\end{proofidea}
% \begin{corollary}
% \vivian{if tight on space remove? seems obvious given the result above}
%     The recurrent circuit function class $\recF{\mathfrak{F}_s}{\mathfrak{F}_s}$, where both the underlying circuit families and those computing the halting functions are allowed to use the sign gate, is closed under the composition of functions.
% \end{corollary}

\section{Network Model}
%\section{Model of Computation}
\label{sec:model_of_comp}

\subsection{Recurrent Graph Neural Networks}
\label{sec:gnns}
% \todo[inline]{define: GNNs, activation functions, recurrent GNN}
% \vivian{decide on a short version of a recurrent GNN we want to define here before we go on to define the circuit version, probably only in text form without going into to much detail, maybe intertwine it with the definitions from the next section if we were to change the order there}
In order to construct a circuit based model of recurrent graph neural networks (GNNs) in the next section, we fix the notion of a recurrent GNN first. 
Generally, recurrent GNNs like their non-recurrent counterpart can classify either individual nodes or whole graphs. 
For the purpose of this paper, we will only introduce node classification, since graph classification works analogously.
We focus on aggregate combine GNNs that aggregate the information of every neighbour of a node and then combine this information with the information of the node itself.
Typically GNNs are defined with a fixed number of layers of computations.
A recurrent GNN does not have this restriction.
Instead the feature values of a given labelled input graph are updated multiple times according to the defined computations of the GNN until a predefined halting condition is met.

\begin{definition} \label{def:labelled_graph}
    Let $G=(V,E)$ be a graph with an ordered set of vertices $V$ and a set of undirected edges $E$ on $V$. 
    Let $g_V \colon V \to \R$ be a function which labels the vertices with so called \emph{feature values}.
    We then call $\lgraph = (V, E, g_V)$ a \emph{labelled graph} and denote by $\Graph$ the class of all labelled graphs.
    We denote the multiset of all labels of a set of nodes $V'\subseteq V$ by $\text{val}(V') \coloneqq \lms g_V(v) \mid v \in V' \rms$.
    For a node $v \in V$, the \emph{neighbourhood of $v$} is defined as $N_{\lgraph}(v) \coloneqq \{w \in V \mid \{v, w\} \in E \}$. 
\end{definition}

As we did in the case of recurrent arithmetic circuits, we define recurrent GNNs with feature values in $\R$, but we remark that
all our results also hold in the more general setting of circuits and GNNs  over $\R^k$ for some $k \in \N$.
A recurrent GNN is defined by the following types of functions.

An \emph{aggregation function} is a permutation-invariant function $\text{AGG}: \R \times \dots \times \R \to \R$, a \emph{combination function} is a function $\text{COM}: \R \times \R \to \R$ and a (Boolean) \emph{classification function} is a function $\text{CLS}: \R \to \{0,1\}$ that classifies a real value as either true or false.
A \emph{halting function} is a function $\text{HLT}\colon \N_{>0} \times \R^* \to \{0,1\}$ that gets the current layer number and the unordered multiset of feature values as an input.
Finally, an \emph{activation function} is a function $\actfct\colon \R \to \R$.

\begin{definition}
    \label{def:ac_gnn}
    Let $d \in \N$.
    A \emph{recurrent aggregate combine graph neural network} (rec-AC-GNN) of period $d$ is a tuple $\mathcal{D}=(\{ \text{AGG}^{(i)}\}_{i=1}^d, \{ \text{COM}^{(i)} \}_{i=1}^d, \{\actfct^{(i)}\}_{i=1}^d, \text{HLT},\allowbreak\text{CLS})$, where $\{ \text{AGG}^{(i)}\}_{i=1}^d$ and $\{ \text{COM}^{(i)} \}_{i=1}^d$ are respectively sequences of aggregation and combination functions, $\{\actfct^{(i)}\}_{i=1}^d$ is a sequence of activation functions, HLT is a halting function and CLS is a classification function.

    % Given a labelled graph $\lgraph=(V,E, g_V)$,
    % the rec-AC-GNN model computes values $x_v^{(i)}$ for every $v \in V$ in every layer $0\leq i \leq  d$ as follows:
    % $x_v^{(0)}=g_V(v)$ is the initial feature value of $v$, and for $1 \leq i  \leq d$
    Given a labelled graph $\lgraph=(V,E, g_V)$, the initial feature values $x_v^{(0)}=g_V(v)$ are set for every $v \in V$ and for every layer $1\leq i \leq  d$ the rec-AC-GNN model computes $x_v^{(i)}$ for every $v \in V$ as follows:
    \begin{align*}   
     &x_v^{(i)} = \actfct^{(i)} \left( \text{COM}^{(i)}\left( x_v^{(i-1)}, y \right)\right)\text{,}\\ 
     &\text{where } y = \text{AGG}^{(i)} \left(\ms[\big]{ x_u^{(i-1)} \mid u \in N_\lgraph(v) } \right). %\\
     \end{align*}
    After every layer $i$ the halting function HLT$\left(i, \left(x^{(i)}_v \mid v \in V\right)\right)$ is computed on the current layer number and the tuple of resulting feature values of that layer.
    If HLT evaluates to $1$ in layer $i$, the computation of the recurrent GNN stops and the classification function CLS: $\R \to \{0,1\}$ is applied to the feature values $x_v^{(i)}$.
    Otherwise the computation continues with the next layer $i+1$ or once layer $d$ is reached, periodically starts again with layer 1. 
\end{definition} 

Non-recurrent GNNs can be recovered as the special case, where the value of the halting function is determined by the inputted layer number.

In this paper we focus on the real-valued computation part of GNNs and discard the classification function. 
We consider the feature values $x_v^{(\ell)}$ after the computation of layer $\ell$ that satisfy our halting condition as our output.
While we could also integrate CLS into our model, for our concerns this is not needed.
\subsection{Recurrent Circuit Graph Neural Networks}
\label{sec:cgnns}
%\laura{Possible rework for the section: maybe explain normal C-GNN as a special case first and then introduce the idea of using recurrent circuits in there? i.\,e., Non-re(current C-GNNs are a special case of the recurrent ones: if the halting function is a constant, we end up with a classical \CGNN with a fixed number of layers. If we use recurrent arithmetic circuits in such a GNN, we call the model (rec-C)-GNN to emphasize the use of recurrence in the internal circuits but the lack of an outer recurrence via the halting function.}

%We aim to analyze the expressivity of recurrent graph neural networks via a circuit based recurrent GNN model. 
%We call this model a \emph{recurrent circuit graph neural network} or in short a \emph{recurrent \CGNN}. 
%We assume the halting of our model to be dependent on a predefined halting function similar to that of our model of recurrent arithmetic circuits in Definition~\ref{def:halting_function}.
Our goal is to analyze the expressive power of recurrent graph neural networks, without restricting to aggregate-combine GNNs or any other particular type. Following \cite{DBLP:conf/nips/BarlagHSVV24}, we will study \emph{circuit GNNs} and their recurrent version.

A \emph{basis} of a recurrent C-GNN is a set of functions of a specific recurrent circuit function class along with a set of activation functions and a halting function which is also in a specific circuit function class via a circuit family which computes it.
The respective networks will be based on these functions.
As GNNs are permutation invariant we also assume that our halting functions for the recurrent C-GNNs have this property. 
We call such a function \emph{tail-symmetric} and write $\haltGNN \colon \N_{>0} \times \R^* \to \{0,1\}$ to emphasize it.
The first parameter of the halting function is the number of the current layer.
We can generally define tail-symmetry of a function as follows.\looseness=-1 
% \begin{definition}
%    \label{def:tail_symmetric_fnc}
%    A function $f \colon \N_{>0} \times \R^p \to \{0,1\}$ is \emph{tail-symmetric}, if
%     \[
%         f(i, (v_1 \dots, v_p)) = f(i, \pi(v_1, \dots, v_p))
%     \]
%     for all permutations $\pi$.
% \end{definition}
\begin{definition}
   \label{def:tail_symmetric_fnc}
   A function $f \colon \R^n \to \{0,1\}$ is \emph{tail-symmetric}, if
    \(
        f(v_1 \dots, v_n) = f(v_1, \pi(v_2, \dots, v_n))
    \)
    for all permutations $\pi$.
\end{definition}
% \jonni{Is the following true?}
% This additional parameter isn't required to be considered, we only need it to have a natural connection to a non-recurrent \CGNN.
% \begin{definition}
%     Let $\mathcal{S}$ be a non-empty set of 
%     functions from $\R^*$ to $\R^*$, 
%     $\mathcal{A}$ be a non-empty set of activation functions and let $\haltGNN \colon \N_{>0} \times \R^* \to \{0,1\}$ be a halting function.
%     We assume $\haltGNN$ to be tail-symmetric.
%     Then we call the set $\mathcal{S} \times \mathcal{A} \times \{\haltGNN\}$ a \emph{recurrent C-GNN-basis}.\looseness=-1
% \end{definition}
The same requirement of tail-symmetry holds for the functions that are computed in each layer of a GNN.
We write $t\mathfrak{F}$ to denote only the tail-symmetric functions of the circuit function class $\mathfrak{F}$ and call $t\mathfrak{F}$ a \emph{tail-symmetric class}.
%We call classes of the form $t\mathfrak{F}$ \emph{tail-symmetric classes}.
%\vivian{decide whether we also want to write something like rec-$\mathfrak{F}$, when we talk about recurrent circuit function classes to avoid confusion with plain circuit function classes}
\begin{definition}
     A \emph{recurrent C-GNN-basis} is a set $\mathcal{S} \times \mathcal{A} \times \{\haltGNN\}$, where $\mathcal{S}$ is a subset of some tail-symmetric recurrent circuit function class $t\mathfrak{F}$, $\mathcal{A}$ is set of activation functions, and $\haltGNN$ is from some tail-symmetric circuit function class $\mathfrak{F}_\textnormal{halt}$.
        We call bases of this kind as \emph{$(t\mathfrak{F}, t\mathfrak{F}_\textnormal{halt})$-bases}.
    %\timon{maybe different notation for circ function classes and recurrent circ function classes?}
\end{definition}

Recurrent C-GNNs of a particular basis essentially consist of a function assigning recurrent circuit families and activation functions from its basis to its different layers, and
in addition, a halting function that governs how often those layers are iterated.

% \begin{definition}
%     Let $B=\mathcal{S} \times \mathcal{A} \times \{\haltGNN\}$ be a recurrent C-GNN-basis. 
%     A \emph{recurrent circuit graph neural network (\recCGNN)} $\textit{rec-}\mathcal{N}=(\NN, \haltGNN)$ consists of a periodic function $\NN \colon \N \to \mathcal{S}\times \mathcal{A}$ and a tail-symmetric halting function $\haltGNN \colon \N_{>0} \times \R^* \to \{0,1\}$.
%     If all functions in $\mathcal{S}$ belong to a circuit function class $\mathfrak{F}_\mathcal{S}$ and the function $\haltGNN$ is computable by a $\mathfrak{F}_\text{halt}$-circuit family, we also say that $\textit{rec-}\mathcal{N}$ is a rec$[\mathfrak{F}_\text{halt}]$-$((\mathfrak{F}_S, \mathcal{A})$-GNN).\looseness=-1 
% \end{definition}

%\jonni{How do you pronounce the basis below?}
\begin{definition}
    Let $B = \mathcal{S} \times \mathcal{A} \times \{\haltGNN\}$ be a $(t\mathfrak{F}, t\mathfrak{F}_\textnormal{halt})$-basis. 
    A \emph{recurrent circuit graph neural network} $\textit{rec-}\mathcal{N}=(\NN, \haltGNN)$ of basis $B$ consists of a periodic function $\NN \colon \N \to \mathcal{S}\times \mathcal{A}$ and a tail-symmetric halting function $\haltGNN \colon \N_{>0} \times \R^* \to \{0,1\}$.
    We say that $\textit{rec-}\mathcal{N}$ is a rec$[t\mathfrak{F}_\text{halt}]$-$(t\mathfrak{F}, \mathcal{A})$-GNN.\looseness=-1 
\end{definition}
% \laura{do we still want to define them generally as in the last paper or go straight to circuit function classes only? In the end, that is what we're interested in in this paper and we probably won't even have space to discuss the generalized framework. I think it might be beneficial to cut unnecessary things, we have a lot of defs/concepts to convey already}
% \timon{I agree. It's probably a good idea to cut everything down to what we actually need. Afterwards, we can still think about adding more explanation/context.}
% \vivian{tried a first version of that}
% Notice that we now have a family of halting functions, as we define the halting to be dependent on all feature vectors of the input graph and therefore on the input size.

\begin{definition}
\label{def:rec_cgnn}
    Let $\textit{rec-}\mathcal{N}= \left(\NN, \haltGNN\right)$ be a recurrent \CGNN.
    The partial function $f_{\textit{rec-}\mathcal{N}} \colon \Graph \to \Graph$ computed by $\textit{rec-}\mathcal{N}$ is defined as follows: 
    Let $\lgraph = (V, E, g_V)$ be a labelled graph.
    The labelled graph $\lgraph'=f_{\textit{rec-}\mathcal{N}}(\lgraph)$ has the same structure as $\lgraph$, however, its nodes have different feature vectors.
    That is $\lgraph' = (V, E, h_V)$ and $h_V$ is defined 
    %inductively 
    as follows:
    \begin{align*}
       h_V^{(0)}(w) &\coloneqq~  g_V(w)\\ 
        h_V^{(i)}(w) &\coloneqq~  \actfct^{(i)} \left( f_{\mathcal{C}^{(i)}}\left(h_V^{(i-1)}(w), M \right) \right),\\
        &\text{ with } M = \lms h_V^{(i-1)}(u) \mid u \in N_{\lgraph}(w) \rms
    \end{align*}
    where ${\NN(i)=\left(\mathcal{C}^{(i)},\actfct ^{(i)}\right)}$.
    Then ${h_V=h_V^{(\ell)}}$ is the smallest value such that $\haltGNN\left(\ell, \lms h_V^{(\ell)}(v) \mid v \in V \rms \right)=1$ where ${\ell \in \N}$. 
    Otherwise $f_{\textit{rec-}\mathcal{N}}$ is undefined.
\end{definition}
The halting function of a recurrent C-GNN is applied to the feature vectors of all nodes contrary to recurrent arithmetic circuits where only the values of some of the gates were taken into consideration.

The recurrent \CGNN model that we have just defined computes functions from labelled graphs to labelled graphs (with the same graph structure) with recurrence in two places. 
To distinguish those two types of recurrence we call the recurrence of the complete network \emph{outer recurrence} and that of the internal circuits \emph{inner recurrence}.

% \jonni{Most of what is said in the next paragraphs below can be deleted, and replaced with "We call rec$[t\mathfrak{F}_\text{halt}]$-$(t\mathfrak{F}, \mathcal{A})$-GNNs where $t\mathfrak{F}$ is a non-recurrent circuit function class a \emph{circuit graph neural networks with outer recurrence}. On the other hand, we call $rec[t\mathfrak{F}_s]$-$(t\mathfrak{F}, \mathcal{A})$-GNNs, where   $t\mathfrak{F}$ is a recurrent circuit function class and $t\mathfrak{F}_\text{halt}$ is determined by the layer number, a \emph{circuit graph neural network with inner recurrence}".}
% As mentioned, every classical circuit family is a subset of its recurrent version.
We call rec$[t\mathfrak{F}_\text{halt}]$-$(t\mathfrak{F}, \mathcal{A})$-GNNs where $t\mathfrak{F}$ is a non-recurrent circuit function class a \emph{circuit graph neural networks with outer recurrence}. 
On the other hand, we call $(\text{rec}[\mathfrak{F}_s]\text{-}t\mathfrak{F}, \mathcal{A})$-GNNs, where $\text{rec}[\mathfrak{F}_s]$-$t\mathfrak{F}$ is a recurrent circuit function class and the halting is determined by a fixed layer number, a \emph{circuit graph neural network with inner recurrence}.
Finally we call $(t\mathfrak{F}, \mathcal{A})$-GNNs, where $t\mathfrak{F}$ is a non-recurrent circuit function class, \emph{circuit graph neural networks}.
This model is closest to that of a standard non-recurrent GNN and was already introduced in \cite{DBLP:conf/nips/BarlagHSVV24}.
By our definition of recurrent arithmetic circuits every function computable by a \CGNN without recurrence is also computable by any of the introduced recurrent \CGNN{}s.
An overview of the different models is given in Figure~\ref{fig:cgnn_models}. 
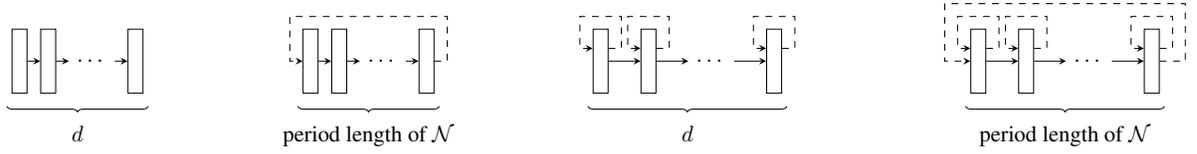
\begin{figure*}[tb]
    \centering
    \begin{subfigure}[t]{0.2\textwidth}
    \centering
        \scalebox{.85}{\begin{tikzpicture}[baseline, gate/.style={circle, minimum width=0.6cm, text width=6mm, align=center, draw}, brace/.style={
             decoration={brace, mirror},decorate}, position label/.style={
               below = 3pt,
               text height = 2ex,
               text depth = 1ex
            }]
            \node[rectangle, draw, minimum height = 1cm]    (l1)    {};
            \node[rectangle, draw, minimum height = 1cm]    (l2)    [right=0.2cm of l1]    {};
            \node[]                                         (dots)  [right=0.2cm of l2]     {$\dots$};
            \node[rectangle, draw, minimum height = 1cm]    (lk)    [right=0.2cm of dots]    {};
            \draw[brace] ($(l1.south)+ (-0.2,-0.2)$) -- node [position label, pos=0.5] {$d$} ($(lk.south)+ (0.2,-0.2)$);
            \draw[-stealth]     (l1) -- (l2);
            \draw[-stealth]     (l2)    --  (dots);
            \draw[-stealth]     (dots)  --  (lk);
        \end{tikzpicture}}
        \caption{C-GNN without recurrence and a fixed depth $d\in\N$ }
    \end{subfigure}
    \hspace{0.1cm}
    \begin{subfigure}[t]{0.2\textwidth}
    \centering
    \scalebox{.85}{
         \begin{tikzpicture}[baseline, gate/.style={circle, minimum width=0.6cm, text width=6mm, align=center, draw}, brace/.style={
             decoration={brace, mirror},decorate}, position label/.style={
               below = 3pt,
               text height = 2ex,
               text depth = 1ex
            }]
            \node[rectangle, draw, minimum height = 1cm]    (l1)    {};
            \node[rectangle, draw, minimum height = 1cm]    (l2)    [right=0.2cm of l1]    {};
            \node[]                                         (dots)  [right=0.2cm of l2]     {$\dots$};
            \node[rectangle, draw, minimum height = 1cm]    (lk)    [right=0.2cm of dots]    {};
            \draw[brace] ($(l1.south)+ (-0.2,-0.2)$) -- node [position label, pos=0.5] {period length of $\NN$ } ($(lk.south)+ (0.2,-0.2)$);
            \draw[-stealth]     (l1) -- (l2);
            \draw[-stealth]     (l2)    --  (dots);
            \draw[-stealth]     (dots)  --  (lk);
            \draw[dashed, -stealth] (lk.east) |- ++(0.2,0) |- ++ (0,.7 ) |- ($(l1.west)+ (-0.2,0.7)$) |- (l1.west) ;
        \end{tikzpicture}}
        \caption{\CGNN with only outer recurrence}
    \end{subfigure}
    \hspace{0.1cm}
   \begin{subfigure}[t]{0.25\textwidth}
   \centering
       \scalebox{.85}{\begin{tikzpicture}[baseline, gate/.style={circle, minimum width=0.6cm, text width=6mm, align=center, draw}, brace/.style={
             decoration={brace, mirror},decorate}, position label/.style={
               below = 3pt,
               text height = 2ex,
               text depth = 1ex
            }]
            \node[rectangle, draw, minimum height = 1cm]    (l1)    {};
            \node[rectangle, draw, minimum height = 1cm]    (l2)    [right=0.5cm of l1]    {};
            \node[]                                         (dots)  [right=0.5cm of l2]     {$\dots$};
            \node[rectangle, draw, minimum height = 1cm]    (lk)    [right=0.5cm of dots]    {};
            \draw[brace] ($(l1.south)+ (-0.2,-0.2)$) -- node [position label, pos=0.5] {$d$} ($(lk.south)+ (0.2,-0.2)$);
            \draw[-stealth]     (l1) -- (l2);
            \draw[-stealth]     (l2)    --  (dots);
            \draw[-stealth]     (dots)  --  (lk);
            \draw[dashed, -stealth] ($(l1.east)+(0,0.2)$) |- ++(0.2,0)  |- ($(l1.west)+ (-0.2,0.7)$) |- ($(l1.west) + (0,.2)$) ;
            \draw[dashed, -stealth] ($(l2.east)+(0,0.2)$) |- ++(0.2,0)  |- ($(l2.west)+ (-0.2,0.7)$) |- ($(l2.west) + (0,.2)$) ;
            \draw[dashed, -stealth] ($(lk.east)+(0,0.2)$) |- ++(0.2,0)  |- ($(lk.west)+ (-0.2,0.7)$) |- ($(lk.west) + (0,.2)$) ;
        \end{tikzpicture}}
    \caption{\CGNN with only inner recurrence and a fixed depth $d\in \N$}
   \end{subfigure}
   \hspace{0.4cm}
    \begin{subfigure}[t]{0.25\textwidth}
    \centering
        \scalebox{.85}{\begin{tikzpicture}[baseline, gate/.style={circle, minimum width=0.6cm, text width=6mm, align=center, draw}, brace/.style={
             decoration={brace, mirror},decorate}, position label/.style={
               below = 3pt,
               text height = 2ex,
               text depth = 1ex
            }]
            \node[rectangle, draw, minimum height = 1cm]    (l1)    {};
            \node[rectangle, draw, minimum height = 1cm]    (l2)    [right=0.5cm of l1]    {};
            \node[]                                         (dots)  [right=0.5cm of l2]     {$\dots$};
            \node[rectangle, draw, minimum height = 1cm]    (lk)    [right=0.5cm of dots]    {};
            \draw[brace] ($(l1.south)+ (-0.2,-0.2)$) -- node [position label, pos=0.5] {period length of $\NN$} ($(lk.south)+ (0.2,-0.2)$);
            \draw[-stealth]     (l1) -- (l2);
            \draw[-stealth]     (l2)    --  (dots);
            \draw[-stealth]     (dots)  --  (lk);
            \draw[dashed, -stealth] ($(l1.east)+(0,0.2)$) |- ++(0.2,0)  |- ($(l1.west)+ (-0.2,0.7)$) |- ($(l1.west) + (0,.2)$) ;
            \draw[dashed, -stealth] ($(l2.east)+(0,0.2)$) |- ++(0.2,0)  |- ($(l2.west)+ (-0.2,0.7)$) |- ($(l2.west) + (0,.2)$) ;
            \draw[dashed, -stealth] ($(lk.east)+(0,0.2)$) |- ++(0.2,0)  |- ($(lk.west)+ (-0.2,0.7)$) |- ($(lk.west) + (0,.2)$) ;
            \draw[dashed, -stealth] (lk.east) |- ++(0.4,0) |- ++ (0,.9 ) |- ($(l1.west)+ (-0.4,0.9)$) |- (l1.west) ;
        \end{tikzpicture}}
        \caption{\CGNN with inner and outer recurrence}
    \end{subfigure}
    \caption{The different models of (recurrent) \CGNN{}s. One block stands for one layer of the different (recurrent) \CGNN models $\NN$, dashed edges mark recurrence}
    \label{fig:cgnn_models}
\end{figure*}

An overview of the relations between the different \CGNN models is given in Figure~\ref{fig:cgnn_results}.
\begin{figure}[H]
    \centering
    \begin{tikzpicture}
        \node[align=center]     (doublerec)                                     {\CGNN\\ inner and\\ outer recurrence};
        \node[align=center]     (neq)           [below =0.2 of doublerec]       {$\neq$};
        \node[align=center]     (innerrec)      [below left =.1of doublerec]    {\CGNN\\ inner recurrence};
        \node[align=center]     (outerrec)      [below right=.1of doublerec]    {\CGNN\\ outer recurrence};
        \node[align=center]     (norec)         [below =1of doublerec]          {\CGNN\\ without recurrence};

        \draw[]     (doublerec.west)     --      (innerrec.north);
        \draw[]     (doublerec.east)     --      (outerrec.north);
        \draw[]     (outerrec.south)     --      (norec.east);
        \draw[]     (innerrec.south)     --      (norec.west);
    \end{tikzpicture}
    \caption{Assume the same circuit function classes for recurrence/ underlying circuit and the same set of activation functions, lines denote subclasses of computable functions}
    \label{fig:cgnn_results}
\end{figure}

\section{Relations Between Recurrent Circuits and Recurrent GNNs}\label{sec:circs_and_gnns}
%\laura{maybe add smth that we generally assume 'same' circuit families for both and therefore show similarities/differences in computational power in the models themselves?}
Our aim is to relate our different models of recurrent \CGNN{}s to the model of recurrent arithmetic circuits we introduced before, i.\,e.~to show which functions that can be computed by recurrent arithmetic circuits can also be computed using a recurrent \CGNN and vice versa.
We generally assume that the the circuits and \CGNN{}s have the same resources, i.e. the same circuit function classes on both sides and access to the same additional functions or activation functions. This allows us to investigate the similarities and differences in expressive power between the two models.

\subsection{Overview of Results} \label{sec:reoadmap}
We show that there exist recurrent arithmetic circuits that are able to simulate the computations of our recurrent \CGNN model in the sense that when given an encoding of a labelled graph they compute the updated feature values and output the updated (encoded) labelled graph.
We distinguish between recurrent arithmetic circuits where the underlying circuits have access to the sign function an can therefore compute non-continuous functions and those that do not.
An overview of the results of that direction is given in Table~\ref{tab:cgnn_to_circ_results}.
\begin{table}[tb]
    \centering
    \begin{tabular}{c|c|| c | c r}
    \toprule
         \multicolumn{2}{c||}{\CGNN with} &  \multicolumn{2}{c}{Recurrent circuit} &\\
        \multirow{2}{1.3cm}{\centering Outer recurrence} & \multirow{2}{1.3cm}{\centering Inner recurrence} & \multirow{2}{1cm}{\centering without sign}& \multirow{2}{1cm}{with sign}&\\
        & &  & &  \\
        \midrule
        \cmark & \xmark & \cmark &\xmark & (Thm. \ref{thm:recCGNN_in_circ})\\
          \xmark & \cmark &  \xmark &\cmark& (Thm. \ref{thm:CrecGNN_in_circ})\\
         \cmark & \cmark &  \xmark &\cmark & (Cor. \ref{cor:CrecrecGNN_in_circ})\\
         \bottomrule
         \end{tabular}
    \caption{Simulating Recurrent C-GNNs with Recurrent Circuits.\\
    One row in the table represents one result.}
    \label{tab:cgnn_to_circ_results}
\end{table}

% \begin{table*}[]
%     \centering
%     \begin{tabular}{c|p{10cm}}
%         Recurrent circuits &  Outer Recurrent C-GNNs with sign in the internal circuits and activation functions only in the internal circuits (only identity as actual activation functions)\\
%         Recurrent circuits (restricted to predecessor normal form) & Outer Recurrent C-GNNs with sign in the internal circuits and activation functions in the usual GNN sense\\
%         Recurrent circuits (restricted to symmetric functions) & Inner Recurrent C-GNNs without sign in the internal circuits and activation functions only in the internal circuits (only identity as actual activation functions)\\
%         \multicolumn{2}{p{18cm}}{Recurrent circuits (restricted to symmetric functions) “cannot be simulated by” Inner Recurrent C-GNNs using activation functions in the usual GNN sense.}
%     \end{tabular}
%     \caption{Simulating Recurrent Circuits with Recurrent C-GNNs}
%     \label{tab:placeholder1}
% \end{table*}

We moreover show that functions computed by recurrent circuits with different restrictions on the functions or the circuits themselves can be simulated by our recurrent \CGNN models in the sense that there exists a recurrent \CGNN for which there exists an input (labelled graph) such that the output (labelled graph) has the output of the recurrent circuit among its feature values. 
Here we differentiate between \CGNN{}s by their modes of recurrence --- inner or outer recurrence --- as well as how they are allowed to use activation functions --- either arbitrarily locally in their internal circuits  or globally on the whole graph as in the usual sense of activation functions in GNNs.
Also of interest is whether the internal circuits of the \CGNN have access to sign gates which allow them to compute non-continuous functions.
The results are shown in Table~\ref{tab:circ_to_cgnn_results}.
We can also show that recurrent circuits (even when restricted to symmetric functions) cannot be simulated by Inner Recurrent \CGNN{}s using activation functions in the usual GNN sense (Lemma~\ref{lem:circ_to_inner_rec_with_activation_fnc}). 
% \begin{table*}[]
%     \centering
%     \begin{tabular}{c||c|c|c}
%          Circuit model & Simulated by & sign in internal circuits & Activation functions \\
%          \hline
%          Recurrent Circuits & Outer Recurrent \CGNN & \cmark & only in internal circuits\\
%          Predecessor Form & Outer Recurrent \CGNN & \xmark & in usual GNN sense\\
%          Symmetric Functions & Inner Recurrent \CGNN & \xmark & only in internal circuits\\
%          \multicolumn{4}{p{18cm}}{Recurrent circuits (restricted to symmetric functions) “cannot be simulated by” Inner Recurrent C-GNNs using activation functions in the usual GNN sense.}
%     \end{tabular}
%     \caption{Simulating Recurrent Circuits with Recurrent \CGNN{}s}
%     \label{tab:circ_to_cgnn_results}
% \end{table*}
\begin{table*}[tb]
    \centering
    \begin{tabular}{c||c|c|c|c|c r}
    \toprule
         \multirow{3}*{Recurrent circuit} & \multicolumn{5}{c}{Recurrent \CGNN}    &\\
         & \multirow{2}*{Outer recurrence}& \multirow{2}*{Inner recurrence}& \multirow{2}{2.3cm}{\centering sign in internal circuits}& \multicolumn{2}{c}{Activation functions}&\\
         &  &  &  & local & global&\\
         %\cmidrule{1-6}
         \midrule
         no restrictions & \cmark & \xmark &\cmark & \cmark & \xmark &(Thm. \ref{thm:circ_to_outer_rec_gnn_with_activation_fnc})\\
         Predecessor Form & \cmark & \xmark &\xmark & \xmark & \cmark& (Thm. \ref{thm:circ_to_outer_rec_gnn_no_activation_fnc})\\
         Symmetric Functions & \xmark & \cmark &\xmark & \cmark & \xmark &(Thm. \ref{thm:circ_to_inner_rec})\\
         %Symmetric Functions & \multicolumn{2}{c|}{no simulation possible} &\cmark & \xmark & \cmark\\
         \bottomrule
    \end{tabular}
    \caption{Simulating Recurrent Circuits with Recurrent \CGNN{}s.\\
    One row in the table represents one result.}
    \label{tab:circ_to_cgnn_results}
\end{table*}

We get the correspondences of the two models under reasonable reductions/encodings that are necessary as \CGNN{}s and arithmetic circuits operate on different types of input, namely labelled graphs and real valued input strings. 
For the most general models (for \CGNN{}s and recurrent circuits) we have a one-to-one correspondence. 
Every recurrent circuit can be simulated with an outer and inner recurrent \CGNN and every outer and inner recurrent \CGNN can be simulated with a recurrent circuit. This also shows that under logspace reductions outer recurrent \CGNN{}s have the same expressivity as \CGNN{}s with both inner and outer recurrence. 
For the other simulations some restrictions are assumed. 

\subsection{Simulating Recurrent GNNs With Recurrent Arithmetic Circuits} \label{sec:gnn_to_circ}
While recurrent \CGNN{}s work on labelled graphs, arithmetic circuits only get an ordered tuple of $\R$-values as an input.
% To be able to talk about arithmetic circuits computing the same function as a recurrent \CGNN we need to define an encoding of labelled graphs that can be used as an input to an arithmetic circuit.
To be able to talk about arithmetic circuits computing the same functions as recurrent \CGNN{}s, we encode labelled graphs as real valued tuples that can then be used as an input to an arithmetic circuit.
We write $\langle \mathfrak{G} \rangle$ for the encoding of a labelled graph $\mathfrak{G}$ as its adjacency matrix followed by the respective feature values.\looseness=-1
%\jonni{Sentence about the encoding.}
% \begin{definition}
%     \vivian{candidate to be shortened/ replaced by explanatory sentence}
%     Let $\lgraph = (V, E, g_V)$ be a labelled graph, $n \coloneqq \lvert V \rvert$, and
%     % We assume $V$ to be of size $n$ and be ordered by $[ \,n ] \,$.
%     $M=\text{adj}(\lgraph)$ be the adjacency matrix of $(V, E)$, where the columns are ordered in accordance with the ordering of $V$.
%     We write $\enc{M}$ to denote the encoding of $M$ as the $n^2$ matrix entries $m_{ij}$, ordered in a row wise fashion. 
%     We write $\enc{\lgraph}$ to denote the encoding of $\lgraph$ as a tuple of real values, such that $\enc{\lgraph}= \left(\enc{M}, \text{val}(\lgraph)\right) \in \R^{n^2+n}$, which consists of the encoding of $M$ followed by $\text{val}(\lgraph)$, the feature values of $\lgraph$.
%     The feature values $\text{val}(\lgraph)$ are $g_V(v)$ for all $v \in V$ and ordered like $V$.
% \end{definition}

As defined in Section~\ref{sec:rec_arithmetic_circs}, we assume that our recurrent arithmetic circuits only have access to the sign function inside their halting function. 
However by a result from \cite[Lemma A.4]{DBLP:journals/corr/abs-2402-17805} we are still able to check for equality for a set of fixed values in the underlying circuits in constant depth.%\timonS{That Lemma is not in the published version, I believe, so we should probably cite the arxiv version. Also explicitly give the Lemma number.}.

The following theorem states that any \CGNN with outer recurrence can be simulated by a recurrent circuit.
Given an encoding of a labelled graph as an input the recurrent circuit computes the encoding of the labelled graph the \CGNN would output.
Using the reverse procedure of the encoding this can then be decoded back to a labelled graph.
%In this manner, outer recurrence \CGNN{}s can be simulated by arithmetic circuits:
%\jonni{Is $<G>$ defined somewhere?}
\begin{restatable}{theorem}{outerrectocirc} \label{thm:recCGNN_in_circ}

         Let $\textit{rec-}\mathcal{N}$ be a $\textnormal{rec}[t\mathfrak{F}_s]$-$(t\mathfrak{F}, \mathcal{A})$-GNN.
    Then there exists a function $f_\textnormal{circ} \in \recF{\mathfrak{F}_s}{\mathfrak{F}}[\mathcal{A}]$
    % \colon \R^* \to \R^*$ 
    s.\,t. for all labelled graphs $\mathfrak{G}$ we have
        ${f_\textnormal{circ}(\langle \mathfrak{G} \rangle) = \langle f_{\textit{rec-}\mathcal{N}}(\mathfrak{G}) \rangle.}$
        % \vivian{before: \[
        %     f_{\textit{rec-}\mathcal{N}}(\mathfrak{G}) = \mathfrak{G}' \Longleftrightarrow f_\textnormal{circ}(\langle \mathfrak{G} \rangle) = \langle \mathfrak{G}' \rangle.
        % \] }
\end{restatable}

% \begin{proof}
%     % Let $B=(\recF{t\mathfrak{F_s}}{t\mathfrak{F}}) \times \mathcal{A}\times \{\haltGNN\}$ be a recurrent \CGNN basis and $\mathcal{N}_\textit{rec} = (\mathcal{N}, \haltGNN)$ with $\mathcal{N} \colon \N \to (\recF{t\mathfrak{F_s}}{t\mathfrak{F}}) \times \mathcal{A}$ a periodic function with period length $d\in\N$ and $\haltGNN(i, D) \coloneqq \chi[i=d]$ a $(\text{rec}[\mathfrak{F}_s]\text{-}\mathfrak{F})\text{-GNN}$. 
%     % W.l.o.g. let the depth of $\mathfrak{F}$-circuit families be bounded by $\bO((\log n)^i)$ and let their size be bounded by $\bO(n^{\bO(1)})$.
%     % For any size $n$ of a labelled graph $\mathfrak{G}$ (or, more precisely, for any encoding size $n$ of $\mathfrak{G}$), we create one circuit $C^{\mathcal{N}}_n$, such that the claim holds.
    
% \end{proof}
\begin{proofidea}
The main idea here is to build a circuit $C_n^{\mathcal{N}}$ for every input length $n$ where we have the $d$ different circuit families of $\recNN$ in parallel for period length $d$ of $\recNN$. We use a modulus $d$ counter to decide which result to feed back into the next iteration. 
    %as sketched in Figure \ref{fig:outerreccgnntoreccirc}.
\end{proofidea}
%\lauraS{we need some kind of spacing here otherwise it is absolutely confusing where the proof sketches actually end}

% Contrary to the previous result, in order to simulate the function computed by an inner recurrent \CGNN via a recurrent circuit, the underlying circuit of that recurrent circuit must have access to a sign gate., i.\,e.~that be from a circuit function class of the form $\recF{\mathfrak{F}_s}{\mathfrak{F}_s}[\mathcal{A}]$.
Contrary to the previous result, in order to simulate the function computed by an inner recurrent \CGNN via a recurrent circuit, we need the underlying circuit of that recurrent circuit to have access to a sign gate., i.\,e. to compute a function from a class of the form ${\mathfrak{F}_s}[\mathcal{A}]$.
This is due to the fact that the proof is based on the composition result of Theorem~\ref{thm:concat_rec_circ}.

% \jonni{The next paragraph is wierd.}
% We can therefore assume that there are functions that are computable by an inner recurrent \CGNN that an outer recurrent \CGNN cannot compute. \vivianS{is this the correct reasoning?}
% \timonS{I think this claim is a bit too ambitious if we are not actually able to show this rigorously.}

\begin{restatable}{theorem}{InnerRecToCirc}
    \label{thm:CrecGNN_in_circ}
    
    Let $\textit{rec-}\mathcal{N}$ be a $(\textnormal{rec}[t\mathfrak{F}_s]$-$t\mathfrak{F}, \mathcal{A})$-GNN.
    Then there exists a function $f_\textnormal{circ} \in \recF{\mathfrak{F}_s}{\mathfrak{F}_s}[\mathcal{A}]$
    % \colon \R^* \to \R^*$ 
    s.\,t. for all labelled graphs $\mathfrak{G}$ we have ${f_\textnormal{circ}(\langle \mathfrak{G} \rangle) = \langle f_{\textit{rec-}\mathcal{N}}(\mathfrak{G}) \rangle.}$
\end{restatable}
\begin{proofidea}
The inner recurrent \CGNN has a fixed number of layers. 
We therefore compose the functions computed by the recurrent circuits of each layer while ensuring the correct connections between nodes are represented by the recurrent circuit that simulates the computations of the recurrent \CGNN.
\end{proofidea}

From Theorems~\ref{thm:recCGNN_in_circ} and \ref{thm:CrecGNN_in_circ} follows that a recurrent \CGNN with both inner and outer recurrence can be simulated by a recurrent circuit.
The proof combines the techniques used in the proofs of the mentioned theorems.

\begin{restatable}{corollary}{DoubleRecToCirc}
    \label{cor:CrecrecGNN_in_circ}
    % Let $t\mathfrak{F}$ be a tail-symmetric circuit function class and let $\textit{rec-}\mathcal{N}$ be a $\textnormal{rec}[t\mathfrak{F}_s]$-$(\textnormal{rec}[t\mathfrak{F}_s]$-$t\mathfrak{F}, \mathcal{A})$-GNN.
    % Then there exists a function $f_\textnormal{circ} \colon \R^* \to \R^*$ defined such that for all labelled graphs $\mathfrak{G}$ it holds that
    %     \[
    %         f_{\textit{rec-}\mathcal{N}}(\mathfrak{G}) = \mathfrak{G}' \Longleftrightarrow f_\textnormal{circ}(\langle \mathfrak{G} \rangle) = \langle \mathfrak{G}' \rangle,
    %     \]
    %     and $f_\textnormal{circ}$ is contained in $\recF{\mathfrak{F}_s}{\mathfrak{F}_s}$.

    Let $\textit{rec-}\mathcal{N}$ be a $\textnormal{rec}[t\mathfrak{F}_s]$-$(\textnormal{rec}[t\mathfrak{F}_s]$-$t\mathfrak{F}, \mathcal{A})$-GNN.
    Then there exists a function $f_\textnormal{circ} \in \recF{\mathfrak{F}_s}{\mathfrak{F}_s}[\mathcal{A}]$
    % \colon \R^* \to \R^*$ 
    defined such that for all labelled graphs $\mathfrak{G}$ it holds that ${
            f_\textnormal{circ}(\langle \mathfrak{G} \rangle) = \langle f_{\textit{rec-}\mathcal{N}}(\mathfrak{G}) \rangle.}$
\end{restatable}

%\vivian{maybe go more into detail what the different circuit function classes for inner and outer rec mean}
% \vivian{We don't actually need this proof, can get rid of it}
% \begin{proof}
%     Proceed as in the proof of Theorem \ref{thm:CrecGNN_in_circ} by concatenating the circuits in one period.
%     Then have the halting function iterate over that.
%     Use the output gates of every circuit as the gates that the new halting function is dependent on. 
%     This allows for the computation to stop in the middle of the period of circuits. \vivianS{instead of writing that here use Theorem \ref{thm:recCGNN_in_circ} instead} 
% \end{proof}

\subsection{Simulating Recurrent Arithmetic Circuits With Recurrent C-GNNs} \label{sec:circ_to_gnn}
\vspace{2mm}\noindent \textbf{Simulation with Outer Recurrent C-GNNs}
%\subsubsection{Simulation with Outer Recurrent C-GNNs}
%\vivian{could shorten encoding description further if we leave out the 3 steps for the symbolic encoding}
%\vivian{do we need to add to our theorems what $\gre{C}$ is? currently we only write logspace symbolic encoding $\sgre{C}$. Or maybe make more clear how the 2 are directly connected}

\noindent Our goal is to construct recurrent \CGNN{}s which are able to simulate the computations of recurrent circuits, i.\,e.~they have the output of a given circuit among the feature vectors of the output labelled graph when given a labelled graph as input.
Those input labelled graphs are constructed given the recurrent arithmetic circuit. 
We therefore need to fix the notion of a graph encoding.
We bound them to be logspace Turing computabile to ensure that our encodings are not too powerful.
As recurrent arithmetic circuits are usually defined via their tuple as described in Definition~\ref{def:rec_circ} and then given an input, we use a so called \emph{symbolic logspace labelled graph encoding} $\sgre{C}$ first which, given a recurrent circuit $C$ as input, uses symbols for the input, auxiliary memory and constant gates instead of the values from $\R$ as node labels. 
This ensures that our encoding can be computed by a Turing machine which only works over a finite alphabet.
% The encoding consists of three components: the encoding of the given recurrent arithmetic circuit into an input for a Turing machine, the actual transformation of this encoded circuit into an encoded symbolic labelled graph and the decoding of this into a symbolic labelled graph.
% The encoding and decoding are fixed operations and the transformation is restricted to logarithmic space.
Afterwards we use the corresponding \emph{labelled graph encoding} $\gre{C}$ that, additionally given the input $\ol{x}$, replaces the symbols for inputs, auxiliary memory and constant values with their respective values and outputs a labelled graph with labels from $\R$.\looseness=-1

% To simplify the proofs of the theorems of this section, from now on we assume the recurrent circuits to be in \emph{path-length normal form}, meaning that the following property holds for both the underlying circuits and the halting circuits. 
% This assumption does not form a restriction as shown by \cite{DBLP:conf/wollic/BarlagV21}.
% \begin{definition}
%     A circuit $C$ is in \emph{path-length normal form} if every path from an input to an output gate has the same length and every non-input and non-output gate has exactly one successor.
% \end{definition}

To simplify the proofs of the theorems of this section, from now on we assume the recurrent circuits to be \emph{balanced DAGs}, meaning that the following property holds for both the underlying circuits and the halting circuits. 
%This assumption does not form a restriction as shown by \cite{DBLP:conf/wollic/BarlagV21}.

\begin{definition}
    A circuit $C$ is a \emph{balanced DAG} if for each gate $g$ in $C$ every path from an input to $g$ has the same length.
\end{definition}

It was shown that this assumption does not form a restriction in the non-recurrent setting~\cite{DBLP:journals/mscs/BarlagCG24} and the same argument can be applied here, since the property of being a balanced DAG only concerns the structure of the underlying extended circuit and has no bearing on the recurrent part of the circuit.

%can be seen by the following lemma. 
% \begin{lemma}[\protect\cite{DBLP:conf/wollic/BarlagV21}]
%     Let $\mathcal{C}$ be an $\mathfrak{F}$-circuit family.
%     Then for every circuit $C \in \mathcal{C}$ there exists a circuit $C'$ of the same depth in path-length normal form such that \(f_{C'}(\ol{x}) = f_{C}(\ol{x})\), for all $\ol{x} \in (\R^k)^*$.
% \end{lemma} \vivianS{maybe remark that dummy gates are used for this process? do readers need this information to understand the proofs?}

We show that recurrent arithmetic circuits can be simulated by \CGNN{}s with outer recurrence.
Here the halting of the recurrent circuit needs to be simulated by the outer recurrent \CGNN, which only has a tail-symmetric halting function available that operates on all feature values.
Therefore, we have to move some of the computations of the halting functions to the internal circuits of the \CGNN.
% For this reason they have to have access to the sign function but otherwise are of the same circuit function class as the underlying circuits of the recurrent circuit that is being simulated. \timonS{splitting hairs again: circuits are not in a circuit function class}
This means that they need to have access to the sign function, however their size and depth restrictions are the same as for the underlying circuit of the recurrent circuit which they are simulating.
% For this reason they have to have access to the sign function but otherwise are of the same circuit function class as the underlying circuits of the recurrent circuit that is being simulated. \timonS{splitting hairs again: circuits are not in a circuit function class}

The halting function of the \CGNN stays in the same circuit function class as that of the simulated circuit. %\vivianS{to be completely precise it is actually $\faco$, as we only need to check whether more than $n+k+1$ gates have value 2, $n$ and $k$ are given by the circuit}

\begin{restatable}{theorem}{circToOuterRecNoActivationFnc}
    \label{thm:circ_to_outer_rec_gnn_no_activation_fnc}
    Let $\mathcal{C} = (C_n)_{n \in \N}$ be a $\text{rec}[\mathfrak{F}_s]\text{-}\mathfrak{F}[\mathcal{A}]$-circuit family.
    Then there exists a $\text{rec}[\mathfrak{tF}_s]\text{-}(\mathfrak{tF}_s[\mathcal{A}], \text{id})$-GNN rec-$\NN = (\mathcal{N}, \haltGNN)$ and a symbolic logspace graph encoding \sgre{C_n}, such that for all $n \in \N$ and $\ol{x}\in \R^n$ 
    %the feature values of the nodes of $f_{\NN}(\vpl{C_n}(\ol{x}))$ corresponding to output gates in $C_n$ are exactly the components of $f_{C_n}(\ol{x})$ in the same order.
    the feature values of $f_{\NN}(\gre{C_n})$ include $f_{C_n}(\ol{x})$.
\end{restatable}

\begin{proofidea}
The underlying circuit and the halting circuit of the recurrent circuit $C_n$ that is being simulated are encoded into one labelled graph ensuring the following criteria:
the nodes representing halting gates are connected to the nodes representing input gates of the halting circuit and all nodes that represent circuit gates have a distinct degree by the addition of dummy nodes.
%Additional dummy nodes are added and connected to each node that represents a circuit gate to ensure they each have a distinct degree.
In each layer of the recurrent \CGNN the operations of the circuit $C_n$ are now performed stepwise via arithmetic circuit families in the nodes of the labelled graph.
The degree of the nodes serves as an identifier which circuit of the internal circuit family of the \CGNN is executed, e.\,g. which operation is performed such that the feature value of each node is equal to the value of the gate that it takes during the execution of the circuit $C_n$.
The computations of the halting function are simulated in the same way.
By a special construction the halting of the \CGNN only has to check for a predefined value in a fixed number of feature values. 
Notably the halting function is even in ${\faco}[\text{sign}]$.
\end{proofidea}

In the previous result the recurrent \CGNN that simulated the computation of functions from a circuit function class rec$[\mathfrak{F}_s]$-$t\mathfrak{F}[\mathcal{A}]$ realised the additional functions from the set $\mathcal{A}$ in the internal circuits, as those were from the circuit function class $t\mathfrak{F}_s[\mathcal{A}]$ (with the additional sign gate).
Naturally in the GNN model additional functions like those from $\mathcal{A}$ are applied as activation functions only after the computations of a layer are completed. 
To be able to simulate functions by recurrent \CGNN{}s of this form, restrictions on the arithmetic circuit families computing them are necessary.
The following type of circuit form was originally introduced in \cite{DBLP:conf/nips/BarlagHSVV24}. 
We now use and extend this notion under the assumption of balanced DAGs to our model of recurrent circuits.\looseness=-1

% \begin{definition}
%     A circuit $C$ is said to be in \emph{function-layer form} if it is in path-length normal form and for each depth $d \leq \depth(C)$ all gates of $C$ at depth $d$ have the same gate type.
% \end{definition}

% \begin{definition}
%     A recurrent arithmetic circuit with additional function gates that is in function layer form, where the last function layer is at depth $d$, is in \emph{predecessor form} if all halting gates and predecessors of memory gates are at depth $>d$.
% \end{definition}

\begin{definition}
    A recurrent circuit is in \emph{predecessor form} if its underlying circuit $C$ is a balanced DAG, for each depth $d \leq \depth(C)$ all gates of $C$ at depth $d$ have the same gate type and the last function layer is at depth $d$, if all halting gates and predecessors of memory gates are at depth $>d$.
\end{definition}
% We write $p\mathfrak{F}$ to denote only the functions that are computable by circuit families in predecessor form of the circuit function class $\mathfrak{F}$.
We write $p\mathfrak{F}$ to denote only the functions that are computable by $\mathfrak{F}$-circuit families in predecessor form and call classes of the form $p\mathfrak{F}$ \emph{predecessor form classes}.
%\timon{Changed the definition of $\mathfrak{F}$-circuit family so that the above definition should be correct. Please doublecheck.}

% \timon{Depending on our definition of $\mathfrak{F}$-family, could maybe rephrase this to ``functions computable by $\mathfrak{F}$-families in p-form''}
% \vivian{we used the exact same formulation as for the other special classes if we need to change this here we probably need to change it for the others as well}
% \timon{For the other classes, we made a restriction on the function itself (tail-symmetric, symmetric). With the current definition, there might exists a function that is computable by a constant-depth family which is not in predecessor form, but the shallowest family in predecessor form that computes it might have e.g. linear depth. I think we would not want that family to be part of $p\mathrm{FAC}^0$, but with the current definition, it is.}
% \timon{go a little bit handwavy: functions computable by p-form families that adhere to restrictions of F}

Functions computed by arithmetic circuit families of this form can be simulated by \CGNN{}s with outer recurrence that use the additional functions from the set $\mathcal{A}$ only as activation functions.

\begin{restatable}{theorem}{circToOuterRecWithActivationFnc}
    \label{thm:circ_to_outer_rec_gnn_with_activation_fnc}
    Let $\mathcal{C} = (C_n)_{n \in \N}$ be a $\text{rec}[\mathfrak{F}_s]\text{-}p\mathfrak{F}[\mathcal{A}]$-circuit family.
    Then there exists a $\text{rec}[\mathfrak{tF}_s]\text{-}(\mathfrak{tF}_s, \mathcal{A} \cup \{\text{id}\})$-GNN rec-$\NN$ and a symbolic logspace graph encoding \sgre{C_n}, such that for all $n \in \N$ and $\ol{x}\in \R^n$ 
    %the feature values of the nodes of $f_{\NN}(\gre{C_n})$ corresponding to output gates in $C_n$ are exactly the components of $f_{C_n}(\ol{x})$ in the same order.
    the feature values of $f_{\NN}(\gre{C_n})$ include $f_{C_n}(\ol{x})$.
\end{restatable}

\begin{proofidea}
The proof proceeds in the same way as that of Theorem~\ref{thm:circ_to_outer_rec_gnn_no_activation_fnc}.
The predecessor form ensures that during the simulation of the recurrent circuit no results are overwritten by the global application of the activation function.
\end{proofidea}
% \begin{remark}
%     if we have a constant number of iterations this basically corresponds to our function layer form in the neurips paper
% \end{remark}

\vspace{2mm}\noindent\textbf{Simulation with Inner Recurrent C-GNNs}

\noindent To be able to simulate recurrent arithmetic circuits with inner recurrent \CGNN{}s, a different restriction on the functions computed by the circuits is needed.
\begin{definition}
   \label{def:symmetric_fnc}
   A function $f \colon \R^n \to \{0,1\}$ is \emph{symmetric}, if
    \(
        f(v_1 \dots, v_n) = f(\pi(v_1, \dots, v_n))
    \)
    for all permutations~$\pi$.
\end{definition}
We write $s\mathfrak{F}$ to denote only the symmetric functions of the circuit function class $\mathfrak{F}$.
We call classes of the form $s\mathfrak{F}$ \emph{symmetric classes}.

Functions from recurrent function classes of this form can now be simulated by inner recurrent \CGNN{}s by essentially simulating the computations of the recurrent circuit computing them in parallel in each of the nodes of the input graph.
The need for the symmetry arises because computations done in a node of a \CGNN have no order on their inputs, i.\,e.~on the feature values of their neighbouring nodes. 
Notice that due to this method of simulation the underlying circuits of the internal circuits do not need to have access to the sign gate, as was required in the previous results.

\begin{restatable}{theorem}{circToInnerRec}
    \label{thm:circ_to_inner_rec}
    Let $\mathcal{C} = (C_n)_{n \in \N}$ be a $\text{rec}[\mathfrak{F}_s]\text{-}s\mathfrak{F}[\mathcal{A}]$-circuit family.
    Then there exists a $(\text{rec}[\mathfrak{F}_s]\text{-}\mathfrak{tF}[\mathcal{A}], \{\text{id}\})$-GNN rec-$\NN$ and a symbolic logspace graph encoding \sgre{C_n}, such that for all $n \in \N$ and $\ol{x}\in \R^n$ the feature values of $f_{\NN}( \gre{C_n})$ include $f_{C_n}(\ol{x})$.
\end{restatable}
\begin{proofidea}
    For every $n$, a complete %\lauraS{a bipartite graph is not fully-connected, but complete (or biclique)} 
    bipartite labelled graph is constructed, such that it has one set of $n$ nodes with the $n$ input values as feature values and one set of $m$ nodes with feature values 1 to $m$ for the $m$ outputs of the circuit $C_n$.
    The \CGNN consists of one layer where in all nodes the computations of $C_n$ are executed. %\lauraS{I find this sentence to be a bit confusing; which nodes? computations \textbf{of} $C_n$?}
    For the $m$ nodes the inputs are the feature values of their $n$ neighbouring nodes, i.\,e. the input values of $C_n$.
    Each of the $m$ nodes then contains one of the outputs of $C_n$ determined by the initial feature value of the node. 
\end{proofidea}
% \vivian{somehow tie this back to our previous results, to be complete it would be nice to also have a result of a function that an outer recurrent \CGNN can't compute (but that is the GNN result that we couldn't prove yet)}

As in Theorem~\ref{thm:circ_to_outer_rec_gnn_no_activation_fnc} the model of \CGNN used here does not use any activation functions besides the identity function.
Additional functions from the simulated circuit family are again realised directly in the internal circuits.

It can be shown that there exist functions from some recurrent circuit function class without the previously mentioned restrictions that cannot be simulated by a \CGNN with only inner recurrence that simulates the additional functions form $\mathcal{A}$ only as activation functions. 
This is based on the fact that an inner recurrent \CGNN has a fixed number of layers that is determined during construction of the network.
The lemma also holds for any non-continuous activation function.
\begin{restatable}{lemma}{CircToInnerRecWithActivationFnc}
    \label{lem:circ_to_inner_rec_with_activation_fnc}
     There exists a $\text{rec}[\mathfrak{F}_s]\text{-}s\mathfrak{F}[\exp]$-circuit family $\mathcal{C} = (C_n)_{n \in \N}$, with $\exp$ the exponential function, 
    % Then there exists no $(\text{rec}[\mathfrak{F}_s]\text{-}\mathfrak{tF}, \{\exp\})$-GNN rec-$\NN$, such that for all $n\in \N$ and for all logspace encoded graphs $\vpl{C_n}$ the feature values of $f_{\NN}(\vpl{C_n})$ include $f_{C_n}(\ol{x})$.
    where there is no symbolic logspace graph encoding $\sgre{C_n}$ such that there exists a $(\text{rec}[\mathfrak{F}_s]\text{-}\mathfrak{tF}, \{\exp\})$-GNN rec-$\NN$, such that the feature values of $f_{\NN}(\gre{C_n}))$ include $f_{C_n}(\ol{x})$.
    %\timon{Is it necessary that we restrict to a symmetric class?}
\end{restatable}

% \vivian{definition is needed what it means that the number of iterations can (easily) be precomputed or not, goes back to the fixed point halting in the section of the circuit definitions}
% \vivian{but probably also BSS-computability... decide if we want to go down this route or only state that it is difficult}
% \vivian{also need to prove that there is no possible input preserving graph encoding that could do the same thing, like multiplication of all neighbours initialised with the constant $e$. In all those cases the number of iterations $d'$ needs to be known and that can't be computed in the encoding }
% \vivian{is it easier to go via another function than the exponential one, as this can be expressed by constants and unbounded multiplication?}

%\begin{remark}
    The circuit family in Lemma~\ref{lem:circ_to_inner_rec_with_activation_fnc} can be defined to be in predecessor form.
    By Theorem~\ref{thm:circ_to_outer_rec_gnn_with_activation_fnc} there exists an outer recurrent $\CGNN$ that has the output of the circuit family among its final feature vectors.
%\end{remark}

Lemma~\ref{lem:circ_to_inner_rec_with_activation_fnc} also holds for an inner recurrent \CGNN where the underlying circuits of the internal recurrent circuits have access to the sign function.
It follows that there exist functions that an outer recurrent \CGNN can compute that an inner recurrent \CGNN cannot.

% \vivian{For a function with only one input and one output inner and outer recurrent GNNs are identical. Inner recurrence GNNs would be able to compute a function with halting dependent on other gates than the out though.
% Also: outer recurrence GNNs can only do this if the input gates predecessor is the output gate and there are no auxiliary memory gates...
% \begin{lemma}
%     Let $\mathcal{C} = (C_n)_{n \in \N}$ be an $\text{rec}[\mathfrak{F}_s]\text{-}\mathfrak{F}$-circuit family.
%     Where $f_{C_n}\colon \R \to \R$ and halting is only dependent on the value of the output gate. 
%     Then there exist a $(\text{rec}[\mathfrak{tF}_s]\text{-}\mathfrak{tF}, \mathcal{A})$-GNN rec-$\NN$ as well as a a $\text{rec}[\mathfrak{tF}_s]\text{-}(\mathfrak{tF}, \mathcal{A})$-GNN rec-$\NN$', such that for all $n \in \N$ and $x\in \R$ and for graph $G$ with one node
%     the feature value of the node of $f_{\text{rec-}\NN}(G)= f_{\text{rec-}\NN'}(G)$ is exactly $f_{C_n}(x)$.
% \end{lemma}
% could technically get rid of the tail-symmetry here
% \begin{proof}
%     Both recurrent GNNs would only have one layer, for the inner recurrence the given circuit would be assigned in that layer, for outer recurrence only use the underlying circuit and iterate the GNN based on the halting function of the circuit.
% \end{proof}}
\section{Conclusion}
% \timon{is it tail-symmetric or tail symmetric? I think we have used both, should probably get consistent}
In this paper, we introduced a model of recurrent arithmetic circuits and a generalisation of recurrent graph neural networks based on arithmetic circuits. 
We defined different models of recurrence for GNNs and showed correspondence of those to arithmetic circuits. 
However, some restrictions needed to be imposed on the particular circuits used in our constructions.
In particular, the circuits used in our recurrent \CGNN{}s needed to be tail-symmetric.
We established that inner recurrent \CGNN{}s can simulate symmetric circuit functions.
To simulate recurrent circuits by an outer recurrent \CGNN with activation functions (as in Theorem~\ref{thm:circ_to_outer_rec_gnn_with_activation_fnc}), we restricted the circuits to be in the more restrictive \emph{predecessor form}. 

An interesting avenue for further research is the question, whether those restrictions can be relaxed or if it can be proven that they are necessary.

One concept in circuit complexity that we did not touch on is the notion of \emph{uniformity}. 
Without further restrictions, the individual circuits within one circuit family do not need to have anything in common.
This means that in general circuit families do not need to be finitely describable, though that is a desirable property for computational models. 
A circuit family which can be finitely described is called \emph{uniform}, since its circuits need to be of a uniform structure.
It is worth noting that all our proofs are constructive, which means that our circuit families are all uniform in some sense, though the exact nature of that uniformity has yet to be investigated.

A further area of research is to directly study the connections between the different recurrent \CGNN models we introduced and to confirm our conjectures based on our results in relation to arithmetic circuits. 
An interesting factor here is that all models work on labelled graphs that cannot be altered when directly comparing computations.
We conjecture the two models of inner and outer recurrence to be provably incomparable via bisimulation and memory capacity arguments.
Likewise, the study of capabilities between different models is interesting under reasonably simple reductions.

Further investigations are required to obtain practical implications of our theoretical results.
Once we fix the architecture of our recurrent \CGNN{}s (i.\,e., a circuit function class for the GNN nodes), we obtain a specific circuit class that characterises the computational power of the recurrent C-GNN model which can then be rigorously studied. 

\section*{AI Declaration}

The authors have not employed any Generative AI tools.

\bibliographystyle{kr}
\bibliography{references}
\cleardoublepage

\appendix
\section{Appendix}
We use $[n]$ to denote the first $n$ non-zero natural numbers $\{1, \dots, n\}$.
\subsection{Proofs of Section~\ref{sec:rec_arithmetic_circs}}
\begin{lemma}\label{lem:hlt_fnc_equivalence_extra_parameter}
    Circuit halting functions with additional parameter for the halting function $\halt \colon \N_{>0} \times \R^k \to \{0,1\}^*$ and without $\halt \colon \R^{k+1} \to \{0,1\}^*$ are equivalent.
\end{lemma}
\begin{proof}
    Add the orange part to the circuit, where the new addition gate is an additional halting gate. 
    Change the halting circuit such that the first parameter is now the value of this gate, scrap the iteration number.
    % \begin{tikzpicture}[gate/.style={circle, minimum width=0.9cm, draw}]
    %     \node[]   (i)     {$i$};
    %     \node[]     (v1)    [right=0.5 of i]    {$v_1$};
    %     \node[]     (dots)    [right=0.5 of v1]    {$\dots$};
    %     \node[]     (vn)    [right=0.5 of dots]    {$v_k$};
    %     \node[rectangle, draw]  (C) [below= of dots]    {$C_\text{halt}$};
    %     \node[gate]  (o) [below= of C]    {$o$};

    %     \draw[-stealth] (i)     --      (C);
    %     \draw[-stealth] (v1)     --      (C);
    %     \draw[-stealth] (vn)     --      (C);
    %     \draw[-stealth] (C)     --      (o);
    % \end{tikzpicture}
    % \begin{tikzpicture}[gate/.style={circle, minimum width=0.9cm, draw}]
    %     \node[]   (i)     {$\ol{i}$};
    %     \node[]   (a)     [right=0.5 of i] {$\ol{a}$};
    %     \node[rectangle, draw]  (C) [below= of $(i)!0.5!(a)$]    {$C$};

    %     \draw[thick, -stealth] (i)     --      (C);
    %     \draw[thick,-stealth] (a)     --      (C);
    %     \draw[thick,dashed, -stealth] (C)     -|      (i);
    %     \draw[thick,dashed, -stealth] (C)     -|      (a);
    % \end{tikzpicture}
    
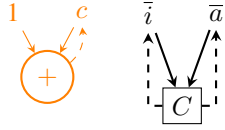
\centering
    \begin{tikzpicture}[gate/.style={circle, minimum width=0.5cm, draw}]
        \node[]   (i)     {$\ol{i}$};
        \node[orange]     (count)     [left=0.5 of i]     {$c$};
        \node[orange]     (1)     [left=0.5 of count]     {1};
        \node[gate, orange, thick]     (add)     [below =0.5 of $(count)!0.5!(1)$]     {$+$};
        \node[]   (a)     [right=0.5 of i] {$\ol{a}$};
        \node[rectangle, draw]  (C) [below= of $(i)!0.5!(a)$]    {$C$};

        \draw[thick, -stealth] (i)     --      (C);
        \draw[thick,-stealth] (a)     --      (C);
        \draw[-stealth, orange] (count)     --    (add);
        \draw[dashed, -stealth, orange] (add)     to[bend right]      (count);
        \draw[-stealth, orange] (1)     --      (add);
        \draw[thick,dashed, -stealth] (C)     -|      (i);
        \draw[thick,dashed, -stealth] (C)     -|      (a);
    \end{tikzpicture}
    
\end{proof}

%\equalityGate*
\begin{remark}[Equality of 2 values]
    \label{rem:eq}
    The equality of 2 values can be computed using a (non-recurrent) circuit of depth 7 facilitating the following formula:
    $x_1 = x_2 \coloneqq ((\text{sign}(x_1-x_2) + \text{sign}(x_2-x_1)) \times -1) +1$
    where each subtraction is computed as $x_1 - x_2 \coloneqq x_1 + (x_2 \times (-1))$. 
\end{remark}
The circuit $C_=$ computing the function can be seen in Figure~\ref{fig:equality}.
        \begin{figure}[tb]
    \centering
    \begin{tikzpicture}
        every node/.style={minimum size =.8cm, outer sep=0pt}, every path/.style={-stealth, ->}
        \tikzset{basic/.style={draw,fill=none,
                       text badly centered,minimum width=2em}}
        \tikzset{gate/.style={basic,circle,minimum width=2em}}
        \tikzset{back/.style={basic, color=blue, ->, dashed}}

        \node[] (name) at (0,0.25) {$C_{=}$:};

        %input and aux memory gates
        \node[] (x1)    at (0,-0.5)    {$\textit{in}_1$};
        \node[] (x2)    at (5,-0.5)    {$\textit{in}_2$};

        %arithmetic gates
         \node[gate] (mul1)  at (1,-2) {$\times$};
         \node[gate] (mul2)  at (4,-2) {$\times$};
         \node[gate] (add1)  at (0,-3){$+$};
         \node[gate] (add2)  at (5,-3) {$+$};
         \node[gate] (sig1)  at (0,-4.5)   {s};
         \node[gate] (sig2)  at (5,-4.5)   {s};
         \node[gate] (add3)  at (2.5,-6)   {$+$};
         \node[gate] (mul3) at (2.5, -7.5)   {$\times$};
         \node[gate] (add4) at (2.5, -9)    {$+$};
        \node[] (out)  at (2.5,-10.5)   {\textit{out}};
        %\node[gate] (add4) at (2,-1.5)    {$+$};
        %\node[gate] (eq1) at (1.5,-3) {$=$};

        %constant gates
         \node[gate] (c1)   at (2,-1)   {-1};
         \node[gate] (c2)   at (3, -1)  {-1};
         \node[gate] (c3)  at (1,-6.5)   {-1};
         \node[gate] (c4)   at (4, -8)   {1};
        % \node[gate] (c4)  at (1.5,-0.5) {-1};

        %wires
        % \draw[->, bend right] (x1) to (add1) {};
        \draw[->] (x1) to (mul1) {};
        \draw[->] (x2) to (mul2) {};
        \draw[->] (c2) to (mul2) {};
        \draw[->] (c1) to (mul1) {};
        \draw[->] (x1) to (add1) {};
        \draw[->] (x2) to (add2) {};
        \draw[->] (mul1) to (add2) {};
        \draw[->] (mul2) to (add1) {};
        \draw[->] (add1) to (sig1) {};
        \draw[->] (add2) to (sig2) {};
        %\draw[->] (c3) to (sig2) {};
        \draw[->] (c3) to (mul3) {};
        \draw[->] (c4) to (add4) {};
        \draw[->] (sig1) to (add3) {};
        \draw[->] (sig2) to (add3) {};
        \draw[->] (add3) to (mul3) {};
        \draw[->] (mul3) to (add4) {};
        \draw[->] (add4) to (out) {};
    \end{tikzpicture}
        \caption{Circuit $C_=$ computing $x_1 = x_2$ for inputs $x_1, x_2$ with depth 7 and size 16.}
        \label{fig:equality}
    \end{figure}
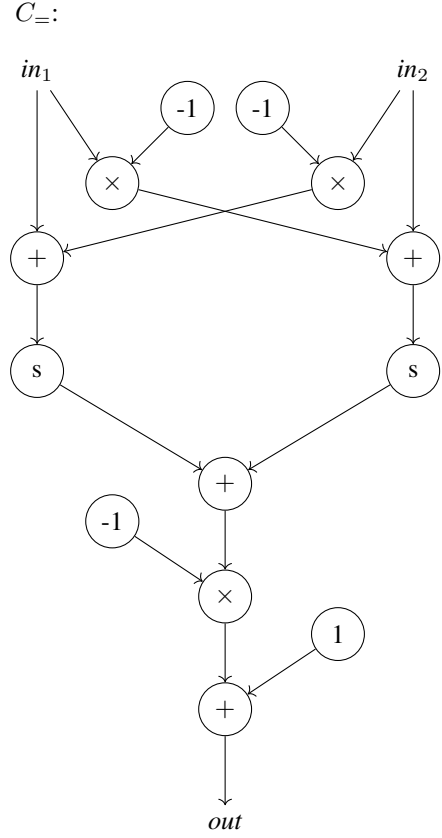

\ConcatRecCirc*

\begin{proof}
    Let $\ell, \ell', n,m,m' \in \N$.
    %Let $\mathcal{C}, \mathcal{C'}$ be the recurrent circuit families that compute $f, f'$.
    Let $\textit{rec-}C = (C,\ol{a} \in \R^\ell, E_\text{rec}, \halt, V_\text{halt})$ and $ \textit{rec-}C' =(C,\ol{a}' \in \R^{\ell'}, E_\text{rec}', \halt, V_\text{halt})$ be the the recurrent circuits with input dimension $n$, output dimension $m$ that compute the functions $f$ and $f'$.
    Let $p = |V_\text{halt}|$ and $ p'=|V_\text{halt}'| $ be the number of halting gates of each circuit.
    The functions ${\halt} \colon \N_{>0} \times \R^p \to \{0,1\}$ and $\halt' \colon \N_{>0} \times \R^{p'} \to \{0,1\}$ are the respective halting functions which are computed by circuits $C_{\text{halt}}$ and $C_{\text{halt}}'$. 
    The goal is to construct a recurrent circuit which computes $f' \circ f \colon \R^n \to \R^{m'}$.

    \begin{figure}[tb]
        \centering
        \begin{tikzpicture}[gate/.style={circle, minimum width=0.6cm, inner sep=0pt, text width=6mm, align=center, draw}, circ/.style={rectangle, minimum width=0.8cm, minimum height=0.6cm, draw}]
            \node[circ]  (Cf1)                                       {$C_\text{halt}^\text{mod}$};
            \node[]                 (in)    [above=0.5cm of Cf1]                      {val$(V_\text{halt})$};
            \node[]                 (it)    [above left=0.5cm of Cf1]                      {$i$};
            \node[circ]     (inv)   [below right=0.5cm of Cf1]          {$C_\text{inv}$};
            \node[circle, draw, minimum width=0.5cm]     (mult)  [right=0.5cm of inv]                      {$\times$};
            \node[circle, draw, minimum width=0.5cm]     (K)     [above=0.5cm of mult]                    {$S$};
            \node[gate, thick]     (invout)    [below=0.5cm of inv]            {${t}^{-1}$};
            \node[gate, thick]     (out)       [below=1.5cm of Cf1]           {$t$};
            \draw[-stealth, thick]  (in)        to              (Cf1.north);
            \draw[-stealth]         (it)        to              (Cf1);
            \draw[-stealth]         (mult)      to [bend right] (K);
            \draw[-stealth]         (K)         to [bend right] (mult);
            \draw[-stealth]         (K)         to [bend right] ($(Cf1.north) + (0.2,0)$);
            \draw[-stealth]         (inv)       to              (mult);
            \draw[-stealth]         (Cf1) -- ++(0,-0.5cm) coordinate(a) -- (out);
            \draw[-stealth]         (a) -- (inv);
            \draw[-stealth]         (inv) -- (invout);
        \end{tikzpicture}
        \caption{Circuit $C_\text{flag}$ that computes the flag whether circuit rec-$C$ has halted, dependent on the current iteration number $i$ and the values of the halting gates val$(V_\text{halt})$.}
        \label{fig:flag_circ}
    \end{figure}
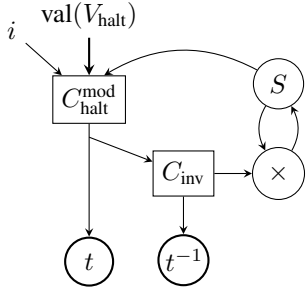
    The circuit $C_\text{flag}$ depicted in Figure~\ref{fig:flag_circ} is used to compute a flag $t$ whether ${\halt} = 1$, i.\,e. the first circuit rec-$C$ has finished its computation.
    It uses a modified version $C_{\text{halt}}^\text{mod}$ of the circuit $C_{\text{halt}}$ that is now also dependent on the value of a new auxiliary memory gate $S$ and extends the computed function by additionally checking if the value of $S$ is unequal to 0.
    A circuit $C_\text{inv}$ is used to invert the binary output of $C_{\text{halt}}$.
    Recurrent edges to and from $S$ are placed in such a way that $C_\text{flag}$ only outputs $t=1$ (respectively ${t}^{-1}=0$) the first time ${\halt} = 1$.

    \begin{figure}[tb]
        \centering
        \begin{tikzpicture}[gate/.style={circle, minimum width=0.5cm, inner sep=0pt, text width=5mm, align=center, draw}, circ/.style={rectangle, minimum width=0.5cm, minimum height=0.8cm, draw}]
            \node[]     (in1)   {$y_1$};
            \node[]     (dots)  [right=0.5cm of in1]    {$\dots$};
            \node[]     (in2)   [right=0.5cm of dots]   {$y_m$};
            \node[]     (f)     [right=0.5cm of in2]    {$t$};
            \node[]     (f1)     [right=0.5cm of f]    {$t^{-1}$};
            \node[gate] (mult1) [below=0.5cm of in1]    {$\times$};
            \node[]     (dots2)  [right=0.5cm of mult1]   {$\dots$};
            \node[gate] (mult2) [below=0.5cm of in2]    {$\times$};
            \node[gate] (add1) [below=0.5cm of mult1]    {$+$};
            \node[]     (dots3)  [right=0.5cm of add1]   {$\dots$};
            \node[gate] (add2) [below=0.5cm of mult2]    {$+$};
            \node[align=left]     (V)     [right =0.5 of f1]      {input gate\\ predecessors};
            \node[]     (out1)   [below=0.5cm of add1]  {$y_1'$};
            \node[]     (dots3)  [right=0.5cm of out1]    {$\dots$};
            \node[]     (out2)   [below=0.5cm of add2]   {$y_n'$};
            % \node[circ]     (C2)    [below = 0.5cm of dots3]  {$C_2$};
            \node[gate] (mult3)  [below = 0.5cm of {$(f1) !0.5! (V)$}]  {$\bm{\times}$};
            \draw[-stealth]     (in1)   to      (mult1);
            \draw[-stealth]     (in2)   to      (mult2);
            \draw[-stealth]     (mult1)   to      (add1);
            \draw[-stealth]     (mult2)   to      (add2);
            \draw[-stealth]     (f)   to      (mult1);
            \draw[-stealth]     (f)   to      (mult2);
            \draw[-stealth]     (add1)   to      (out1);
            \draw[-stealth]     (add2)   to      (out2);
            \draw[-stealth, thick]     (V)   to      (mult3);
            \draw[-stealth]     (f1)   to      (mult3);
            \draw[-stealth]     (mult3)   to      (add1);
            \draw[-stealth]     (mult3)   to      (add2);
            % \draw[-stealth]     (add1)   to      (C2.north);
            % \draw[-stealth]     (add2)   to      (C2.north);
        \end{tikzpicture}
        \caption{Circuit $C_\text{in}$ that sets the input $y_1', \dots, y_m'$ to the second circuit correctly dependent on the flag $t$ and given the output $y_1, \dots, y_m$ of the first circuit}
        \label{fig:input_circ}
    \end{figure}
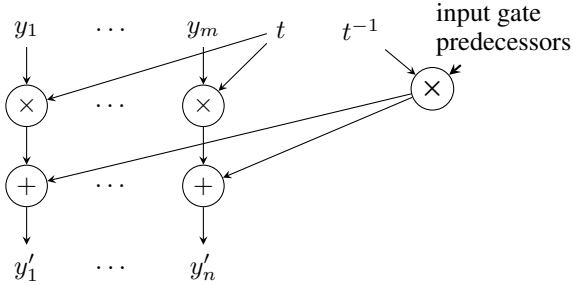
    The circuit $C_\text{in}$ depicted in Figure~\ref{fig:input_circ} sets the input of the circuit rec-$C'$ to the output values of rec-$C$ if and only if $t=1$ which indicates that circuit rec-$C$ has finished its computation.
    This is achieved by multiplying the output $y_1, \dots, y_m$ of rec-$C$ with $t$.
    The value of the predecessors of the original input gates of rec-$C$ is only taken into consideration once $t^{-1}=1$. 

    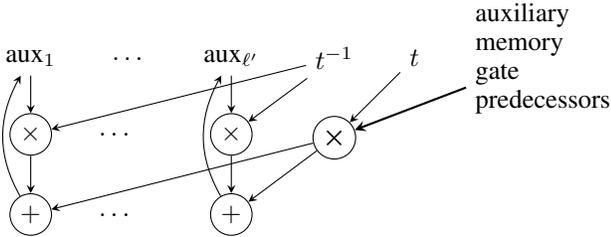
\begin{figure}[tb]
        \centering
        \begin{tikzpicture}[gate/.style={circle, minimum width=0.5cm, inner sep=0pt, text width=5mm, align=center,, draw}, circ/.style={rectangle, minimum width=0.5cm, minimum height=0.8cm, draw}]
            \node[]     (b1)        {$\text{aux}_1$};
            \node[]     (dots)  [right=0.5cm of b1] {$\dots$};
            \node[]     (b2)    [right=0.5cm of dots]   {$\text{aux}_{\ell'}$};
            \node[]     (f1)     [right=0.5cm of b2]    {$t^{-1}$};
            \node[]     (f)     [right=0.5cm of f1]    {$t$};
            \node[gate] (mult1) [below=0.5cm of b1]    {$\times$};
            \node[]     (dots2)  [right=0.5cm of mult1]   {$\dots$};
            \node[gate] (mult2) [below=0.5cm of b2]    {$\times$};
            \node[gate] (add1) [below=0.5cm of mult1]    {$+$};
            \node[]     (dots3)  [right=0.5cm of add1]   {$\dots$};
            \node[gate] (add2) [below=0.5cm of mult2]    {$+$};
            \node[align=left]     (V)     [right =0.5 of f]      {auxiliary\\ memory\\ gate\\ predecessors};
            \node[gate] (mult3) [below=0.5cm of f1]     {$\bm{\times}$};
            \draw[-stealth]     (mult1)   to      (add1);
            \draw[-stealth]     (mult2)   to      (add2);
            \draw[-stealth, thick]     (V)   to      (mult3);
            \draw[-stealth]     (f)   to      (mult3);
            \draw[-stealth]     (mult3)   to  (add1);
            \draw[-stealth]     (mult3)   to  (add2);
            \draw[-stealth]     (b1)   to      (mult1);
            \draw[-stealth]     (b2)   to      (mult2);
            \draw[-stealth]     (add1)   to[bend left]      (b1);
            \draw[-stealth]     (add2)   to[bend left]      (b2);
            \draw[-stealth]     (f1)   to      (mult1);
            \draw[-stealth]     (f1)   to      (mult2);
        \end{tikzpicture}
        \caption{The circuit $C_{\text{aux\_mem}}$ which ensures that the values of the auxiliary memory gates are not changed until the flag $t$ is set}
        \label{fig:circ_aux_mem}
    \end{figure}
    The circuit $C_{\text{aux}}$ depicted in Figure~\ref{fig:circ_aux_mem} updates the values of the auxiliary memory gates $\text{aux}_1, \dots, \text{aux}_{\ell'}$.
    While $t\neq 1$ the originally assigned constant values are fed back to the gates.
    Once $t=1$, i.\,e. circuit rec-$C$ has finished its computation the auxiliary memory gates are set to the values of their predecessors.
    By using those two gadgets rec-$C$ is iterated on arbitrary values until ${\halt}=1$ for the first time.

    \begin{figure}[tb]
        \centering
        \begin{tikzpicture}[gate/.style={circle, minimum width=0.5cm, inner sep=0pt, text width=5mm, align=center, draw}, circ/.style={rectangle, minimum width=0.5cm, minimum height=0.8cm, draw}]
            \node[]     (in)        {input};
            \node[circ] (C1)    [below=0.5cm of in]     {$C$};
            \node[circ] (Cin)   [below=0.5cm of C1]     {$C_\text{in}$};
            \node[circ] (C2)    [below=0.5cm of Cin]    {$C'$};
            \node[]     (out)   [below=0.5cm of C2]     {output};
            \node[circ] (Cmem)  [right=0.5cm of Cin]    {$C_\text{aux\_mem}$};
            \node[circ] (Cflag) [right=0.5cm of C1]     {$C_\text{flag}$};
            \node[]     (i)     [above=0.5cm of Cflag]     {$i$};
            \draw[-stealth] (in)    to      (C1);
            \draw[-stealth] (C1)    to      (Cin);
            \draw[-stealth] (Cin)    to      (C2);
            \draw[-stealth] (Cmem)    to     ($(C2.north)+(0.2, 0)$);
            \draw[-stealth] (C2)    to      (out);
            \draw[-stealth] (Cflag)    to      (Cmem);
            \draw[-stealth] (Cflag)    to      ($(Cin.north)+(0.2, 0)$);
            \draw[-stealth, dotted] (C1)    to      (Cflag);
            \draw[-stealth, dotted] (i)    to      (Cflag);
            \draw[-stealth, dashed]  (C2.west) |- ++(-0.5,0cm) |- (Cin.west);
            \draw[-stealth, dashed]  (C1.west) |- ++(-0.5,0cm) |- (in.west);
            \draw[stealth-, dashed] (Cmem.east)    -- ++(0.5,0cm) |-      (C2.east);
        \end{tikzpicture}
        \caption{The construction of the composed circuit, dashed lines denote values from the respective predecessors, the dotted line denotes the connection from the halting gates}
        \label{fig:concat_circ}
    \end{figure}
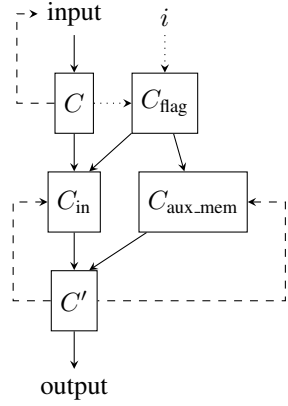
    By combining the previously mentioned gadget circuits and the circuits $C$ and $C'$ we get the composed circuit as depicted in Figure~\ref{fig:concat_circ}.
    The family of underlying circuits is in $\mathfrak{F} \cup \mathfrak{F}_s$
    The new halting function is defined as $\halt^\text{concat} \coloneqq \chi[\text{val}(S)=0]{\halt}'$.
    Therefore, we have $\halt^\text{concat}$ is a subset of a circuit family in $\mathfrak{F}_s$.
    This results in $f^\text{concat}=f'\circ f$ being from a family in the circuit function class $\recF{\mathfrak{F}_s}{\left(\mathfrak{F} \cup \mathfrak{F}_s\right)} = \recF{\mathfrak{F}_s}{\mathfrak{F}_s}$.
\end{proof}

\begin{corollary}
    The recurrent circuit function class $\recF{\mathfrak{F}_s}{\mathfrak{F}_s}$, where both the underlying circuit families and those computing the halting functions are allowed to use the sign gate, is closed under the composition of functions.
\end{corollary}

\subsection{Proofs of Section~\ref{sec:gnn_to_circ}}
\begin{definition}\label{def:labelled_graph_encoding}
    Let $\lgraph = (V, E, g_V)$ be a labelled graph, $n \coloneqq \lvert V \rvert$, and
    % We assume $V$ to be of size $n$ and be ordered by $[ \,n ] \,$.
    $M=\text{adj}(\lgraph)$ be the adjacency matrix of $(V, E)$, where the columns are ordered in accordance with the ordering of $V$.
    We write $\enc{M}$ to denote the encoding of $M$ as the $n^2$ matrix entries $m_{ij}$, ordered in a row wise fashion. 
    We write $\enc{\lgraph}$ to denote the encoding of $\lgraph$ as a tuple of real values, such that $\enc{\lgraph}= \left(\enc{M}, \text{val}(\lgraph)\right) \in \R^{n^2+n}$, which consists of the encoding of $M$ followed by $\text{val}(\lgraph)$, the feature values of $\lgraph$.
    The feature values $\text{val}(\lgraph)$ are $g_V(v)$ for all $v \in V$ and ordered like $V$.
    The reversed procedure is the decoding.
\end{definition}
Note that the decoding could be applied to the output of the recurrent circuits, which is an encoding of a labelled graph, back to the labelled graph itself.

To still be able to check for equality for fixed values we use the following function.
\begin{lemma}[\protect\cite{DBLP:conf/nips/BarlagHSVV24}]
        Let $A \subseteq \R$ be finite and let $a \in A$.
        Then the function $\chi_{A, a} \colon A \to \{0, 1\}$ defined as
        \[
            \chi_{A, a}(x) \coloneqq 
            \begin{cases}
                1, \textnormal{ if $ = a$}\\
                0, \textnormal{ otherwise}
            \end{cases}
        \]
        is computable by an arithmetic circuit whose size only depends on $\lvert A \rvert$ and whose depth is constant.
    \label{lem:chiA}
    \end{lemma}
    \begin{proof}
        Let $A = \{a, a_1, \dots, a_\ell\} \subseteq \R$ be a finite set and let the univariate polynomial $p_{A, a} \colon A \to \R$ be defined as follows:
        \[
            p_{A, a}(x) \coloneqq \left( \prod_{1 \leq i \leq \ell} (a_i - x) \right) \cdot \frac{1}{\prod_{1 \leq i \leq \ell} (a_i - a)},
        \]
        Now, for every element $e \in A \setminus \{a\}$, $p_{A, a}(e)$ yields $0$, since one factor in the left multiplication will be $0$.
        On the other hand, $p_{A, a}(a)$ evaluates to $1$, since the left product yields exactly $\prod_{1 \leq i \leq \ell} (a_i - a)$, which is what it gets divided by on the right. 
        Thus, $p_{A, a} = \chi_{A, a}$, and since for any given $A$ and $a \in A$, $p_{A, a}$ is a polynomial of constant degree, it can be evaluated by a circuit of constant depth and polynomial size in $\lvert A \rvert$.
    \end{proof}

    We will make use of this Lemma in our proofs and refer to the gadget constructed as the equality gadget, as it enables us to check for equality in fixed sets.

\outerrectocirc*
    \begin{figure}
        \centering
        \begin{tikzpicture}
        every node/.style={minimum size =1cm, outer sep=0pt}, every path/.style={-stealth, ->}
        \tikzset{basic/.style={draw,fill=none,
                       text badly centered,minimum width=2em}}
        \tikzset{gate/.style={basic,circle,minimum width=2.3em, inner sep=0pt}}
        \tikzset{circ/.style={basic,circle,minimum width=1em, inner sep=0pt, color=red}}
        \tikzset{back/.style={basic, ->, dashed}}

                    \node[]     (i1)   at (1,-0.5)    {$\textit{in}_1$};
                    \node[]     (dots1) at (2.5,-0.5)    {$\dots$};
                    \node[]     (in)    at (4,-0.5)    {$\textit{in}_n$};
                    \node[]     (m1)    at (6,-0.5)    {$\textit{aux}_1$};
                    %\node[]     (dots2)     [right =0.1 of m1]      {$\dots$};
                    %\node[]     (ml)        [right =0.1 of dots2]     {$\textit{aux}_\ell$};
                    \node[]     (o1)    at (2,-6)  {$\textit{out}_1$};
                    \node[]     (dots3) at (3.5,-6){$\dots$};
                    \node[]     (om)    at (5,-6)  {$\textit{out}_m$};
                    \node[fill=gray!15, rectangle, minimum width =7cm, minimum height=4cm, draw] (rectangle) [below = of $(i1)!0.5!(m1)$] {
                    };
                    \node[rectangle, fill=white, draw, minimum height=2cm, minimum width=0.75cm] (sk1) [below =1.4 of i1] {$\mathcal{C}^1$};
                    \node[]     (dots2)     [below= 2.4 of dots1]      {$\dots$};
                    \node[rectangle, fill=white, draw, minimum height=2cm, minimum width=0.75cm] (skd) [below = 1.4 of in] {$\mathcal{C}^d$};
                    \node[rectangle, fill=white, draw, minimum height=2cm, minimum width=0.75cm] (cnt) [below = 1.4 of m1] {cnt};
                    \node[rectangle, fill=white, draw, minimum height=0.5cm, minimum width = 6cm] (switch) [below = 1.5 of $(sk1)!0.5!(cnt)$] {switch};

                    %\node[circ, color=red] (vh1)   [below=2.75 of i1] {};
                    %\node[circ, color=red] (vhd)   [below=2.75 of in] {};
                    %\node[circ, color=red] (vhc)   [below=2.75 of m1] {};

                    \draw[->]     (i1)      --       (sk1);
                    \draw[->]     (in)      --       (skd);
                    \draw[->]     (i1)      --       (skd.north);
                    \draw[->]     (in)      --       (sk1.north);
                    \draw[->]     (m1)      --       (cnt);
                    \draw[->]     (sk1.south)      --      (sk1.south|-switch.north);
                    \draw[->]     (skd.south)      --      (skd.south|-switch.north);
                    \draw[->]     (cnt.south)      --      (cnt.south|-switch.north);
                    \draw[<-]      (o1.north)  --   (o1.north|-switch.south);
                    \draw[<-]      (om)  --   (om.north|-switch.south);

                    \node[] (right) at (7.5,0) {};
                    \node[] (left)  at (-0.5,0){};
                    \node[] (up)    at (0,0.25) {};
                    \draw[->, back] (cnt.east) -- (right|-cnt.east) -- (right|-up) -- (m1.north|-up)  --  (m1.north);
                    \draw[->, back] (switch.west) -- (left|-switch.west) -- (left|-up) -- (in.north|-up) -- (in.north);
                    \draw[->, back] (i1.north|-up) -- (i1.north);

    \end{tikzpicture}
    \caption{Schema of the recurrent circuit $C_{n}^{\mathcal{N}}$ where recurrent edges are dashed.}
        \label{fig:outerreccgnntoreccirc}
    \end{figure}
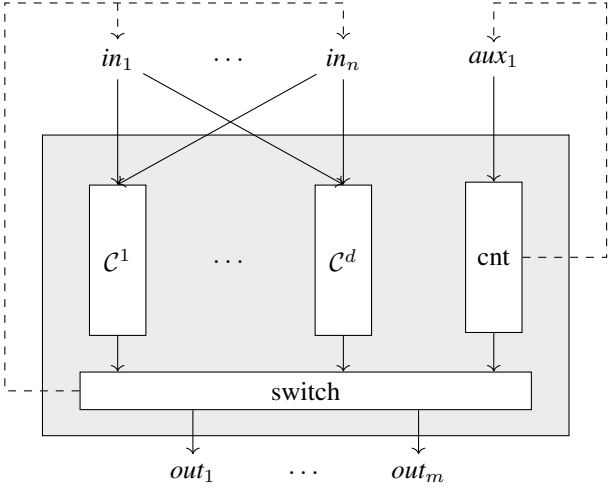
\begin{proof}
    Let %$\mathfrak{G}=(V,E, g_V)$ be the input graph, 
    $B=({t\mathfrak{F}}) \times \mathcal{A}\times \{\haltGNN\}$ be a recurrent \CGNN basis and $\mathcal{N}_\textit{rec} = (\mathcal{N}, \haltGNN)$ a $\text{rec}[\mathfrak{F}_s]\text{-}(t\mathfrak{F}. \mathcal{A})\text{-GNN}$ with halting function $\haltGNN$ and $\mathcal{N} \colon \N \to ({t\mathfrak{F}}) \times \mathcal{A}$ a periodic function with period length $d\in\N$.
    Then the image of \NN consists of $d$ possibly distinct arithmetic circuit families from $t\mathfrak{F}$ and activation functions from $A$. Let $\mathcal{C}^{\mathcal{N}(\ell)}$ for $\ell \in [d]$ be those circuit families composed with the accompanying activation function gate $a \in \mathcal{A}$.
    For any size $n$ of a labelled graph $\mathfrak{G}=(V,E, g_V)$ (or, more precisely, for any encoding size $n$ of $\mathfrak{G}$), we create one circuit $C^{\mathcal{N}}_n$, such that the claim holds.
    Since the graph structure does not change, we essentially only need to compute the updates of the feature values.

    We construct $C^{\mathcal{N}}_n$ as follows:
    For every family $\mathcal{C}^{\mathcal{N}(\ell)}$ of \NN, we have a subcircuit in parallel as shown in Figure~\ref{fig:outerreccgnntoreccirc}.% where we add an additional gate at the end of every circuit family that computes the activation function $a \in \mathcal{A}$ that is used in layer $\ell$. 
    In each of these subcircuits, we have $m=|V|$ copies of the circuits of that family again in parallel, one for each possible arity of neighbourhood relation, i.\,e. $C^{\mathcal{N}(i)}_1$ to $C^{\mathcal{N}(i)}_{|V|}$.
    For every node $v \in V$, we now want to compute the updated feature value that $\recNN$ would compute in layer $\ell$.
    We use the same construction as in \cite{DBLP:conf/nips/BarlagHSVV24} to essentially choose the correct arity circuit, where we reorder the inputs such that the feature values of the neighbours of $v$ are in the front and the remaining feature values get set to 0. 
    
    %\laura{Define the connections of feature vectors to each circuit based on the graph (i.\,e. based on the neighbourhood relations). (same as NeurIPS paper).}
    In addition to that, we have a modulus $d$ counter, that counts the current layer $\ell$ of the simulation of $\recNN$ starting at 1 and making use of Lemma~\ref{lem:chiA} to reset the counter every $d$ steps.
    Using this counter we build a switching gadget, such that we use the results (i.\,e. the updated feature values for all $m$ nodes of the input graph) of the correct circuit $C^{\mathcal{N}(\ell)}_i$ to feed back into the next iteration of the circuit or, if the circuit halts, to the output gates as sketched in Figure \ref{fig:switch}. 
    This construction again makes use of the equality gadget described in Lemma~\ref{lem:chiA}. Please note that we cannot use the construction described in Remark~\ref{rem:eq}, as this uses sign gates.
    
    The halting function of $C^{\mathcal{N}}_n$ is essentially the same as for the recurrent \CGNN: it is computed on the updated feature values for every node which are determined in the switching gadget and the current iteration number of the circuit.
    For an encoded input graph with $m$ nodes, we use the function from $\haltGNN$ with arity $m+1$.  
    The functions computed by the resulting sequence of recurrent circuits then make up $f_\textnormal{circ}$. 
    %\laura{halting function from GNN - need to choose right arity for subcircuit used - no, its always number of nodes of G}
    %\laura{argue why this is in line with size/depth constraints}

    It is left to show that the construction does not exceed the depth and size constraints of the circuit function classes used in $\recNN$.
    For depth, the reordering of the feature values and the switching gadget each take constant depth.
    The circuit families for every layer take polylogarithmic depth and as they are in parallel; the construction overall takes polylogarithmic depth.
    For size, it takes a linear overhead to simulate each individual circuit family (i.\,e. the circuit family for one layer of $\recNN$). 
    This is done again in parallel for a constant number of families.
    So overall, the size stays polynomial.
        \begin{figure}
        \centering
        \begin{tikzpicture}
            every node/.style={minimum size =1cm, outer sep=0pt}, every path/.style={-stealth, ->}
        \tikzset{basic/.style={draw,fill=none,
                       text badly centered,minimum width=2em}}
        \tikzset{gate/.style={basic,circle,minimum width=2.3em, inner sep=0pt}}
        \tikzset{circ/.style={basic,circle,minimum width=1em, inner sep=0pt, color=red}}
        \tikzset{back/.style={basic, color=blue, ->, dashed}}
        \node[] (ci)    at (0.5,0)    {$\mathcal{C}^{\mathcal{N}(\ell)}$};
        \node[] (cnt)   at (2,0)    {cnt};
        \node[gate] (mul)   at (1,-2)   {$\times$};
        \node[gate] (eq)  at (1.5,-1) {$=\ell$};
        \node[gate] (add)   at (1,-3.5) {$+$};
        \node[gate] (out)   at (1,-5)   {out};

        %other circuits
        \node[] (cnt2) at (-0.5,0)  {cnt};
        \node[] (c-1)   at (-2,0) {$\mathcal{C}^{\mathcal{N}(\ell-1)}$};
        \node[] (c+1)   at (4,0) {$\mathcal{C}^{\mathcal{N}(\ell+1)}$};
        \node[gate] (ch) at (-1.5, -2) {$\times$};
        \node[gate] (cj) at (3.5,-2)   {$\times$};
        \node[] (dot1) at (-2.5,-2) {$\dots$};
        \node[] (dot2) at (4.5,-2) {$\dots$};
        \node[gate] (eq2) at (3,-1)  {$=\ell$};
        \node[gate] (eq3) at (-1,-1) {$=\ell$};

        \draw[->] (ch) -- (add);
        \draw[->] (cj) -- (add);
        \draw[->] (cnt) -- (eq2);
        \draw[->] (eq2) -- (cj);
        \draw[->] (eq3) -- (ch);
        \draw[->] (c-1) -- (ch);
        \draw[->] (c+1) -- (cj);
        \draw[->] (cnt2) -- (eq3);

        \draw[->] (ci) -- (mul);
        \draw[->] (cnt) -- (eq);
        \draw[->] (eq) -- (mul);
        \draw[->] (mul) -- (add);
        \draw[->] (add) -- (out);
        \end{tikzpicture}
        \caption{Depiction of switching gadget for layer $\ell$. The inputs to this circuit are the outputs of the counting subcircuit (cnt) and of the subcircuits for each circuit family in $\mathcal{N}$ ($\mathcal{C}^{\mathcal{N}(\ell)}$). The gates labelled $=\ell$ are shorthand for the circuit computing equality with constant $\ell$.}
        \label{fig:switch}
    \end{figure}
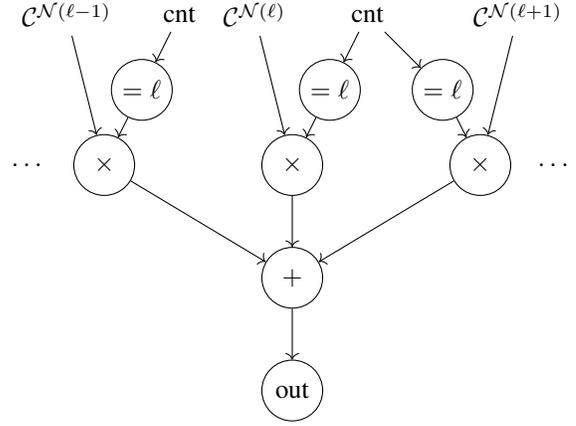
\end{proof}
\InnerRecToCirc*

\begin{proof}
    Let $B=(\recF{t\mathfrak{F_s}}{t\mathfrak{F}}) \times \mathcal{A}\times \{\haltGNN\}$ be a recurrent \CGNN basis and $\mathcal{N}_\textit{rec} = (\mathcal{N}, \haltGNN)$ with $\mathcal{N} \colon \N \to (\recF{t\mathfrak{F_s}}{t\mathfrak{F}}) \times \mathcal{A}$ a periodic function with period length $d\in\N$ and $\haltGNN(i, D) \coloneqq \chi[i=d]$ a $(\text{rec}[\mathfrak{F}_s]\text{-}\mathfrak{F})\text{-GNN}$. 
    W.l.o.g. let the depth of $\mathfrak{F}$-circuit families be bounded by $\bO((\log n)^i)$ and let their size be bounded by $\bO(n^{\bO(1)})$.
    For any size $n$ of a labelled graph $\mathfrak{G}$ (or, more precisely, for any encoding size $n$ of $\mathfrak{G}$), we create one circuit $C^{\mathcal{N}}_n$, such that the claim holds.
    
    % We construct a circuit that computes the same function as $\mathcal{N}$ by concatenating the recurrent circuits of the the $d$ layers of the \CGNN.
    For each layer of $\mathcal{N}_\textit{rec}$, we create a circuit $C$ that simulates that layer and we afterwards concatenate those circuits into $C^{\mathcal{N}}_n$, making use of Theorem~\ref{thm:concat_rec_circ}.
    % We construct a circuit $C$ that simulates a single layer of $\mathcal{N}_\textit{rec}$ for every layer and then concatenate  those circuits, making use of Theorem~\ref{thm:concat_rec_circ}.
    
    For $\ell \in \N$, let $\mathcal{N}(\ell) = (C^{(\ell)}, \mathcal{A})$, where $C^{(\ell)}$ computes the function $f_\ell \in \recF{t\mathfrak{F_s}}{t\mathfrak{F}}$ via the recurrent circuit family $\mathcal{C}^{(\ell)} = (C_n)_{n \in \N}^{(\ell)}$.
    % To simulate the layer $\ell$, $C$ first determines the neighbours of each node $g$ of its encoded input graph and then simulates the underlying circuit of $(C^\ell_{n+1})$, where $g$ has $n$ neighbours.
    To simulate the layer $\ell$, for each node $v$ of its encoded input graph, $C$ first determines the neighbours of $v$ and then simulates the underlying circuit of $C^{(\ell)}_{n_v+1}$, where $n_v$ is the number of neighbours of $v$.
    Afterwards, $C$ simulates the halting function of $C^{(\ell)}_{n_v + 1}$ (on the halting set of the simulation of $C^{(\ell)}_{n_v + 1}$ and the iteration number (initially $1$)).
    Both, the output gates of the simulation of $C^{(\ell)}_{n_v + 1}$ and of its halting function, have recurrent edges to additional auxiliary memory gates in $C$.
    If the latter have value $1$, then instead of doing any computation, the values of the former immediately overwrite the aforementioned output gates to ensure that the output does not change after the halting condition is met. 
    %\vivianS{this isn't clear, why the auxiliary memory gates, what is the idea behind it T: rephrased the previous sentence a bit.} 
    Furthermore, the simulation of the halting function is then just replaced by a constant $1$, to make sure that this value no longer changes, either. %\vivianS{maybe mention that this is still in the one layer T: Have not found a good way to do that.}
    %\vivianS{this is necessary because some node might be done with their computation earlier than others? T: Yes, added a sentence.}
    Ensuring that the output does not change after the halting condition is met is necessary, since the computations in the different nodes in the input graph might have different recursion depth.
    
    The halting set of the entire circuit $C$ consists only of a single multiplication gate, that gets the output values of all the halting function simulations as predecessors, so that it only has the value $1$ after all halting function simulations are outputting $1$.
    The halting function of $C$ is then simply the function $f(n, x) = \mathrm{sign}(x)$. 
    The resulting values of the simulations of the $C^{(\ell)}_{n_v + 1}$ for all nodes of the labelled input graph are taken as the new values for the respective nodes after applying the activation function. 

    This results in a recurrent circuit $C$, where the respective recurrent circuit of the recurrent C-GNN is simulated for each node until its halting condition is met, after which the value no longer changes.
    Since each node in the input graph receives the thus computed value as its new value, the output of $C$ is exactly the encoding of the labelled input graph with the labels updated according to one layer of $\mathcal{N}_\textit{rec}$.

    The circuit $C$ needs to simulate the halting functions of the recurrent circuits of $\mathcal{N}_\textit{rec}$.
    For that it needs access to gates for activation functions, which is why it needs to be extended by $\mathrm{sign}$ and $\mathcal{A}$.

    The simulation of $C^{(\ell)}_{n_v+1}$ for any one node requires polynomial size and polylogarithmic depth in $n_v$.
    Furthermore, the halting set of $C^{(\ell)}_{n_v+1}$ is at most polynomially large in $n_v$, thus the size of the simulation of the halting function of $C^{(\ell)}_{n_v+1}$ is a polynomial in a polynomial, which remains polynomial, as $(n^c)^k \in \bO(n^{\bO(1)})$ for all fixed $c, k \in \N$.
    Additionally, since $\log(n^c)^i = (c \cdot \log n)^i \in \bO((\log n)^i)$ for all fixed $c \in \N$, the depth of the halting function simulation is also bounded by $\bO((\log n)^i)$.
    Finally, if $n$ is the number of nodes in $\mathfrak{G}$, then we perform $n$ such simulations in parallel.
    This incurs a linear overhead in size, which thus stays polynomial.
    The overhead in depth is only constant, which means that in total, the size of $C$ is bounded by $\bO(n^{\bO(1)})$ and the depth of $C$ is bounded by $\bO((\log n)^i)$.

    Since $\mathcal{N}_\textit{rec}$ is a $(\textnormal{rec}[t\mathfrak{F}_s]$-$t\mathfrak{F}, \mathcal{A})$-GNN, we know that it only has a constant number of layers, and we can therefore create a circuit of the above form for each such layer and concatenate them as per Theorem~\ref{thm:concat_rec_circ}.

    % \timon{Argue about size and depth. But need to settle on circuit function classes first.}
    
    % \timon{Simulating a single $\CGNN$ layer is somewhat nontrivial, but should work.}
    % This behaviour is shown in the proof of Theorem \ref{thm:concat_rec_circ}.
    % As the number of layers is fixed we only have a constant factor on the depth of the resulting circuit. \vivianS{argue about size, maybe do this in the concat proof, as it increases by the size of the circuit of the halting function as well}
    % \vivian{todo: describe the technicalities of encoding the graph, choosing the correct circuit from the family etc. like in our NeurIPs paper}
\end{proof}

\DoubleRecToCirc*

\begin{proof}
    The proof combines the techniques of the proofs of Theorems~\ref{thm:recCGNN_in_circ} and \ref{thm:CrecGNN_in_circ} and constructs a circuit family $\mathcal{C}$ which computes the function $f_\textnormal{circ}$. 
    We assume the same simplifications as before.
    Let $d \in \N$ be the period length of $\textit{rec-}\mathcal{N}$.
    Similar to the proof of Theorem~\ref{thm:recCGNN_in_circ} the circuit families $\mathcal{C}^{(\ell)}$ for $\ell \in [d]$ simulating layers of $\textit{rec-}\mathcal{N}$ are put in parallel with a switching gadget that decides based on a counter value which circuit family is to be executed.
    As the internal circuit families of $\textit{rec-}\mathcal{N}$ are now recurrent we additionally need to incorporate the techniques from the proof of Theorem~\ref{thm:CrecGNN_in_circ} in the construction of the circuit families $\mathcal{C}^{(\ell)}$.
    The counter of $\mathcal{C}$ introduced in the proof of Theorem~\ref{thm:recCGNN_in_circ} is extended in such a way that it is only increased once the recurrent circuit family of the currently simulated layer has halted, i.\,e. it is dependent on the halting function of the recurrent circuit family of that layer.
    We also use the gadgets from the composition result Theorem~\ref{thm:concat_rec_circ} to ensure that the auxiliary memory gates of the circuit families $\mathcal{C}^{(\ell)}$ are correctly set at the right point in time.
    The same is done for the input gates. 
    As we use those gadgets the function computed by the underlying circuit is in $\mathfrak{F}_s$ as well.
    The circuit families $\mathcal{C}^{(\ell)}$ of the different layers are connected correctly such that they can either be executed again with the correct input (whilst their halting condition is not met) or their output is fed to the circuit families $\mathcal{C}^{(\ell+1)}$ of the next layer.
    The new halting function is constructed as in the proof of Theorem~\ref{thm:recCGNN_in_circ} and therefore is in the circuit function class $\mathfrak{F}_s$.
    The procedure is visualised in Figure~\ref{fig:double_rec_to_circ_proof};

    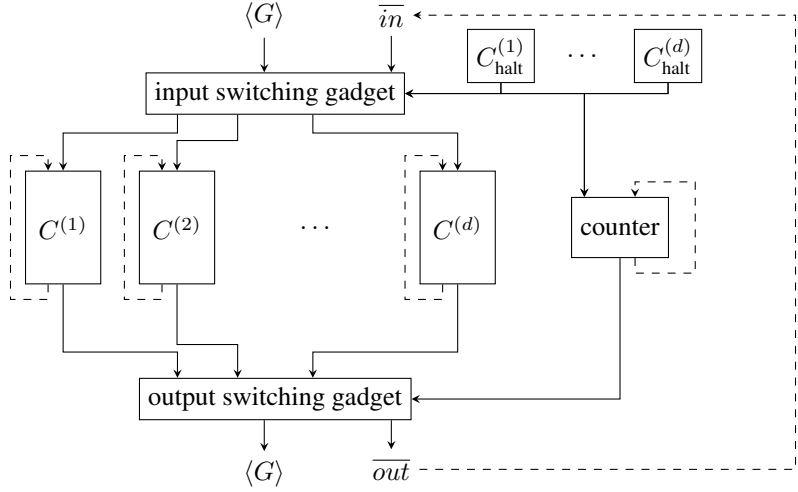
\begin{figure*}[tb]
    \centering
        \begin{tikzpicture}[block/.style={rectangle, draw, minimum height=1.5cm, minimum width=1cm}]
            \node[]             (G)                                     {$\enc{G}$};
            \node[]             (in)        [right=of G]                {$\ol{in}$};
            \node[block]        (C1)        [below left=2.5 of G]       {$C^{(1)}$};
            \node[block]        (C2)        [right= 0.5 cm of C1]       {$C^{(2)}$};
            \node[]             (dots)      [right = of C2]             {$\dots$};
            \node[block]        (Cd)        [right=of dots]             {$C^{(d)}$};
            \node[rectangle, draw, minimum height=0.8cm]  (cnt)       [right=of Cd]   {counter};
            \node[rectangle, draw]  (switch)        [below = 2cm of $(C2)!0.7!(dots)$]      {output switching gadget};  
            \node[rectangle, draw]  (INswitch)      [above = 3.5cm of switch]               {input switching gadget};  
            \node[]     (G')    [below=5.5cm of G]      {$\enc{G}$};
            \node[]     (out)   [below=5.5cm of in]     {$\ol{out}$};  
    
            \node[rectangle, draw]      (h1)        at ($(INswitch)+(3,0.5)$)           {$C_\text{halt}^{(1)}$};
            \node[]                     (dots2)     [right=0.3 of h1]                   {$\dots$};
            \node[rectangle, draw]      (hd)        [right=0.3 of dots2]                {$C_\text{halt}^{(d)}$};

            \draw[-stealth]    (switch.south-|G'.north)    --      (G'.north);
            \draw[-stealth]    (switch.south-|out.north)    --      (out.north);
            \draw[-stealth]     (cnt.south)      |-         (switch.east);
            \draw[-stealth]     (Cd.south)  -- ($(Cd.south)+(0,-0.9)$)      -| ($(switch.north)+(0.5,0)$); 
            \draw[-stealth]     (C2.south)  -- ($(C2.south)+(0,-0.8)$)      -| ($(switch.north)+(-0.5,0)$);
            \draw[-stealth]     (C1.south)  -- ($(C1.south)+(0,-0.9)$)      -| ($(switch.north)+(-1.3,0)$);

            \draw[-stealth]     (G.south)  --  (INswitch.north-|G.south);
            \draw[-stealth]     (in.south)  --  (INswitch.north-|in.south);
            \draw[stealth-]     (Cd.north)  -- ($(Cd.north)+(0,0.5)$)      -| ($(INswitch.south)+(0.5,0)$); 
            \draw[stealth-]     (C2.north)  -- ($(C2.north)+(0,0.4)$)      -| ($(INswitch.south)+(-0.5,0)$);
            \draw[stealth-]     (C1.north)  -- ($(C1.north)+(0,0.5)$)      -| ($(INswitch.south)+(-1.3,0)$);
    
            \draw[-stealth, dashed] ($(C1.south)+(-0.2,0)$)     |-   ++(-0.5,-0.2)    |-  ($(C1.north)+(-0.2,0.2)$) --($(C1.north)+(-0.2,0)$);
            \draw[-stealth, dashed] ($(C2.south)+(-0.2,0)$)     |-   ++(-0.5,-0.2)    |-  ($(C2.north)+(-0.2,0.2)$) --($(C2.north)+(-0.2,0)$);
            \draw[-stealth, dashed] ($(Cd.south)+(-0.2,0)$)     |-   ++(-0.5,-0.2)    |-  ($(Cd.north)+(-0.2,0.2)$) --($(Cd.north)+(-0.2,0)$);
            \draw[-stealth, dashed] ($(cnt.south)+(0.2,0)$)     |-   ++(0.8,-0.2)    |-  ($(cnt.north)+(0.2,0.2)$) --($(cnt.north)+(0.2,0)$);
    
            \draw[-stealth, dashed] (out.east)     |-   ++(5,0)    |-  ($(in.east)+(5,0)$) -- (in.east);
    
            \draw[-stealth]         (h1.south)      |-      (INswitch.east);
            \draw[-stealth]         (hd.south)      |-      (INswitch.east);
            \draw[-stealth]         (h1.south)      |- ($(dots2)+(0,-0.5)$.south) -- (dots2.south|-cnt.north);
            \draw[-stealth]         (hd.south)      |- ($(dots2)+(0,-0.5)$.south) -- (dots2.south|-cnt.north);
        \end{tikzpicture}
        \caption{Proof of Corollary~\ref{cor:CrecrecGNN_in_circ}: Circuit simulating a \CGNN with both methods of recurrence, $\enc{G}$ is the encoded input graph of the \CGNN, $\ol{in}$ the ordered vector of input values and $\ol{out}$ the ordered vector of output values, $C^{(i)}$ are the circuits that simulate each of the $d$ layers of the \CGNN and $C_\text{halt}^{(i)}$ are the circuits computing the respective halting functions.}
        \label{fig:double_rec_to_circ_proof}
    \end{figure*}
\end{proof}

\subsection{Proofs of Section~\ref{sec:circ_to_gnn}}
\subsubsection{Encoding}\label{sec:encoding}
%\todo[inline]{Makro $\backslash{}$sgre for symbolic labelled graph $\sgre{C}$ and $\backslash{}$gre for non symbolic $\gre{C}$}
Our goal is to construct recurrent \CGNN{}s which are able to simulate the computations of recurrent circuits, i.\,e.~have the output of a given circuit among the feature vectors of the output labelled graph when given a labelled graph as input.
Those input labelled graphs are constructed given the recurrent arithmetic circuit. 
We therefore need to fix the notion of a graph encoding.
We restrict the complexity to logspace Turing computability to ensure that our encodings are not to powerful.
As recurrent arithmetic circuits are usually defined via their tuple as described in Definition~\ref{def:rec_circ} and then given an input, we define a so called \emph{symbolic logspace labelled graph encoding} $\sgre{C}$ first which, given a recurrent circuit $C$ as input, uses symbols for the input, auxiliary memory and constant gates instead of the values from $\R$ as node labels. 
This ensures that our encoding can be computed by a Turing machine which only works over a finite alphabet.
The encoding consists of three components: the encoding of the given recurrent arithmetic circuit into an input for a Turing machine, the actual transformation of this encoded circuit into an encoded symbolic labelled graph and the decoding of this into a symbolic labelled graph.

\begin{definition}
    We write $\mathfrak{Circ}$ to denote the class of all recurrent arithmetic circuits.
\end{definition}

\begin{definition}
    Let $C$ be a recurrent arithmetic circuit with additional function gates from a set $\mathcal{A}$ and w.\,l.\,o.\,g. let the gates of $C$ (i,\,e. those of the underlying circuit and those of the circuit computing the halting function) be ordered.
    We then write $\langle C \rangle$ to denote a fixed binary encoding $\langle \cdot \rangle \colon \mathfrak{Circ} \to \mathbb{B}*$ of $C$, which includes the encoding of the underlying circuit and the circuit computing the halting function, where constant, input and auxiliary memory gates are encoded by their index in the gate ordering instead of their (real) values.
    Additional function gates from $\mathcal{A}$ are encoded analogously.
    The other gates are encoded b their input while also encoding their respective gate type.
    %\vivian{this also does not describe what from the circuit is given to the encoding, i.\,e. adjacency matrices, halting gates etc.}
\end{definition}

The goal is to construct what we call a \emph{symbolically labelled graph}, a version of a labelled graph where the labels are from a fixed predefined set of symbols.
As we are interested in encoding a recurrent circuit as such we directly define it with symbols representing input, auxiliary memory and constant gates.
\begin{definition} \label{def:binary_enc_circ}
    Let $(V, E)$ be a graph with $V$ being ordered, let $n, \ell, k\leq \lvert V \rvert$ and let
    \begin{alignat*}{2}
        g_V \colon V \to  \, &\{\textit{in}_i \mid i \in [n]\} &&\cup \\ % \cup \{\textit{out}_i \mid i \in [m]\} \cup \phantom{a} \\
        & \{\textit{aux}_i \mid i \in [\ell]\} &&\cup\\
        & \{\textit{const}_i \mid i \in [k]\} &&\cup \\
        & \{x \in \N \mid x \leq \lvert V \rvert \}
        % & \{0\} \cup \left[\lvert V \rvert\right]
        % & \{f_i \mid i \in [\ell]\} \cup \{+, \times\} \cup \{\textit{dummy}_0, \textit{dummy}_1\}
    \end{alignat*}
    %\timon{once this is fine, fix linebreaks}
    be a function which labels the nodes in $V$ with symbols.
    We then call $G = (V, E, g_V)$ a \emph{symbolically labelled graph}.

    % Let furthermore $C$ be a recurrent arithmetic circuit with ordered gates and $n$ input gates, $m$ output gates, $a$ auxiliary memory gates, $k$ constant gates and $\ell$ gates for additional functions and let $\ol{x} \in \R^n$. 
    % We then additionally write $G_{C, \ol{x}}$ to denote the labelled graph $(V, E, g_V')$ where $g_V'$ maps 
    % \begin{itemize}
    %     \item the nodes labelled $\textit{const}_i$ or $\textit{aux}_i$ to the respective initial constant values of the $i$th constant (resp. auxiliary memory) gate, 
    %     \item the nodes labelled $\textit{in}_i$ to $x_i$,
    %     \item the nodes labelled $\textit{dummy}_i$ to $i$
    %     \item and all remaining nodes to the index of their corresponding gate in the gate ordering of $C$. \vivianS{we need the gate types of the other gates, i.\,e. +,mult, function}
    % \end{itemize}
    % \timon{dont technically need out and $f$ for this definition}
\end{definition}

\begin{definition}
    %\timon{keep this?}
    We write $\mathfrak{sGraph}$ to denote the class of all symbolically labelled graphs.
    
\end{definition}

We want to define the complexity of the encoding with respect to a logspace transducer that works in binary and therefore outputs a binary encoding of a symbolically labelled graph, when given the binary encoding of a recurrent circuit as an input. 
Hence we need to define a notion of decoding a symbolically labelled graph from a binary encoding.
This process is similar to that of Definition~\ref{def:labelled_graph_encoding}, where we described the real valued encoding (and decoding) of a labelled graph with labels from $\R$ but changed such that the symbols of the labels are encoded in binary as well.
\begin{definition}\label{def:decoding_symb_labelled_graph}
    % Let $G$ be a symbolically labelled graph.
    % We write $\textit{enc}(G)$ to denote a fixed binary encoding of $G$.
    % We furthermore write $\textit{dec}$ for the respective decoding, such that $\textit{dec}(\textit{enc}(G)) = G$ for all symbolically labelled graphs $G$.

    Let $\textit{enc} \colon \mathfrak{sGraph} \to \mathbb{B}^*$ be a fixed binary encoding function of symbolically labelled graphs.
    Let furthermore $\textit{dec}$ be the respective decoding function, i.\,e., $\textit{dec} = \textit{enc}^{-1}$.
\end{definition}

In order to meaningfully talk about logarithmic space computation, we make use of \emph{logspace transducers} \cite[Chapter 5]{DBLP:series/txcs/Kozen06}, which are Turing machines that use only logarithmic space and output a solution bit by bit.%\timonS{maybe find a reference for logspace transducers?}

With that at hand, we are now able to turn to our desired logspace encodings.
% Given the binary encoding of a recurrent circuit from Definition~\ref{def:binary_enc_circ} and the decoding of a symbolically labelled graph from binary from Definition~\ref{def:decoding_symb_labelled_graph} we now define the complete procedure of encoding of a recurrent arithmetic circuit as a symbolically labelled graph. 

\begin{definition}
    We say that a function $\mathfrak{G} \colon \mathfrak{Circ} \to \mathfrak{sGraph}$ is a \emph{symbolic logspace graph encoding}, if there exists a logspace transducer computing a function $f \colon \mathbb{B}^* \to \mathbb{B}^*$, such that
    \[
        \sgre{C} = \textit{dec}(f(\langle C \rangle)).
    \] %\lauraS{say that $C$ is a circuit? And that $\mathfrak{G}(C)$ is a slg?}
    %\timonS{Should be clear from the function signature, I'd say}
\end{definition}
We now extend the definition of a symbolic logspace encoding and define the notion of a logspace encoding, meaning that the previously used symbols are replaced by actual values from $\R$.
This procedure is done for all nodes labelled with symbols.
Note that those may be nodes corresponding to gates from the underlying circuit as well as the halting circuit.
\begin{definition}
    Let $\mathfrak{G} \colon \mathfrak{Circ} \to \mathfrak{sGraph}$ be a symbolic logspace encoding.
    % We then define the function $\mathfrak{G'} \colon \mathfrak{Circ} \times \R^* \to \mathfrak{Graph}$ such that $\mathfrak{G'}(C, \ol{x}) = G_{C, \ol{x}}$, where $G = \mathfrak{G}(C)$.
    We then define a \emph{logspace graph encoding} as the function $\mathfrak{G'} \colon \mathfrak{Circ} \times \R^* \to \mathfrak{Graph}$ as
    \[
        \gre{C} = (V, E, g_V'),
    \]
    where $\sgre{C} = (V, E, g_V)$ for some symbolic labeling $g_V$ and where $g_V'$ maps 
    \begin{itemize}
        \item the nodes labelled $\textit{const}_i$ or $\textit{aux}_i$ by $g_V$ to the respective initial constant values of the $i$th constant (resp. $i$th auxiliary memory) gate in $C$ for all $i \in \N$, 
        \item the nodes labelled $\textit{in}_i$ by $g_V$ to $x_i$ for all $i \in \N$ and
        \item the nodes labelled $y \in \N$ by $g_V$ to $y$. %\vivianS{or for clarity write that the rest of the labels are not changed?}
        % \item and all remaining nodes to the index of their corresponding gate in the gate ordering of $C$. \vivianS{we need the gate types of the other gates, i.\,e. +,mult, function}
    \end{itemize}
    %\timon{mention that this goes for the underlying circ and the halting circ}
    
    % \timon{todo: try to make this notation nicer, possibly offload $G_{C, \ol{x}}$ stuff here? Anyways fix $G_{C, \ol{x}}$ stuff to work with our proofs}
    % \timon{Overloading $\mathfrak{G}$ made my toenails curl up, so I added a '.}
\end{definition}

\subsubsection{Simulation with Outer Recurrent \CGNN{}s}
\circToOuterRecNoActivationFnc*
\begin{proof}
    Let $C_n$ be a recurrent arithmetic circuit with $n$ inputs, $m$ outputs, and $\ell$ auxiliary memory gates. Let $r$ be the maximum degree of all the gates.
    The gates of $C_n$ are assumed to be ordered.

    Let $C_\text{halt}$ be a the circuit that computes the halting function $\halt \colon \N_{>0} \times \R^* \to \{0,1\}$ of $C_n$.
    Let $C'_\text{halt}$ be a modified version where the output gate is replaced by $n+m+\ell+1$ copies of the gate. 
    The function $\halt'$ computed by $C'_\text{halt}$ is defined like $\halt$ but modified such that $\halt' \colon \R^*  \to \left\{\ol{0},\ol{2}\right\}$, replacing 1 with 2 and outputting 0/2-vectors of dimension $n+m+\ell+1$ instead.
    The dependence on the iteration number is implemented via an additional counting gate in $C'_\text{halt}$.
    
    Let \sgre{C_n} be the symbolically labelled graph $(G, g_{V})$, where $G=(V, E)$ consists of the undirected version of the graph of $C_n$ extended by so called dummy nodes of degree one in such a way that all non-dummy nodes have a distinct degree.
    Each node is connected to $i \times r -\text{deg}(g)$ dummy nodes where $i$ is the number of the circuit gate $g$ the node corresponds to. 
    The degree can be computed by summing over the entries for the $g$th row in the adjacency matrix.
    The undirected version of the graph of the halting circuit $C'_\text{halt}$ is encoded in the same way.
    The nodes representing halting gates of the circuit $C_n$ are called \emph{halting nodes}.
    All halting nodes are connected to the corresponding nodes representing input gates of the halting circuit $C'_\text{halt}$.
    We call the nodes corresponding to the output gates of $C'_\text{halt}$ \emph{halting output nodes}.
    
    The function $g_{V}$ is defined as follows.
    % Every dummy node connected to a node that represents a multiplication gate in the circuits has the initial feature value 1.
    % Every other dummy node has the initial feature value 0.
    Nodes corresponding to input and auxiliary memory gates of the underlying circuit are initialised with their respective symbolic values. 
    All other nodes are assigned 1 as a feature value.
    % The initial values of the other nodes are assigned such that nodes corresponding to successors of multiplication gates have value 1 the others value 0.
    % To be able to correctly assign feature values to the nodes corresponding to the successors of a gate we require the path-length normal form of the circuit.
    % This ensures there are no ambiguities in which feature value needs to be assigned to each of the nodes of $\sgre{C_n}$. 
    For $\gre{C_n}$ the symbols of the constant, auxiliary memory and input gates are replaced with their respective values.
    An example of such an encoding is given in Figure~\ref{fig:circ_to_cgnn_feature_graph_construction}.

    The encoding is computable in logspace.
    The new adjacency matrix is constructed node-wise, as the number of dummy nodes for each node is determined one after the other.
    This is possible because the dummy nodes are only connected to one node.
    The same holds for the symbolic feature values. 
\begin{figure*}[tb]
    \centering
    \begin{tikzpicture}[gate/.style={circle, minimum width=0.3cm, draw}, baseline=(mult1.base)]
        \node[]     (in1)       {in};
        \node[]     (aux1)      [right=of in1]     {aux};
        \node[gate, thick] (mult1)     [below=of $(in1)!0.5!(aux1)$]       {$\times$};
       
        \node[] (out)       [below=of mult1]     {out};

        \draw[-stealth]     (in1)    --     (mult1);
        \draw[-stealth]     (aux1)   --     (mult1);
        \draw[-stealth]     (mult1)   --     (out);
        \draw[-stealth, dashed]     (mult1)     to[bend left]  (in1);
        \draw[-stealth, dashed]     (mult1)     to [bend right] (aux1);

        \node[]     (h1)        [right=2cm of aux1]     {$i$};
        \node[]     (h2)        [right=of h1]       {$v_\times$};
        \node[rectangle, draw, minimum height=7mm]      (C)     [below=of $(h1)!0.5!(h2)$]   {$C_\text{halt}$};
        \node[]     (o1)        [below=of C]        {$\text{out}_\text{halt}$};
        
        \draw[-stealth]     (h1)    --     (C);
        \draw[-stealth]     (h2)   --     (C);
        \draw[-stealth]     (C)   --     (o1);

        \draw[dotted, -{Stealth[flex']}]  (mult1)    ..controls ($(mult1)+(2,0)$) and ($(mult1)+(2.5,3.5)$)..      (h2);
        
    \end{tikzpicture}
    \hspace{2cm}
    \centering
    \begin{tikzpicture}[gate/.style={circle, minimum width=0.9cm, draw}, dummyMult/.style={circle, minimum width=0.1cm, draw, orange}, dummy/.style={circle, minimum width=0.1cm, draw, teal}, halt/.style={circle, minimum width=0.2cm, draw, purple}, brace/.style={
             decoration={brace, mirror},decorate}, position label/.style={
               below = 10pt,
               text height = 2ex,
               text depth = 1ex
            }, baseline=(mult1.base)]
        \node[gate]     (in1)      [label={east:$(i_1)$}] {$v_\text{in}$};
        \node[gate]     (aux1)      [right=of in1, label={east:$(a_1)$}]     {$v_\text{aux}$};
        \node[gate, thick]     (mult1)     [below=of $(in1)!0.5!(aux1)$, label={south east:$(1)$}]       {$v_\times$};
        \node[gate]     (out)       [below=of mult1, label={east:$(1)$}]     {$v_\text{out}$};

        \node[dummy]    (d1)    [above=0.5cm of in1]      {};
        \node[dummy]    (d2)    [left=0.1cm of d1, label={[text=teal]west:$(1)$}]  {};
        \node[dummy]    (d3)    [right=0.1cm of d1]  {};

        \node[dummy]    (d4)    [above=0.5cm of aux1]      {};
        \node[dummy]    (d5)    [left=0.1cm of d4]  {};
        \node[dummy]    (d6)    [right=0.1cm of d4]  {};
        \node[dummy]    (d7)    [right=0.1cm of d6, label={[text=teal]east:$(1)$}]  {};

        \node[dummy]    (d8)    [left= of mult1]   {};
        \node[dummy]    (d9)    [above=0.1cm of d8]   {};
        \node[dummy]    (d10)    [above=0.1cm of d9]   {};
        \node[dummy]    (d11)    [below=0.1cm of d8]   {};
        \node[dummy]    (d12)    [below=0.1cm of d11, label={[text=teal]south:$(1)$}]   {};

        \node[teal]         (dots)  [below=0.5cm of out]    {$\dots$};
        \node[dummy]    (d13)   [left= 0.1cm of dots]       {};
        \node[dummy]    (d14)   [right = 0.1cm of dots, label={[text=teal]east:$(1)$}]    {};

        \draw[brace, teal] ($(d13.south)+ (-0.2,-0.2)$) -- node [position label, pos=0.5, align=center] {$6$\\ dummy nodes} ($(d14.south)+ (0.2,-0.2)$);

        \node[rectangle, draw, purple, minimum height=7mm]      (H)         [right=2cm of mult1]            {$C'_\text{halt}$};
        \node[purple]                       (dots2)     [below=0.5cm of H]              {$\dots$};
        \node[halt]                         (h1)        [left=0.1cm of dots2]           {};
        \node[halt]                         (h2)        [right=0.1cm of dots2, label={[text=purple]east:$(1)$}]          {};
        \draw[brace, purple] ($(h1.south)+ (-0.2,-0.2)$) -- node [position label, pos=0.5, align=center] {$n+m+\ell+1$\\ halting output nodes} ($(h2.south)+ (0.2,-0.2)$);
        
        % \node[halt]     (h3)    [right=2cm of mult1]     {};
        % \node[halt]     (h4)    [above=0.1cm of h3]    {};
        % \node[halt]     (h1)    [above=0.1cm of h4]    {};
        % \node[halt]     (h5)    [below=0.1cm of h3]    {};   
        % \node[halt]     (h6)    [below=0.1cm of h5]    {};   
        % \node[halt]     (h2)    [below=0.1cm of h5, label={[text=purple]south:$(1)$}]    {};   

        %\draw[brace, purple]    ($(h2.south east)+ (0.2,0)$) -- node [right = 3pt,
               % text height = 2ex,
               % text depth = 1ex, pos=0.5] {$n+k+m=3$} ($(h1.north east)+ (0.2,0)$);

        % \node[dummy]    (d15)   [above right=0.5cm of h1]   {};
        % \node[dummy]    (d16)   [below right=0.5cm of h2, label={south:$(0)$}]   {};
        % \node[align=center, teal]         (dots2) at  ($(d15)!0.5!(d16)$)  {\rotatebox{90}{$\cdots$}};

        \draw[]     (in1)    --     (mult1);
        \draw[]     (aux1)   --     (mult1);
        \draw[]     (mult1)   --     (out);

        \draw[teal]     (d1)    --      (in1);
        \draw[teal]     (d2)    --      (in1);
        \draw[teal]     (d3)    --      (in1);

        \draw[teal]     (d4)    --      (aux1);
        \draw[teal]     (d5)    --      (aux1);
        \draw[teal]     (d6)    --      (aux1);
        \draw[teal]     (d7)    --      (aux1);

        \draw[teal]   (d8)    --      (mult1);    
        \draw[teal]   (d9)    --      (mult1);    
        \draw[teal]   (d10)    --      (mult1);    
        \draw[teal]   (d11)    --      (mult1);    
        \draw[teal]   (d12)    --      (mult1);    

        \draw[teal]    (d13)   --  (out);    
        \draw[teal]    (d14)   --  (out);

        \draw[purple]      (H)  --  (mult1);
        \draw[purple]      (H)  --  (h1);
        \draw[purple]      (H)  --  (h2);
        % \draw[purple]   (h1)    -- (mult1);
        % \draw[purple]   (h2)    -- (mult1);
        % \draw[purple]   (h3)    -- (mult1);

        % \draw[teal]     (d15)   --  (h1);
        % \draw[teal]     (d16)   --  (h2);
        % \draw[teal]     (h3)    --  (dots2);
        % \draw[teal]     (h1)    --  (dots2);
        % \draw[teal]     (h2)    --  (dots2);
    \end{tikzpicture}

    \caption{Circuit $C$ on the left with halting circuit $C_\text{halt}$, encoded feature graph on the right. Teal dummy nodes have feature value 1, purple halting circuit representation, dummy nodes left out for halting circuit for better readability}
    \label{fig:circ_to_cgnn_feature_graph_construction}
\end{figure*}
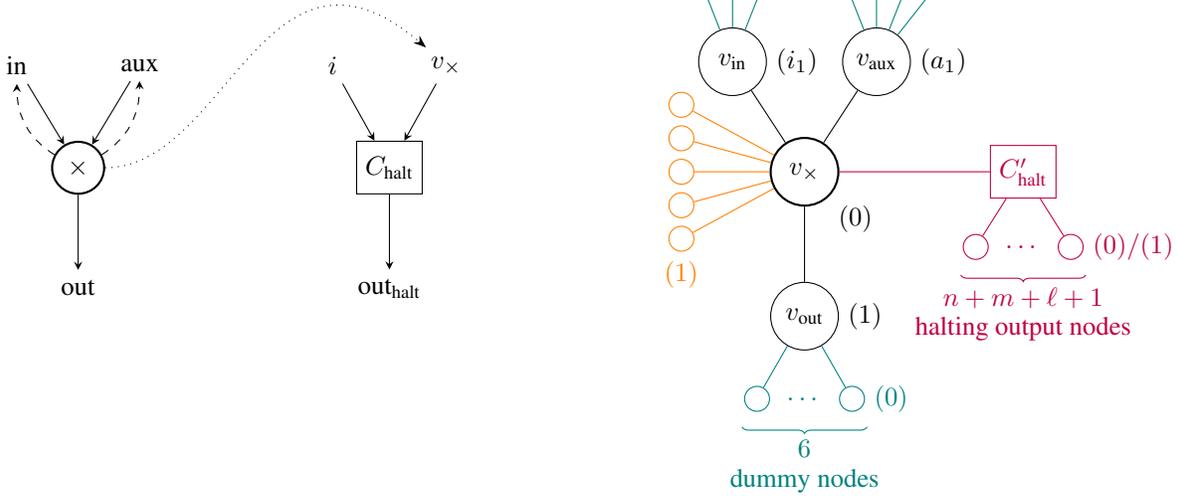

% \vivian{see if that is to superficial, leave it out completely or describe in more detail:}
% Given the adjacency matrices of the recurrent circuit the new adjacency matrix of $\sgre{C_n}$ can be constructed row wise in parallel for each node. 
% The recurrent edges are added by simply adding the matrices for $C$ and $E_rec$. 
% The connection to $C_\text{halt}$ is done similarly by adding additional zero entries for all nodes except those form the set of halting nodes $V_\text{halt}$ that are connected to the input nodes of the halting circuit.
% The same procedure is used for adding the dummy nodes. 
% Depending on the number of gates in $C$ and $C_\text{halt}$ entries are added to each row of the new matrix.
% The number of the gate the node corresponds to is used to determine which entries are 1.
% The labeling function $g_V$ is constructed by using a lookup table for the gate types and then assigning the correct feature values in parallel for each node.
% The finite set of symbols for the input, auxiliary memory and constant gates can also be encoded in $\mathbb{B}$.
% Therefore $\sgre{C_n}$ is $\acoB$ encoded.
    
The stepwise execution of the circuit is simulated by the layers of the \CGNN with outer recurrence rec-$\NN$.
There the degree of a node $v_g$ in $\gre{C_n}$ is used as an identifier to determine the exact circuit gate $g$ the node corresponds to.
The circuit families for each layer of rec-$\NN$ are defined in such a way that the circuit for input size $1$ is always the identity which ensures that the values of the dummy nodes are never changed.
Each layer $i$ of rec-$\NN$ simulates the gates of the corresponding depth of the circuit.
The nodes corresponding to the gates that are due to be executed at that depth of the circuit are identified by their degree. 
The circuit family of a layer of rec-$\NN$ is defined accordingly such that the circuits of the family where the input size is equal to the degree of a node $v_g$ execute the operation of the respective gate $g$. 
For addition and multiplication gates the neighbours are aggregated and added or respectively multiplied, the current feature value of the node is ignored. 
Nodes corresponding to input and auxiliary memory gates are treated the same as nodes corresponding to addition gates.
The same happens for nodes that correspond to gates that compute functions from the set $\mathcal{A}$ with the addition that after the summation of the neighbours the respective function from $\mathcal{A}$ is applied directly by the circuit in the node.
For all nodes that correspond to non-multiplication gates neighbouring nodes with feature value 1, like the dummy and successor nodes would influence the computations of that gate.
We therefore subtract the number of dummy and successor nodes.
% As the feature values of the dummy nodes are set to the neutral element of both operations they do not effect the resulting value.
% The same holds for the feature values of nodes that represent successor gates.
In every layer all nodes that do not correspond to predecessors of memory gates, the output gates or halting gates are reset to feature value 1.
As a \CGNN with outer recurrence can only compute tail-symmetric halting functions that are dependent on all nodes of the graph a special treatment is needed to simulate the halting function $\halt$ of the circuit $C_n$ that might only depend on a subset of the gate values.
To achieve this the modified halting circuit $C'_\text{halt}$ computing $\halt$ with multiplied outputs is also simulated by rec-$\NN$ after the simulation of one iteration of $C_n$ is finished. 
For that the same procedure of evaluating a circuit is applied to the part of the graph $\gre{C_n}$ representing the circuit $C_\text{halt}$ computing the halting function.
The values of the nodes corresponding to halting gates in the circuit are the input for the evaluation of the halting circuit.
Unless the halting gates are also predecessors of memory gates they are also reset after this step is completed.
Afterwards the global halting function $\haltGNN$ of rec-$\NN$ is applied to the set $V^{(i)}$ of all feature values $v^{(i)}$ of $\gre{C_n}$, where
\[
\haltGNN(i, X) \coloneqq \sum_{ x \in X}\chi[x=2] > (n+m+\ell),
\]
which checks whether at least $n+m+\ell+1$ nodes have value 2.
All halting output nodes, i.\,e. the nodes corresponding to the outputs of $C'_\text{halt}$ always have the same value, the dummy nodes have value 1 and all nodes but those corresponding to predecessors of memory gates (at most $n+\ell$) and output gates ($m$ gates) are reset to 1 as well. 
Therefore, for the halting function $\haltGNN$ to be satisfied at least one of the halting nodes has to have feature value 2 which means that the halting function $\halt$ of the circuit $C_n$ is satisfied.
If that is the case the \CGNN  is not iterated again and the feature values in the nodes corresponding to output gates are the exactly the output of the circuit.
Otherwise in the next layer the memory gates add all their neighbours and their own feature value together to update their value.
As their unique predecessor is among their neighbours and the feature values of all other neighbours are set to 1 and subtracted from the result they now have the value of the predecessor as their feature value.
Afterwards the predecessor, halting and output nodes are reset and the simulation of the circuit by the \CGNN starts again with the updated input values.
The procedure is described by Algorithm~\ref{alg:circ_outer_rec_def_circ_family}.
The circuits of the circuit family for layer $i$ are defined accordingly.
\end{proof}

\circToOuterRecWithActivationFnc*
\begin{proof}
    The proof follows the same concept as the proof of Theorem~\ref{thm:circ_to_outer_rec_gnn_no_activation_fnc} but changes the way gates computing the functions from $\mathcal{A}$ are simulated.
    Instead of applying the respective function directly in the nodes the function is used as a global activation function of rec-$\NN$ and therefore is applied to all nodes.
    This procedure is described by Algorithm~\ref{alg:circ_outer_rec_def_activation_function}.
    In the layer following the application of an activation function the reset procedure of Algorithm~\ref{alg:reset_alg} is applied to all nodes except the nodes corresponding to the activation function gates.
    % As we assume the functions $\sigma \in \mathcal{A}$ to be non constant there exist two values $c_0, c_1$ such that $\sigma(c_0) \neq \sigma(c_1)$. 
    % The dummy nodes can be set to $c_0/ c_1$ before the application of the activation function and reset to $0/1$ based on the feature value of the node, $\sigma(c_0)/\sigma(c_1)$. 
    % As those are fixed values Lemma~\ref{lem:chiA} can be used for the comparisons.
    Note that the simulation of an activation function gate takes two layers.
    When defining the layer numbers of operations to be executed this has to be taken into consideration during initialisation of the procedure.
    Afterwards the simulation of the circuit proceeds as previously described.
    As the circuit that is being simulated is in predecessor form the application of the global activation function and the reset of feature values of nodes does not affect the computations. 
    Feature values that need to be saved (i.\,e. those of nodes corresponding to output and halting gates and predecessors of memory gates only occur after the last global activation function has been applied.
    The simulation of the halting circuit is conducted without any changes to the proof of Theorem~\ref{thm:circ_to_outer_rec_gnn_no_activation_fnc}.
    
\end{proof}

\subsubsection{Simulation with Inner Recurrent \CGNN{}s}
\circToInnerRec*
\begin{proof}
    Let rec-$\NN = (\NN, \haltGNN)$ be a $(\text{rec}[\mathfrak{F}_s]\text{-}\mathfrak{tF}[\mathcal{A}],\{\text{id}\} )$-GNN, where $\NN(i) = (\mathcal{C}', \mathcal{A})$ and $\haltGNN(i, D) = \chi[i=1]$. 
    This means that rec-$\NN$ only consists of one layer in which a recurrent circuit family $\mathcal{C}'$ is executed in every node.
    Let $m \in \N$ be the number of output gates of the circuit $C_n$ in $\mathcal{C}$.
    Let $ \gre{C_n} = (V, E, g_v)$ be a labelled bipartite graph with $n+m$ nodes, where $\sgre{C_n}$ was the corresponding symbolically labelled graph.
    There are $m$ nodes that later represent the $m$ outputs of $C_n$ and $n$ nodes that represent the inputs.
    The function $g_V$ assigns the values $1, \dots, m$ to the $m$ nodes.
    The $n$ other nodes are each assigned one of the $n$ input values of $C_n$. 
    Each of the $m$ nodes is connected to all $n$ nodes. 
    The graph $ \gre{C_n}$ is depicted in Figure~\ref{fig:graph_circ_to_inner_rec_gnn}.
    The constructions can be done by a logspace transducer. 
    The $n$ nodes for the input are copied from the encoding of the recurrent circuit.
    Using a counter the $m$ additional nodes are added and initialised with the values from $[m]$. 
    The new adjacency matrix is constructed row-wise. 
    The entries are predefined to be $(\underbrace{0 \dots 0}_n, \underbrace{1 \dots 1}_m)$ for the $n$ nodes representing inputs and $(\underbrace{1 \dots 1}_m, \underbrace{0 \dots 0}_n)$ for the $m$ nodes representing outputs to achieve the complete bipartite graph.
   
    Let $C_n^m$ be a tail-symmetric version of the symmetric circuit $C_n$ with only one output gate and $n+1$ input gates. 
    The circuit $C_n$ is extended with a switching gadget that decides based on the new head of the inputs which of the original $m$ outputs is forwarded to the new output gate (see proof of Theorem~\ref{thm:recCGNN_in_circ}).  
    Let $\mathcal{C}' = (C_n^m)_{n\in \N}$.
    
    In the single layer of rec-$\NN$, that gets $ \gre{C_n}$ as an input, the circuit $C_n^m$ is then executed on the feature value of the $m$ nodes (i.\,e. the number signifying which output is to be computed) and the feature values of the neighbours (i.\,e. the $n$ inputs of $C_n$. 
    After the execution the feature values of the $m$ nodes equal the outputs of $C_n$.

   \begin{figure}
       \centering
        \begin{tikzpicture}[gate/.style={circle, minimum width=0.9cm, draw}]
            \node[gate]     (i1)    [label={$(i_1)$}]                        {$v_{\text{i}_1}$};
            \node[gate]     (i2)    [right=0.5 of i1, label={$(i_2)$}]       {$v_{\text{i}_2}$};
            \node[]     (dots)  [right=0.5 of i2]     {$\dots$};
            \node[gate]     (in)    [right=0.5 of dots, label={$(i_n)$}]       {$v_{\text{i}_n}$};
            
            \node[gate]     (o1)    [below=of i1, label={below:$(1)$}]               {$v_{\text{o}_1}$};
            \node[gate]     (o2)    [right=0.5 of o1, label={below:$(2)$}]       {$v_{\text{o}_2}$};
            \node[]     (dots2) [right=0.5 of o2]     {$\dots$};
            \node[gate]     (om)    [right=0.5 of dots2, label={below:$(m)$}]       {$v_{\text{o}_m}$};

            \draw[]     (i1)    --  (o1);
            \draw[]     (i2)    --  (o1);
            \draw[]     (in)    --  (o1);
            \draw[]     (i1)    --  (o2);
            \draw[]     (i2)    --  (o2);
            \draw[]     (in)    --  (o2);
            \draw[]     (i1)    --  (om);
            \draw[]     (i2)    --  (om);
            \draw[]     (in)    --  (om);
        \end{tikzpicture}
        \caption{The bipartite labelled graph $ \gre{C_n}$. The nodes $v_{i_1}, \dots, v_{i_n}$ represent the inputs, $v_{o_1}, \dots, v_{o_m}$ represent the outputs, $\ol{x}=i_1, \dots, i_n$ are the inputs to the circuit $C_n$.}
        \label{fig:graph_circ_to_inner_rec_gnn}
   \end{figure}
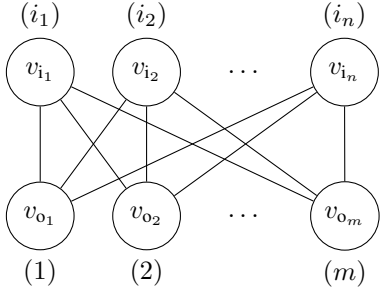
\end{proof}

\CircToInnerRecWithActivationFnc*

\begin{proof}
% \vivian{add this in: circuit with a non-constant halting function dependent on the memory gates, where $\exp$ is the exponential function and at least one gate labelled with the exponential function is the predecessor of both a memory and an output gate.}
    Let $C_n$ be the recurrent circuit depicted in Figure~\ref{fig:exp_circ} with some non-constant halting function that is dependent on the gates labelled with the exponential function (e.\,g. $\forall v_e$ val$(v_e) > 600$).
    Those gates are the predecessors of memory gates.
    Arithmetic circuits without additional functions can only compute polynomials.
    The gates of the circuits $C_n$ that compute the function $\exp$ can therefore only be simulated by activation functions in a $(\text{rec}[\mathfrak{F}_s]\text{-}\mathfrak{tF}, \{\exp\})$-GNN. 
    A \CGNN with inner recurrence has a fixed number of layers and is therefore only able to executed a fixed number of activation functions.
    Activation functions can only be applied after the computations of the recurrent circuit of that respective layer are done.
    Meaning that for an arbitrary node of the input graph $\gre{C_n}$ a \CGNN with inner recurrence of depth $d$ can at most compute the feature value $(e\uparrow \uparrow d)^{i}$ for some $i \in \R$.
    \def\rddots#1{\cdot^{\cdot^{\cdot^{#1}}}}
    Here the notation $e \uparrow \uparrow d$ stands for $\underbrace{e^{\rddots e}}_d$.
    As the functions computed by $\mathcal{C}$ are from a recurrent function class its circuits $C_n$ are executed iteratively until a halting condition is met.
    As the halting condition itself is from a family of functions in a circuit function class the number of iterations is dependent on the input $\ol{x}$ and therefore cannot be determined in the symbolic construction of the \CGNN (see Section~\ref{sec:encoding}). 
    %\vivian{the BSS halting problem is undecidable over $\R$, recurrent circuit is basically a BSS machine, when underlying circuit and halting circuit are connected directly. Can't decide whether that machine halts on a given input, therefore the number of steps=iterations machine does is also undecidable, only semi-decidable, which means they can only be computed by actually running the machine. For that we need the actual inputs, which we don't have, as it is only a symbolic encoding from circuit to graph}
    Let $d'$ the number of iterations of $C_n$ given an arbitrary input $(i_1, \dots, i_n) \in \R^n$. 
    Then $(e \uparrow \uparrow d')^{\prod_{j \in [n]} i_j}$ is the output of $C_n$.
    It holds that $e \uparrow \uparrow d' \neq e \uparrow \uparrow d$ for all $d' \neq d$.
    Therefore there exists no inner recurrent \CGNN without additional functions in its internal circuits and without any outer recurrence that has the output of $f_{C_n}(\ol{x})$ among its feature values.

    Here a circuit with constant depth was chosen. 
    Notice that the proof would work for any symmetric circuit of depth $\log^i(n)$ for any $i$ that has the circuit $C_n$ as a subcircuit and therefore for any circuit function class $\mathfrak{F}$.
\begin{figure}
    \centering
    \begin{tikzpicture}[gate/.style={circle, minimum width=0.7cm, draw}]
            \node[circle, minimum width=0.9cm]     (i1)                            {$i_1$};
            \node[circle, minimum width=0.9cm]     (dots)  [right=0.5 of i1]     {$\dots$};
            \node[circle, minimum width=0.9cm]     (in)    [right=0.5 of dots]       {$i_n$};

            \node[gate, thick]     (exp)   [below=0.5 of i1]       {$e$};
            \node[gate, thick]     (exp2)    [below=0.5 of in]         {$e$};
            \node[gate]     (mult)      [below=0.5 of $(exp)!.5!(exp2)$]    {$\times$};
            \node[]     (out)       [below=0.5 of mult]     {out};

            \draw[-stealth]     (i1)    to[bend left]      (exp);
            \draw[-stealth, dashed]     (exp)    to[bend left]      (i1);
            \draw[-stealth]     (exp)    --      (mult);
            \draw[-stealth]     (in)    to[bend left]      (exp2);
            \draw[-stealth]     (mult)    --      (out);       
            \draw[-stealth]     (exp2)    --      (mult);       
            \draw[-stealth, dashed]     (exp2)    to[bend left]       (in);       
    \end{tikzpicture}
    \caption{$C_n$, the gate $e$ computes the exponential function and is the halting gate}
    \label{fig:exp_circ}
\end{figure}
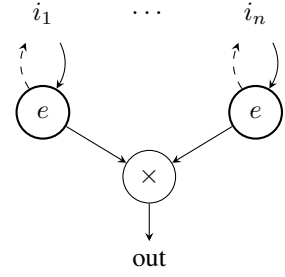
\end{proof}

\begin{algorithm}[htb]
    \caption{Algorithm for activation function layer handling}
    \label{alg:circ_outer_rec_def_activation_function}
    \begin{algorithmic}
        \STATE {\bfseries Input:} $v_g$, $v_g^{(i-1)}$, $U_g = \{u^{(i-1)} \mid u \in \mathcal{N}_G(v_g)\}$\
        \STATE $A = \{ a \mid a \text{ circuit depth with activation functions}\}$\; \COMMENT{The numbers of layers that correspond to activation functions} 
        \STATE $X_\sigma = \{ deg(v_g) \mid g \text{ is $\sigma$ gate} \}$\;
        \IF[$i$ is activation function layer]{$i \in A$}
            \IF[If $v_g$ corresponds to an activation function gate]{$deg(v_g)\in X_{\sigma}$}
                \STATE $v_g^{(i)} = \sum U_g$\;
                \STATE Apply activation function $\sigma$ globally.\;
            \ENDIF
        \ELSIF[$i-1$ was activation function layer]{$i-1 \in A$}
            \IF[If $v_g$ corresponds to activation gate]{$deg(v_g)\in X_\sigma$}
                \STATE $v_g^{(i)}=v_g^{(i-1)}$\;
            \ELSE
                \STATE $v_g^{(i)}= \text{Reset}(v_g)$\;
            \ENDIF
        \ENDIF
        
        % \IF[If $v_g$ is dummy node]{$deg(v_g)=1$}
        %     \STATE $v_g^{(i)} = v_g^{(i-1)}$
        % \ENDIF
        % \IF[If a node needs to be evaluated in this layer] {$deg(v_g) \in Y^{(i)}$} 
        %     \IF[If $v_g$ corresponds to an sign gate] {$deg(v_g)\in X_\text{sign}$}
        %     \STATE $v_g^{(i)} = \text{sign}(\sum U_g)$\; 
        %     \ELSIF[If $v_g$ corresponds to a multiplication gate]{$deg(v_g)\in X_{\times}$}
        %      \STATE $v_g^{(i)} = \prod U_g$\; 
        %     \ELSE[If $v_g$ corresponds to any other gate type]
        %     \STATE $v_g^{(i)} = \sum U_g$\;
        %     \ENDIF
        % \ELSE[If $v_g$ does not correspond to a gate at depth $i$]
        %     \IF[If $v_g$ corresponds predecessor of a memory gate or output gate] {$deg(v_g)\in X_\text{pred}/X_\text{out}$}
        %     \STATE $v_g^{(i)} = v_g^{(i-1)} $\; 
        %     \ELSIF[If $v_g$ corresponds to halting gate] {$deg(v_g)\in X_\text{halt}$}
        %             \IF[If still in simulation of underlying circuit] {$i < \depth(C)$}
        %             \STATE $v_g^{(i)} = v_g^{(i-1)} $\; 
        %             \ELSE[Halting circuit is being simulated]
        %             \STATE $v_g^{(i)} = \text{Reset}(v_g) $\;
        %             \ENDIF 
        %     \ELSE[all other gates where values don't need to be saved]
        %         \STATE $v_g^{(i)} = \text{Reset}(v_g) $\;
        %     \ENDIF
        % \ENDIF
    \end{algorithmic}
    \end{algorithm}

        \begin{algorithm}[H]
    \caption{Algorithm for the circuit family in layer $i$ of the \CGNN}
    \label{alg:circ_outer_rec_def_circ_family}
    \begin{algorithmic}
        \STATE {\bfseries Input:} $v_g$, $v_g^{(i-1)}$, $U_g = \{u^{(i-1)} \mid u \in \mathcal{N}_G(v_g)\}$, $q_g $ \COMMENT{number of dummy and successor nodes of $g$}\; 
        \STATE $Y^{(i)} = \{ deg(v_g) \mid \depth(g) =i\}$\; \COMMENT{The set of degrees of nodes corresponding to gates that need to be evaluated in this layer} 
        \STATE $X_+/X_{\times}/X_\sigma/X_\text{sign} = \{ deg(v_g) \mid g \text{ is +/$\times$/$\sigma$/sign gate} \}$\;
        \STATE $X_\text{halt}/X_\text{out} = \{ deg(v_g) \mid g \text{ is halting/ output gate} \}$\;
        %\STATE $X_\text{in}/X_\text{aux} = \{ deg(v_g) \mid g \text{ is input/auxiliary memory gate} \}$\;
        \STATE $X_\text{pred} = \{ deg(v_g) \mid g \text{ is predecessor of memory gate} \}$\;
            \IF[If $v_g$ is dummy node]{$deg(v_g)=1$}
                \STATE $v_g^{(i)} = v_g^{(i-1)}$
            \ENDIF
            \IF[If a node needs to be evaluated in this layer] {$deg(v_g) \in Y^{(i)}$} 
                \IF[If $v_g$ corresponds to an sign gate] {$deg(v_g)\in X_\text{sign}$}
                \STATE $v_g^{(i)} = \text{sign}(\sum U_g - q_g)$\; 
                \ELSIF[If $v_g$ corresponds to a multiplication gate]{$deg(v_g)\in X_{\times}$}
                 \STATE $v_g^{(i)} = \prod U_g$\; 
                \ELSIF[If $v_g$ corresponds to an activation function gate]{$deg(v_g)\in X_{\sigma}$}
                \STATE $v_g^{(i)} = \sigma(\sum U_g -q_g)$\;
                \ELSE[If $v_g$ corresponds to any other gate type]
                \STATE $v_g^{(i)} = \sum U_g -q_g$\;
                \ENDIF
                \ELSE[If $v_g$ does not correspond to a gate at depth $i$]
                \IF[If $v_g$ corresponds predecessor of a memory gate or output gate] {$deg(v_g)\in X_\text{pred}/X_\text{out}$}
                \STATE $v_g^{(i)} = v_g^{(i-1)} $\; 
        %     \algstore{alg2}
        % \end{algorithmic}
        % \end{algorithm}
        % \begin{algorithm}[H]
        % \begin{algorithmic}

        %         \algrestore{alg2}
                \ELSIF[If $v_g$ corresponds to halting gate] {$deg(v_g)\in X_\text{halt}$}
                        \IF[If still in simulation of underlying circuit] {$i < \depth(C)$}
                        \STATE $v_g^{(i)} = v_g^{(i-1)} $\; 
                        \ELSE[Halting circuit is being simulated]
                        \STATE $v_g^{(i)} = \textbf{Reset}(v_g) $\;
                        \ENDIF 
                \ELSE[all other gates where values do not need to be saved]
                    \STATE $v_g^{(i)} = \textbf{Reset}(v_g) $\;
                \ENDIF
            \ENDIF
    \end{algorithmic}
    \end{algorithm}
     \begin{algorithm}[H]
    \caption{Reset procedure}
    \label{alg:reset_alg}
    \begin{algorithmic}
        \STATE {\bfseries Input:} $v_g$, $c_g$ \COMMENT{constant value of gate $g$, empty if $g$ is not a constant gate}\
        \STATE $X_\text{const} = \{ deg(v_g) \mid g \text{ is constant gate} \}$\;
            \IF[If $v_g$ corresponds to constant gate] {$deg(v_g) \in X_\text{const}$} 
                \STATE $v_g^{(i)}=c_g$;
            \ELSE
                \STATE $v_g^{(i)}=1$;
            \ENDIF
    \end{algorithmic}
    \end{algorithm}

\end{document}